\documentclass[12pt]{article}

\usepackage[english]{babel}
\usepackage[utf8]{inputenc}
\usepackage{amsmath}
\usepackage{graphicx}
\usepackage{subfigure}
\usepackage{amssymb}
\usepackage{amsthm}
\usepackage{bm}
\usepackage{bbm}
\usepackage{tikz-cd}
\usepackage{mathrsfs}
\usepackage[colorinlistoftodos]{todonotes}
\usepackage{enumitem}
\usepackage{yfonts}
\usepackage{ dsfont }
\usepackage{geometry}
\usepackage{setspace}
\usepackage{tikz}
\usepackage{pgfplots}
\usepackage{authblk}
\pgfplotsset{compat=newest}
\usepgfplotslibrary{fillbetween}
\usetikzlibrary{positioning}
\usetikzlibrary{decorations.pathreplacing}
\usetikzlibrary{decorations,arrows}
\usetikzlibrary{decorations.markings}
\usetikzlibrary{patterns}
\usetikzlibrary{quotes,angles}

\numberwithin{equation}{section}
\title{Bulk-edge correspondence in finite photonic structure}

\author{Jiayu Qiu\thanks{Department of Mathematics, 
 HKUST,  Clear Water Bay, Kowloon, Hong Kong SAR, China, jqiuaj@connect.ust.hk.}\,\, Hai Zhang\thanks{Department of Mathematics, HKUST, Clear WaterbBay, Kowloon, Hong Kong SAR, China, haizhang@ust.hk. HZ was partially supported by Hong Kong RGC grant GRF 16307024 and NSFC grant 12371425.}}

\date{\today}

\newtheorem{theorem}{Theorem}[section]
\newtheorem{lemma}[theorem]{Lemma}

\newtheorem{definition}[theorem]{Definition}

\newtheorem{hypothesis}[theorem]{Hypothesis}
\newtheorem{corollary}[theorem]{Corollary}

\newtheorem{remark}[theorem]{Remark}
\newtheorem{proposition}[theorem]{Proposition}

\allowdisplaybreaks[4]

\begin{document}

\maketitle
\begin{abstract}
In this work, we establish the bulk-edge correspondence principle for finite two-dimensional photonic structures. Specifically, we focus on the divergence-form operator with periodic coefficients and prove the equality between the well-known gap Chern number (the bulk invariant) and an edge index defined via a trace formula for the operator restricted to a finite domain with Dirichlet boundary conditions. We demonstrate that the edge index characterizes the circulation of electromagnetic energy along the system's boundary, and the BEC principle is a consequence of energy conservation. The proof leverages Green function techniques and can be extended to other systems. These results provide a rigorous theoretical foundation for designing robust topological photonic devices with finite geometries, complementing recent advances in discrete models. 
\end{abstract}

\section{Introduction}
\subsection{Background}

Since the discovery of the quantum Hall effect (QHE) and its topological origin \cite{Klitzing80qhe,vonKlitzing2020forty_years,stone1992quantum}, the topological phase of matter has gained great interest within both the physics and mathematics communities. The remarkable properties of topological materials are rooted in the fundamental concept of the bulk-edge correspondence (BEC) principle. This principle asserts that differences in the topological phases of two materials separated by an interface give rise to transportation channels at the interface \cite{qizhang11topo_insulator,bernivig13topo_insulator}. Due to its topological nature, edge transport is robust against perturbations such as medium deformations or impurities, enabling a wide range of applications. Motivated by the developments of topological phases of matter in condensed-matter systems, it has been argued and demonstrated that the topological properties of materials are ubiquitous across various periodic systems, such as photonic crystals \cite{haldane08realization,wang2009observation,ozawa19topological_photonics}. Enabled by the BEC principle, the emergence of robust edge transport of light has driven advancements in the design of high-quality optical devices and spurred significant interest in topological photonics.

The mathematical research on BEC begins with the interpretation of the integer quantum Hall effect \cite{Hatsugai93chern}. Since then, the principle of BEC has been established across various models, primarily focusing on electronic systems described by the (magnetic or non-magnetic) Schrödinger operator \cite{Kellendonk02landau+ktheory,Kellendonk2004landau+ktheory,Kellendonk04landau+functional,taarabt2014landau+Ktheory,cornean2021landau+functional,shapiro2022shrodinger+functional,drouot2021microlocal}, Dirac Hamiltonians \cite{bal2019dirac+functional,bal2023dirac+microlocal}, and tight-binding models \cite{graf2018shortrange+transfer,avila2013shortrange+transfer,ludewig2020shortrange+coarse,graf2013shortrange+scattering,elgart2005shortrange+functional,shapiro2022shrodinger+functional}. The techniques used in these studies include K-theory \cite{bourne2017ktheory,kubota2017ktheory,prodan2016ktheory,Kellendonk02landau+ktheory,Kellendonk2004landau+ktheory,taarabt2014landau+Ktheory}, functional analysis \cite{bal2019dirac+functional,cornean2021landau+functional,Kellendonk04landau+functional,elgart2005shortrange+functional,shapiro2022shrodinger+functional}, transfer matrix and scattering theory \cite{graf2013shortrange+scattering,lin2022transfer,thiang2023transfer,graf2018shortrange+transfer,avila2013shortrange+transfer}, Toeplitz theory \cite{ammari2024toeplitz_1,ammari2024toeplitz_2,braverman2018spectralflowfamilytoeplitz}, coarse geometry \cite{ludewig2020shortrange+coarse}, and microlocal and semiclassical analysis \cite{drouot2021microlocal,bal2023dirac+microlocal}. While extensive research has validated BEC in infinite domains, such as half-plane models with edges or interface models consisting of two separated half-planes (see references above), the mathematical framework for finite geometries remains underexplored. Recent works \cite{bal2023edge_curved,ludewig2020shortrange+coarse,drouot2024bec_curvedinterfaces} have addressed more general domains in which the edge or interface may not be straight, while still requiring the infinite size of domains. These developments, along with the rapid growth of topological photonics and the apparent fact that all realistic experiments are conducted in finite-size domains, motivate the need for a rigorous mathematical framework for BEC in finite photonic structures.

In this work, we rigorously prove the bulk-edge correspondence principle for finite two-dimensional photonic structures. Specifically, for a divergence-form operator with periodic coefficients, we establish the equality between the well-known gap Chern number and an edge index defined by a trace formula (Definition \ref{def_edge_index}). This edge index is evaluated for the operator restricted to a finite, simply connected, and bounded domain with Dirichlet boundary conditions. The newly defined edge index is central to our proof of BEC. Unlike the edge index in previous studies, which primarily measures wave propagation along a specific direction, the index defined in \eqref{eq_edge_index} describes \textit{the circulation of waves within the bounded structure}. This shift in perspective is driven by the physical intuition: in a finite domain, unidirectional propagation along the edge is canceled out at opposite faces, necessitating a new observable to describe edge transport. Notably, with this newly defined edge index, the BEC proved in this paper has an intuitive physical interpretation: since 1) our edge index measures the circulation of electromagnetic energy along an impenetrable boundary and 2) the Chern number measures the energy current in the bulk excited by an external source, the BEC principle arises naturally from energy conservation; we elaborate on this physical explanation detailedly in the Appendix. Our proof of BEC relies on the estimates of Green functions, especially the effect of the boundary as the domain size becomes large, as outlined in Section \ref{sec_outline_proof}. We expect the framework developed here to extend to study BEC in various settings, including finite quantum Hall systems or systems with disorder. These perspectives are discussed in Section \ref{sec_relation_perspec}.

\subsection{Main results}
We consider the following divergence-form operator in $L^2(\mathbf{R}^2)$ with $\mathbf{Z}^2-$periodic coefficients, which models the propagation of time-harmonic TE polarized electromagnetic waves in a 2D square lattice:
\begin{equation*}
\mathcal{L}=-\nabla\cdot A\nabla \quad\text{with}\quad
A(x+e)=A(x)\quad \forall e\in \mathbf{Z}^2.
\end{equation*}
We assume $A\in C^3(\mathbf{R}^2,\mathbf{C}^{2\times 2})$ is Hermitian and positive definite. Then $\mathcal{L}$ admits nonnegative spectrum, i.e. $\sigma(\mathcal{L})\subset [0,\infty)$. Suppose $\mathcal{L}$ has a band gap $\Delta:=(\lambda_{\text{low}},\lambda_{\text{upp}})$ in its spectrum, i.e.
\begin{equation*}
\overline{\Delta}\cap \sigma(\mathcal{L})=\emptyset.
\end{equation*}
Then the gap Chern number is well-defined and integer-valued. It is given as the integral of Berry curvature over the torus $\mathbf{T}^2=\mathbf{R}^2/\mathbf{Z}^2$ \cite{tong2016lecturesquantumhalleffect,ozawa19topological_photonics},
\begin{equation} \label{eq_chern_num}
\mathcal{C}_{\Delta}=\frac{-1}{4\pi}\int_{\mathbf{T}^2}d\kappa\sum_{n\in F}
\text{Im}\big(\partial_{\kappa_1}u_{n}(\cdot;\kappa),\partial_{\kappa_2}u_{n}(\cdot;\kappa)\big)
\in\mathbf{Z}.
\end{equation}
Here $F:=\{n:\,\lambda_n(\kappa)<\lambda_{\text{low}}\}$ is referred to as the index set of filled bands and $E:=\{n:\,\lambda_n(\kappa)>\lambda_{\text{upp}}\}=\mathbf{N}\backslash F$ the unfilled bands. The inner product is taken in $L^2(Y)$ with $Y=(0,1)^2$ being the unit cell. $(\lambda_n(\kappa), v_{n}(x;\kappa))$ are the Bloch eigenpairs and $u_{n}(x;\kappa)=e^{-i\kappa\cdot x}v_{n}(x;\kappa)$ denotes the periodic part of Bloch eigenfunction. 

Let $\Omega\subset \mathbf{R}^2$ be a simply connected and bounded domain with a $C^{\infty}$ boundary that contains the origin. Define the scaled expansion of $\Omega$ by $\Omega_L:=L\cdot \Omega$ ($L>0$). We consider the restriction of $\mathcal{L}$ to $\Omega_L$ with Dirichlet boundary conditions, i.e.
\begin{equation*}
\mathcal{L}_{\Omega_L}:=\mathcal{L}\big|_{H_0^1(\Omega_L)}.
\end{equation*}
Let $g\in C^{\infty}_{c}(\mathbf{R})$ satisfy that
\begin{equation*}
g(x)=1\quad \text{if }x<\lambda_{low},\quad
g(x)=0\quad \text{if }x>\lambda_{upp},\quad
\text{and}\quad \text{supp }g^{\prime} \subset \Delta.
\end{equation*}
Here $g'$ denote the derivative of $g$. 
We define the edge index of $\mathcal{L}_{\Omega_L}$ as:
\begin{definition}[edge index] \label{def_edge_index}
    \begin{equation} \label{eq_edge_index}
EI_{L}(\Delta):=\text{Tr}_{\Omega_L}
\big(i(x_1\mathcal{V}_2-x_2\mathcal{V}_1)g^{\prime}(\mathcal{L}_{\Omega_L})\big),
\end{equation}
where $x_i$ is the position operator and $$\mathcal{V}_i=[\mathcal{L},x_i]=\mathcal{L}x_i-x_i\mathcal{L}=-e_{i}\cdot \nabla A-\nabla\cdot Ae_i.$$
\end{definition}

The trace in \eqref{eq_edge_index} is well-defined by noting the following facts. First, $g^{\prime}(\mathcal{L}_{\Omega_L})$ is \textit{a finite sum} of rank-one projections, i.e. $\sum_{k=1}^{m}(\cdot,u_{k}(x))u_k(x)$, with each $u_k$ being an eigenfunction of $\mathcal{L}_{\Omega_L}$ (recall that $\mathcal{L}_{\Omega_L}$ has a discrete spectrum and $g^{\prime}$ has compact support). On the other hand, each eigenfunction $u_k$ of the elliptic operator $\mathcal{L}_{\Omega_L}$ admits $H^2$ regularity, i.e. $u_k\in H^2(\Omega_L)$. As a consequence, $\sum_{k=1}^{m}(\cdot,u_{k}(x))(x_1\mathcal{V}_2-x_2\mathcal{V}_1)u_k(x)$ is finite-rank and bounded on $\mathcal{B}(L^2(\Omega_L))$, and hence is trace-class.

The main result of this paper is the following:
\begin{theorem} \label{thm_bec}
Assuming \eqref{eq_GL_singularities_1}-\eqref{eq_G_sharp_singularities_2} on the singularities of Green functions, we have
\begin{equation} \label{eq_bec}
\lim_{L\to\infty}\frac{1}{|\Omega_L|}EI_{L}(\Delta)
=\mathcal{C}_{\Delta}.
\end{equation}
\end{theorem}

In the Appendix, we provide a physical interpretation of Theorem \ref{thm_bec} as a natural consequence of energy conservation. Using linear response theory, we demonstrate that the gap Chern number $\mathcal{C}_{\Delta}$ characterizes the expectation of photonic energy transport induced by external perturbations, analogous to the Hall conductance in the quantum Hall effect (QHE). On the other hand, as suggested by its definition, the edge index \eqref{eq_edge_index} describes the energy circulation near the boundary $\partial \Omega_L$. By energy conservation, these two quantities must be equal in a lossless system, as indicated by the equality in \eqref{eq_bec}.

The emergence of edge spectrum follows directly from Theorem \ref{thm_bec}:
\begin{corollary} \label{corol_edge_spectrum}
Let $\Delta_n$, $1\leq n \leq N$, be a finite collection of pairwise disjoint open intervals contained in $\Delta$. If $\mathcal{C}_{\Delta}\neq 0$, then for sufficiently large $L>0$, we have $\sigma(\mathcal{L}_{\Omega_L})\cap \Delta_n \neq \emptyset$ for all $1\leq n \leq N$.  
\end{corollary}
\begin{proof}
For each $n$, we have 
$$
\lim_{L\to\infty}\frac{1}{|\Omega_L|}EI_{L}(\Delta_n)
=\mathcal{C}_{\Delta_n} = \mathcal{C}_{\Delta} \neq 0.
$$
Therefore, for $L$ sufficiently large, $EI_L(\Delta_n) \neq 0$ for all $1\leq n \leq N$. 
It follows that $\sigma(\mathcal{L}_{\Omega_L})\cap \Delta_n \neq \emptyset$ for all $1\leq n \leq N$.
\end{proof}

\begin{remark}
The Chern number vanishes when the operator $\mathcal{L}$ is time-reversal symmetric, meaning that the coefficient matrix $A$ is real-valued. Therefore, introducing complex coefficients in $A$ is essential to achieve nontrivial topology. In experiments, complex coefficients are typically introduced using magneto-optic materials \cite{haldane08realization} or bi-anisotropic media \cite{khanikaev2013photonic}; see also \cite{zhu2019elliptic,ozawa19topological_photonics}.
\end{remark}

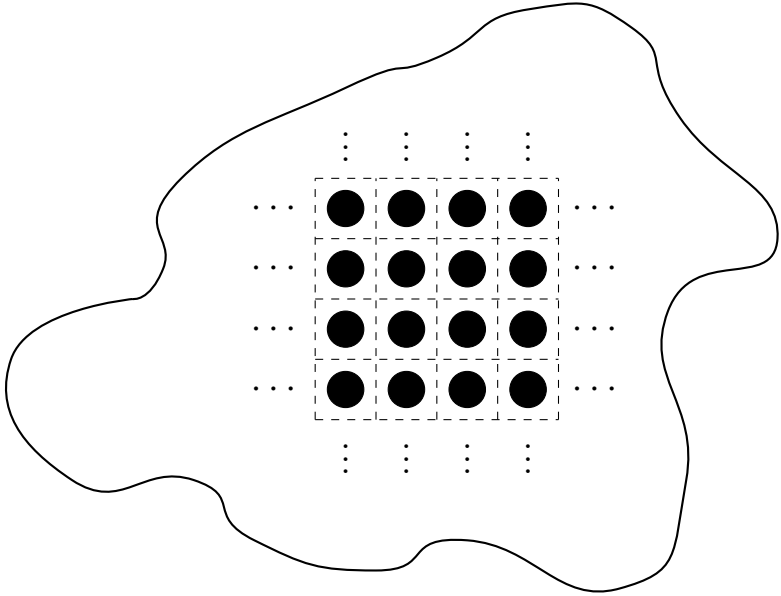
\begin{figure}
\begin{center}
\begin{tikzpicture}[scale=0.8]
\draw[dashed] ({-2},{0})--({2},{0});
\draw[dashed] ({-2},{1})--({2},{1});
\draw[dashed] ({-2},{2})--({2},{2});
\draw[dashed] ({-2},{-1})--({2},{-1});
\draw[dashed] ({-2},{-2})--({2},{-2});
\draw[dashed] ({0},{-2})--({0},{2});
\draw[dashed] ({1},{-2})--({1},{2});
\draw[dashed] ({2},{-2})--({2},{2});
\draw[dashed] ({-1},{-2})--({-1},{2});
\draw[dashed] ({-2},{-2})--({-2},{2});
\draw[fill=black,opacity=1] ({1/2},{1/2}) ellipse(0.3 and 0.3);
\draw[fill=black,opacity=1] ({1+1/2},{1/2}) ellipse(0.3 and 0.3);
\draw[fill=black,opacity=1] ({-1+1/2},{1/2}) ellipse(0.3 and 0.3);
\draw[fill=black,opacity=1] ({-2+1/2},{1/2}) ellipse(0.3 and 0.3);
\draw[fill=black,opacity=1] ({1/2},{1/2+1}) ellipse(0.3 and 0.3);
\draw[fill=black,opacity=1] ({1+1/2},{1/2+1}) ellipse(0.3 and 0.3);
\draw[fill=black,opacity=1] ({-1+1/2},{1/2+1}) ellipse(0.3 and 0.3);
\draw[fill=black,opacity=1] ({-2+1/2},{1/2+1}) ellipse(0.3 and 0.3);
\draw[fill=black,opacity=1] ({1/2},{1/2-1}) ellipse(0.3 and 0.3);
\draw[fill=black,opacity=1] ({1+1/2},{1/2-1}) ellipse(0.3 and 0.3);
\draw[fill=black,opacity=1] ({-1+1/2},{1/2-1}) ellipse(0.3 and 0.3);
\draw[fill=black,opacity=1] ({-2+1/2},{1/2-1}) ellipse(0.3 and 0.3);
\draw[fill=black,opacity=1] ({1/2},{1/2-2}) ellipse(0.3 and 0.3);
\draw[fill=black,opacity=1] ({1+1/2},{1/2-2}) ellipse(0.3 and 0.3);
\draw[fill=black,opacity=1] ({-1+1/2},{1/2-2}) ellipse(0.3 and 0.3);
\draw[fill=black,opacity=1] ({-2+1/2},{1/2-2}) ellipse(0.3 and 0.3);

\node[left,scale=1.2] at ({-2},{1/2}) {$\cdots$};
\node[left,scale=1.2] at ({-2},{1/2+1}) {$\cdots$};
\node[left,scale=1.2] at ({-2},{1/2-1}) {$\cdots$};
\node[left,scale=1.2] at ({-2},{1/2-2}) {$\cdots$};
\node[right,scale=1.2] at ({2},{1/2}) {$\cdots$};
\node[right,scale=1.2] at ({2},{1/2+1}) {$\cdots$};
\node[right,scale=1.2] at ({2},{1/2-1}) {$\cdots$};
\node[right,scale=1.2] at ({2},{1/2-2}) {$\cdots$};
\node[above,scale=1.2] at ({1/2},{2+1/16}) {$\vdots$};
\node[above,scale=1.2] at ({1/2+1},{2+1/16}) {$\vdots$};
\node[above,scale=1.2] at ({1/2-1},{2+1/16}) {$\vdots$};
\node[above,scale=1.2] at ({1/2-2},{2+1/16}) {$\vdots$};
\node[below,scale=1.2] at ({1/2},{-2+1/8}) {$\vdots$};
\node[below,scale=1.2] at ({1/2+1},{-2+1/8}) {$\vdots$};
\node[below,scale=1.2] at ({1/2-1},{-2+1/8}) {$\vdots$};
\node[below,scale=1.2] at ({1/2-2},{-2+1/8}) {$\vdots$};

\draw [thick] plot [smooth,samples=400, tension=1] coordinates {(-5,0) (-7,-1) (-6,-3) (-4,-3) (-3,-4) (-1,-4.5) (0.5,-4) (3,-4.8) (4.1,-3) (3.8,-0.2) (5.6,1) (4,3) (3.2,4.5) (1.5,4.8) (0,4) (-1.5,3.5)  (-4.2,2) (-4.5,0.5) (-5,0)};
\end{tikzpicture}
\caption{A finite sample consists of dielectric rods with PEC boundary.}
\label{fig_finite_sample}
\end{center}
\end{figure}

\subsection{Outline of the Proof of Theorem \ref{thm_bec}} \label{sec_outline_proof}
We outline the main ideas behind the proof of Theorem \ref{thm_bec}, which builds on the approach of Mario \cite{mario19proof}.  The key of the proof is that both the gap Chern number and the edge index are expressed in terms of appropriate Green functions.

Let $\tilde{g}\in C_{0}^{\infty}(\mathbf{C})$ be the almost analytic extension of $g$ (see Section 3.1 for an introduction). Throughout this paper, we denote the integral kernel of resolvents $R(z,\mathcal{L})=(\mathcal{L}-z)^{-1}$ and $R(z,\mathcal{L}_{B_{L}})=(\mathcal{L}_{\Omega_L}-z)^{-1}$ by $G_{\sharp}(x,x^{\prime};z)$ and $G_{L}(x,x^{\prime};z)$, respectively. These represent the Green functions for the periodic and finite structure, respectively.  A key observation is that the gap Chern number $\mathcal{C}_{\Delta}$ can be expressed as follows:
\begin{theorem} \label{prop_Chern_number_expression}
\begin{equation} \label{eq_Chern_number_expression}
\begin{aligned}
\mathcal{C}_{\Delta}
=-\frac{1}{\pi}\lim_{L\to \infty}\frac{1}{|\Omega_L|}\int_{\mathbf{C}}dm\frac{\partial \tilde{g}(z)}{\partial \overline{z}}
\int_{\Omega_L}dx\int_{\Omega_L}dx^{\prime}
\sum_{i,j\in\{1,2\}}\tilde{\varepsilon}_{ij}(\mathcal{V}_i G_{\sharp})(x,x^{\prime};z)
(\mathcal{V}_j \partial_z G_{\sharp})(x^{\prime},x;z).
\end{aligned}
\end{equation}
Here $|\Omega_L|$ denotes the area of $\Omega_L$, $m$ is the Lebesgue measure on $\mathbf{C}$ and $\tilde{\varepsilon}_{ij}:=\varepsilon_{ij3}$, where $\varepsilon_{ijk}$ is the standard Levi-Civita symbol.
\end{theorem}
Similarly, and importantly, the edge index $EI_{L}(\Delta)$ admits a similar expression in terms of the Green function $G_{L}$:
\begin{theorem} \label{prop_edge_index_expression}
\begin{equation} \label{eq_edge_index_expression}
\begin{aligned}
EI_{L}(\Delta)
=-\frac{1}{\pi}\int_{\mathbf{C}}dm\frac{\partial \tilde{g}(z)}{\partial \overline{z}}
\int_{\Omega_L}dx\int_{\Omega_L}dx^{\prime}
\sum_{i,j\in\{1,2\}}\tilde{\varepsilon}_{ij}(\mathcal{V}_i G_{L})(x,x^{\prime};z)
(\mathcal{V}_j \partial_z G_{L})(x^{\prime},x;z).
\end{aligned}
\end{equation}
\end{theorem}
\begin{proof}[Proof of Theorem \ref{prop_Chern_number_expression} and \ref{prop_edge_index_expression}]
See Sections 4 and 5.
\end{proof}
We emphasize that the similarity between the formula \eqref{eq_Chern_number_expression} and \eqref{eq_edge_index_expression} is a consequence of the Dirichlet-type boundary,  which plays a crucial role in ensuring energy conservation and, consequently, the bulk-edge correspondence in finite structures. This connection is explicitly established in the proof of Theorem \ref{prop_edge_index_expression}.
From this perspective, the results of this paper can be readily extended to domains with other boundary conditions such as Neumann-type boundaries.

The remainder of the proof involves estimating the difference between the two expressions \eqref{eq_Chern_number_expression} and \eqref{eq_edge_index_expression} as $L \to \infty$, which is carried out in Section 6. The key idea is that the influence of the boundary on the bulk region becomes negligible as the domain size increases, due to the absence of propagating modes within the spectral gap of the bulk. This is demonstrated in the proof by leveraging the exponential decay of the Green functions. For further details, see Section 6.

\subsection{Relation to previous works and perspectives} \label{sec_relation_perspec}

Compared with the previous studies of BEC, which focus primarily on domains with infinite size, the main result of this paper, Theorem \ref{thm_bec}, holds on finite domains. It states that if the gap Chern number of the bulk medium is nonzero (i.e., the medium is topologically nontrivial), the truncated system necessarily exhibits an in-gap spectrum as long as the domain is sufficiently large (but still finite). On the other hand, we note that an important feature of the usual edge index in the previous studies on infinite structures, such as the in-gap spectral flow along the direction of edge \cite{avila2013shortrange+transfer}, is that one can independently prove its quantization. However, the edge index defined in \eqref{eq_edge_index} is not quantized; instead, Theorem \ref{thm_bec} shows that it scales as $\mathcal{O}(L^2)$, with its average value converging to an integer as $L \to \infty$. Loosely speaking, this scaling can be understood intuitively: on a domain of diameter $L$, the position operator $x_i$ scales as $\mathcal{O}(L)$, and the momentum operator $\mathcal{V}_i = [\mathcal{L}, x_i]$ scales similarly. This distinction suggests that the quantization of the edge index is not a universal feature but depends on the specific system under consideration.

We point out that Theorem \ref{thm_bec} can be extended to quantum Hall systems on a finite sample by replacing the divergence-form operator $\mathcal{L} = -\nabla \cdot A \nabla$ with the magnetic Schrödinger Hamiltonian $\mathcal{H} = -(\nabla + i\mathcal{A})^2 + V$, where $\mathcal{A}$ and $V$ denote the vector and scalar potentials, respectively. In this case, the Chern number \eqref{eq_chern_num} should be replaced by the Hall conductance expressed through the Kubo formula (cf. \cite{Avron1994charge,taarabt2014landau+Ktheory,elgart2005shortrange+functional}, etc.):
\begin{equation} \label{eq_conduct_kubo}
    \sigma=\lim_{R\to\infty}\frac{i}{4R^2}\text{Tr}_{(-R,R)^2} \big(P_{E}\big[[P_{E},x_1],[P_{E},x_2]\big]\big),
\end{equation}
where $P_{E}=\mathbbm{1}_{(-\infty,E)}(\mathcal{H})$ ($E\in\Delta$) denotes the projection to the ground state with Fermi level lying in the gap. It is known that when $\mathcal{H}$ is periodic, the Hall conductance \eqref{eq_conduct_kubo} coincides with the Chern number \eqref{eq_chern_num} \cite{Avron1994charge}. More importantly, regardless of whether $\mathcal{H}$ is periodic, the Hall conductance can be expressed in the form \eqref{eq_Chern_number_expression}. Similarly, the edge index \eqref{eq_edge_index} remains valid in this setting by replacing $\mathcal{L}|_{\Omega_L}$ with the Hamiltonian $\mathcal{H}|_{\Omega_L} := \mathcal{H}\big|_{H_0^1(\Omega_L)}$ on the truncated domain. Physically, the edge index in this case represents the expectation value of the orbital angular momentum of in-gap electrons. With these adjustments, the proof of BEC for finite quantum Hall systems follows the same lines as sketched in Section \ref{sec_outline_proof}. 

We also expect that Theorem \ref{thm_bec} can be extended to a general second-order periodic self-adjoint elliptic operator of the form $\mathcal{L}=\sum_{|\alpha|\leq 2}c_{\alpha}(x)\partial^{\alpha}$, as long as the regularity estimates of the Green functions in Section 3 hold.

Another direction of generalization involves the consideration of disorder, which is common in realistic physical systems. This work focuses exclusively on the idealized case of a clean (disorder-free) system. However, we expect our results to extend to disordered systems. In such cases, the coefficient matrix $A$ may include a random perturbation of the form $A_{\omega}(x)=\sum_{n\in\mathbf{Z}^2}\omega_{n}s(x-n)$, where $\{\omega_{n}\}_{n}$ are i.i.d. random variables and $s(x)$ is a compactly supported single site function. To ensure that edge transport remains unaffected, it is generally assumed that the disorder induces only global localization without contributing to edge conduction phenomena \cite{Kellendonk02landau+ktheory,Kellendonk04landau+functional}. This scenario can be modeled, for example, by replacing the spectral gap $\Delta$ with a mobility gap \cite{elgart2005shortrange+functional,taarabt2014landau+Ktheory,cornean2021landau+functional}. Under these conditions, we expect Theorem \ref{thm_bec} to hold in the sense of expectation over randomness, a direction that we leave for future investigation.

Finally, we remark that Theorem \ref{thm_bec} is expected to remain valid when the smoothness assumption $A\in C^3(\mathbf{R}^2,\mathbf{C}^{2\times 2})$ is relaxed to allow $A$ to be piecewise constant, as is often the case in realistic physical systems. A possible approach to proving Theorem \ref{thm_bec} proceeds as follows. One starts with the divergence-form operator $\mathcal{L}$ with piecewise-constant coefficient matrix $A$, and construct a ``mollified'' operator $\mathcal{L}^{\epsilon}$, whose coefficient matrix $A^{\epsilon}$ is smooth and converges to $A$ as $\epsilon\to 0$. The bulk and edge indices associated with $\mathcal{L}^{\epsilon}$ are denoted as $\mathcal{C}^{\epsilon}$ and $EI_{L}^{\epsilon}$, respectively. With the correspondence between $\mathcal{C}^{\epsilon}$ and $EI_{L}^{\epsilon}$ established in this paper, it is sufficient to prove the convergence $\mathcal{C}^{\epsilon} \to \mathcal{C}$ and $EI_{L}^{\epsilon}\to EI_{L}$ as $\epsilon \to 0$. This can be achieved by establishing appropriate convergence of the eigenpairs of $\mathcal{L}^{\epsilon}$ and $\mathcal{L}^{\epsilon}\big|_{H_{0}^{1}(\Omega_L)}$ to those of $\mathcal{L}$ and $\mathcal{L}\big|_{H_{0}^{1}(\Omega_L)}$, respectively, which can be accomplished by standard perturbation analysis. We do not address this generalization here and leave this for interested readers.

\section{Preliminary}

\subsection{Notations and conventions}
In this paper, $B_{x_0,r}$ denotes an open ball in $\mathbf{R}^2$ centered at $x_0$ with radius $r$. When $x_0=0$, we abbreviate the notation as $B_r:=B_{0,r}$. For any $B_{x,r}$, $\chi_{x, r}$ denotes its indicator function.

The brackets $(\cdot,\cdot)_{\mathcal{H}}$ denote the inner product in the Hilbert space $\mathcal{H}$. If no subscript is present, $(\cdot,\cdot)$ denote the usual inner product in $\mathcal{H}=L^2(Y)$ with $Y=(0,1)\times (0,1)$ being the unit cell.

For $\kappa\in Y^*$ with $Y^*$ being the dual cell, we denote $L_{\kappa}^2(Y):=\{u\in L^2_{loc}(\mathbf{R}^2):\, u(x+n)=e^{i\kappa\cdot n}u(x)\,(\forall n\in\mathbf{Z}^2)\}$ which is equipped with $L^2(Y)$ inner product. The corresponding Sobolev spaces are denotes as $H^{m}_{\kappa}(Y):=\{u\in H^{m}_{loc}(\mathbf{R}^2):\, u\in L^2_{\kappa}(Y)\}$.

We denote $\mathcal{L}$ by the self-adjoint operator $-\nabla\cdot A\nabla$ that acts on $L^2(\mathbf{R}^2)$ and has the domain $\text{Dom}(\mathcal{L})=\{u\in L^2(\mathbf{R}^2):\,-\nabla\cdot A\nabla u\in L^2(\mathbf{R}^2)\}$. 

We denote $\mathcal{L}_{\Omega_L}:=\mathcal{L}\big|_{H_0^1(\Omega_L)}$ the restriction of $\mathcal{L}$ to $\Omega_L$ with the Dirichlet boundary conditions.

We denote $\mathcal{L}(\kappa)=\mathcal{L}|_{L_{\kappa}^2}$ the restriction of $\mathcal{L}$ to $L_{\kappa}^2(Y)$, which is referred to as
the Floquet transform of $\mathcal{L}$. 

We denote $\tilde{\mathcal{L}}(\kappa):=-(\nabla+i\kappa)\cdot A(\nabla+i\kappa)$, which acts on the space of periodic functions $L_{\kappa=0}^2(Y)$. $\tilde{\mathcal{L}}(\kappa)$ is 
unitarily equivalent to $\mathcal{L}(\kappa)$ . More precisely, $\mathcal{L}(\kappa)=e^{-i\kappa\cdot x}\tilde{\mathcal{L}}(\kappa)e^{i\kappa\cdot x}$.

For an integral operator $P$, we denote its associated integral kernel as $P(x,y)$. 

For a two-variable function $f(x, y)$ and a differential operator $H$, we write $H_x f$ (or $H_y f$) to indicate that the operator acts on the $x$ (or $y$) variable. By default, we refer to the case where the operator acts on the $x$ variable when no subscript is specified, i.e., $Hf := H_x f$.

For any multi-index $\alpha=(\alpha_1,\alpha_2)\in\mathbf{N}^2$, we denote the partial derivative $\partial^{\alpha}:=\partial^{\alpha_1}_{x_1}\partial^{\alpha_2}_{x_2}$. The order of $\partial^{\alpha}$ is denoted as $|\alpha|:=\alpha_1 +\alpha_2$.

We define the reduced Levi-Civita symbol $\tilde{\varepsilon}_{ij}:=\varepsilon_{ij3}$, where $\varepsilon_{ijk}$ is the standard Levi-Civita symbol.

We denote $F := \{n : \lambda_n(\kappa) < \lambda_{\text{low}}\}$ the set of indices of filled bands and $E := \{n : \lambda_n(\kappa) > \lambda_{\text{upp}}\} = \mathbf{N} \setminus F$ the set of indices of unfilled bands.

\subsection{Almost-analytic extension and Hellfer-Sjöstrand formula}

Let $g\in C_c^{\infty}(\mathbf{R})$ be the function in the formula \eqref{eq_edge_index}. A complex function $\tilde{g}$ is called an almost-analytic extension of $g$ if it satisfies
\begin{equation*}
\tilde{g}\in C_{c}^{\infty}(\mathbf{C}),\quad \tilde{g}|_{\mathbf{R}}=g,\quad \partial_{\overline{z}}\tilde{g}=\mathcal{O}(|\text{Im }z|^{\infty}).
\end{equation*}
The third condition above means that for any $N>0$ there exists $C_N>0$ such that $|\partial_{\overline{z}}\tilde{g}|\leq C_N|\text{Im }z|^{N}$. The following identities are obtained from Green's formula and Cauchy integral formula \cite{zworski2012semiclassical}
\begin{equation} \label{eq_domain_residue_formula}
g(t)=\frac{1}{\pi i}\int_{\mathbf{C}}\partial_{\overline{z}}\tilde{g}(z)(t-z)^{-1}dm ,\quad
g^{(n)}(t)=\frac{(-1)^n n!}{\pi i}\int_{\mathbf{C}}\partial_{\overline{z}}\tilde{g}(z)(t-z)^{-(n+1)}dm.
\end{equation}
Recall that $g(x)\equiv 1$ for $x<\lambda_{low}$ and $g(x)\equiv 0$ for $x>\lambda_{upp}$. The following statement follows directly from \eqref{eq_domain_residue_formula}:
\begin{corollary} \label{corol_residue_formula_for_g}
Suppose $f(z)$ is meromorphic in the support of $\tilde{g}$ with the Laurent expansion
\begin{equation*}
f(z)=\sum_{n=1}^{p_1}\frac{a_n}{z-z_n}+\sum_{n=p_1+1}^{p_1+p_2}\frac{b_n}{(z-z_n)^2}+\sum_{n=p_1+p_2+1}^{p_1+p_2+p_3}\frac{c_n}{(z-z_n)^3}+h(z),
\end{equation*}
where $h$ is holomorphic in the support of $\tilde{g}$. Assume the poles of $f$ satisfy
\begin{equation*}
    z_n\in \mathbf{R},\quad z_n<\lambda_{low}.
\end{equation*}
Then it holds that
\begin{equation*}
\frac{1}{\pi i}\int_{\mathbf{C}}\partial_{\overline{z}}\tilde{g}(z)f(z)dm
=\sum_{n=1}^{p_1}a_n.
\end{equation*}
\end{corollary}
Finally, we recall the Hellfer-Sjöstrand formula as a tool in functional calculus (cf. Theorem 14.8 of \cite{zworski2012semiclassical}): the operator $g^{\prime}(\mathcal{L}_{\Omega_L})$ is expressed as
\begin{equation} \label{eq_hellfer_sjostrand}
g^{\prime}(\mathcal{L}_{\Omega_L})=\frac{1}{\pi i}\int_{\mathbf{C}}\frac{\partial^2\tilde{g}(z)}{\partial \overline{z}\partial z}(\mathcal{L}_{\Omega_L}-z)^{-1}dm.
\end{equation}

\subsection{Trace-class properties of singular integral operators}
We review some results on the trace-class properties of singular integral operators on a (bounded or not bounded) domain $O\subset \mathbf{R}^2$.

Let $P\in \mathcal{B}(L^2(O))$ with the associated integral kernel $P(x,y)$. We say $P$ is trace-class if $\sum_{k}(|P|e_k,e_k)_{L^2(O)}<\infty$, where $|P|=\sqrt{P^* P}$ and $e_k$ is an orthonormal basis of $L^2(O)$. The set of trace-class operators is denoted as $S_1(L^2(O))$. For $P\in S_1(L^2(O))$, its trace is defined as $\text{Tr }(P)=\sum_{k}(Pe_k,e_k)_{L^2(O)}$. $S_1(L^2(O))$ is an ideal of $\mathcal{B}(L^2(O))$ in the following sense: 
\begin{proposition}[\cite{reed1972methods}]
\label{prop_trace_cycle}
If $P\in S_1(L^2(O))$ and $Q\in \mathcal{B}(L^2(O))$, then $PQ,QP\in S_1(L^2(O))$ and $\text{Tr }(PQ)=\text{Tr }(QP)$.
\end{proposition}
A criterion for $P$ being trace-class is given by \cite[Theorem 3.1]{brislawn1988kernels}:
\begin{proposition}
\label{prop_trace_class_kernel}
If $\sup_{x,y\in O}|P(x,y)|<\infty$ and $O$ is bounded, then $P$ is trace-class. Moreover
\begin{equation*}
    \text{Tr}(P)=\int_{O}P(x,x)dx.
\end{equation*}
\end{proposition}
Note that \cite[Theorem 3.1]{brislawn1988kernels} originally states that $P$ is trace-class if $\int_{O}(MP)(x,x)dx<\infty$, where $MP$ is the Hardy-Littlewood maximal function of the kernel $P(x,y)$. This condition is weaker than the uniform boundedness as stated in Proposition \ref{prop_trace_class_kernel}, which already serves the need of this paper.

The following proposition justifies the composition rule of integral kernels.
\begin{proposition} \label{prop_composition rule}
Suppose $P,Q\in \mathcal{B}(L^2(O))$ ($O$ is not necessarily bounded) whose integral kernels are measurable and satisfy
\begin{equation} \label{eq_composition_rule_cond_1}
    \sup_{x\in O}\int_{O}|P(x,y)|dy<\infty ,\quad
    \sup_{x\in O}\int_{O}|Q(x,y)|dy<\infty .
\end{equation}
Then the integral kernel of the composition $PQ$ is given by $(PQ)(x,y)=\int_{O}P(x,z)Q(z,y)dz$. The integral operators that satisfy \eqref{eq_composition_rule_cond_1} form an algebra of $\mathcal{B}(L^2(O))$ in the sense that
\begin{equation} \label{eq_composition_rule_cond_2}
    \sup_{x\in O}\int_{O}|(PQ)(x,y)|dy<\infty,
\end{equation}
whenever \eqref{eq_composition_rule_cond_1} is satisfied.
\end{proposition}
\begin{proof}
To show $(PQ)(x,y)=\int_{O}P(x,z)Q(z,y)dz$, it's sufficient to show the following equality for $f\in C_c^{\infty}(O)$: 
\begin{equation*}
\int_{O}P(x,z)dz\int_{O}Q(z,y)f(y)dy=\int_{O}\big(
\int_{O}P(x,z)Q(z,y)dz \big)f(y)dy
,\quad\text{a.e. $x\in O$}.
\end{equation*}
The interchange of integrals above is justified by the Fubini theorem and the following inequality
\begin{equation*}
\begin{aligned}
\int_{O\times O}|P(x,z)||Q(z,y)||f(y)|dzdy
&\leq \|f\|_{L^{\infty}(O)}\int_{O\times O}|P(x,z)||Q(z,y)|dzdy \\
&\leq \big(\sup_{z\in O}\int_{O}|Q(z,y)|dy\big)\big(\sup_{x\in O}\int_{O}|P(x,z)|dz\big)\|f\|_{L^{\infty}(O)}.
\end{aligned}
\end{equation*}
Similarly, \eqref{eq_composition_rule_cond_2} is justified by the following inequality
\begin{equation*}
\begin{aligned}
\sup_{x\in O}\int_{O}|(PQ)(x,y)|dy
&\leq \sup_{x\in O}\int_{O}dy\big(\int_{O}|P(x,z)||Q(z,y)|dz\big) \\
&=\sup_{x\in O}\int_{O}|P(x,z)|dz\int_{O}|Q(z,y)|dy \\
&\leq \sup_{x\in O}\int_{O}|P(x,z)|dz
\cdot\sup_{z\in O}\int_{O}|Q(z,y)|dy.
\end{aligned} 
\end{equation*}

\end{proof}

\subsection{$L^p$-estimate of second-order divergence-form operators with complex coefficients} \label{sec_Lp}
We review some results on the $L^p$-estimate of second-order divergence-form operators with complex coefficients. Those estimates help to prove the exponential decay of Green functions, which lies in the center of the proof of Theorem \ref{thm_bec}. Note that we assume the coefficient matrix $A\in C^3(\mathbf{R}^2;\mathbf{C}^{2\times 2})$ such that all regularity estimates in this section hold.

We first recall the following interior De Giorgi-type estimate \cite{auscher96complex_de_giorgi}, where the author proves the following results under the weaker assumption that $A(x)$ is uniformly continuous.
\begin{proposition} \label{prop_de_giorgi_local}
Suppose $u\in H^1(\mathbf{R}^2)$ solves $\mathcal{L}u=f$ weakly in an open ball $B_{x_0,r}$ for some $f\in L^2(B_{x_0,r})$. Then, for any $r^{\prime}<r$, it holds
\begin{equation} \label{eq_de_giorgi_local}
    \|u\|_{L^{\infty}(B_{x_0,r^{\prime}})}\leq C \|f\|_{L^{2}(B_{x_0,r})},
\end{equation}
where $C>0$ depends only on $r^{\prime},r$ and the coefficient matrix $A(x)$.
\end{proposition}
This result also applies to a bounded domain with smooth boundary.
\begin{proposition} \label{prop_de_giorgi_local_boundary}
Let $\Omega\subset \mathbf{R}^2$ is bounded and $\partial \Omega$ is $C^{\infty}$. Suppose $u\in H^1(\Omega)$ solves $\mathcal{L}u=f$ weakly in  $B_{x_0,r}\cap \Omega$ for some $f\in L^2(B_{x_0,r})$ and $u\big|_{\partial \Omega\cap B_{x_0,r}}=0$. Then, for any $r^{\prime}<r$, it holds
\begin{equation} \label{eq_de_giorgi_local_boundary}
    \|u\|_{L^{\infty}(B_{x_0,r^{\prime}}\cap\Omega)}\leq C \|f\|_{L^{2}(B_{x_0,r}\cap\Omega)},
\end{equation}
where $C>0$ depends on $r^{\prime},r$, the coefficient matrix $A(x)$ and $\partial \Omega$.
\end{proposition}
The solution to higher-order elliptic equations has improved regularity compared with the $L^{\infty}$ estimate in Proposition \ref{prop_de_giorgi_local}, as shown below. The proof is given in Lemma 3.2 of \cite{auscher1998heat} (see also Corollary 4.14 of \cite{auscher96complex_de_giorgi}).

\begin{proposition} \label{prop_de_giorgi_double_resolvent}
Suppose $u\in H^1(\mathbf{R}^2)$ solves $\mathcal{L}^2 u=f$ weakly in an open ball $B_{x_0,r}$ for some $f\in L^2(B_{x_0,r})$. Then, there exists $\nu\in (0,1/2)$ such that for any $r^{\prime}<r$ and $p\in[2,\frac{2}{1-2\nu})$, it holds
\begin{equation} \label{eq_de_giorgi_double_resolvent}
    \|u\|_{W^{1,p}(B_{x_0,r^{\prime}})}\leq C \|f\|_{L^{2}(B_{x_0,r})},
\end{equation}
where $C>0$ depends only on $r^{\prime},r$ and the coefficient matrix $A(x)$.
\end{proposition}
\begin{remark}
The Sobolev embedding indicates $W^{1,p}(Y)$ functions are Hölder continuous and, therefore belong to $L^{\infty}(Y)$. This highlights the improvement provided by Proposition \ref{prop_de_giorgi_double_resolvent} compared to Proposition \ref{prop_de_giorgi_local}.
\end{remark}
The following $W^{1,p}$ estimate of Bloch eigenfunction follows from Proposition \ref{prop_de_giorgi_double_resolvent}.

\begin{corollary} \label{corol_eigenfunction_w1p}
Let $\{\lambda_n(\kappa),v_{n}(x;\kappa)\}$ be the Bloch eigenpairs of $\mathcal{L}(\kappa)$. Then, for any $p\in [2,2+\frac{2}{1-\nu})$ ($\nu>0$ is introduced in Proposition \ref{prop_de_giorgi_double_resolvent}), it holds
\begin{equation} \label{eq_eigenfunction_w1p}
    \|v_n(\cdot;\kappa)\|_{W^{1,p}(Y)}\leq C\lambda_{n}^{2-\frac{3}{p}}(\kappa),
\end{equation}
where $C>0$ only depends on $p$ and the coefficient matrix $A(x)$.
\end{corollary}
\begin{proof}
Note that 
$$
\mathcal{L}v_n(x;\kappa)=\lambda_n(\kappa) v_n(\kappa),\quad x\in\mathbf{R}^2. 
$$
Taking inner product with $v_n$ at both sides and utilizing the ellipticity of $\mathcal{L}$, one obtains $\|v_n(\cdot;\kappa)\|_{W^{1,2}(Y)}\leq C\lambda_{n}^{\frac{1}{2}}(\kappa)$, which justifies \eqref{eq_eigenfunction_w1p} for $p=2$. On the other hand,
$$
\mathcal{L}^2v_n(x;\kappa)=\lambda_n^2 (\kappa) v_n(\kappa),\quad x\in\mathbf{R}^2.
$$
Hence, for any $p_0\in (2,\frac{2}{1-\nu})$, Proposition \ref{prop_de_giorgi_double_resolvent} indicates
\begin{equation*}
    \|v_n(\cdot;\kappa)\|_{W^{1,p_0}(Y)}\leq C\lambda_n^2(\kappa).
\end{equation*}
Then \eqref{eq_eigenfunction_w1p} follows by interpolation
\begin{equation*}
\|v_n(\cdot;\kappa)\|_{W^{1,p}(Y)}\leq C\|v_n(\cdot;\kappa)\|_{W^{1,p_0}(Y)}^{\frac{p-2}{p}}\cdot \|v_n(\cdot;\kappa)\|_{W^{1,2}(Y)}^{\frac{2}{p}}
\leq C\lambda_{n}^{\frac{2(p-2)}{p}}(\kappa)\cdot \lambda_{n}^{\frac{1}{p}}(\kappa).
\end{equation*}

\end{proof}

\section{Estimate of Green functions }

\subsection{Estimate of the Green function $G_{\sharp}(x,x^{\prime};z)$ and $G_{L}(x,x^{\prime};z)$}

Recall that the Green functions $G_{\sharp}(x,x^{\prime};z)$ and $G_{L}(x,x^{\prime};z)$ are defined as the integral kernels of $(\mathcal{L}-z)^{-1}$ and $(\mathcal{L}_{\Omega_L}-z)^{-1}$, respectively. We prove the exponential decay of $G_{\sharp}$ and $G_{L}$ when $z\notin\mathbf{R}$, an important ingredient of the proof of Theorem \ref{thm_bec}.
\begin{proposition}
\label{prop_G_sharp_decay}
For $z\in K$, where $K\subset\mathbf{C}$ is compact, there exists $q_1,p_1>0$ such that for $|x-x^{\prime}|\geq \frac{1}{2}$ and $|\alpha|,|\beta|\leq 1$, the following estimates hold: 
\begin{equation}
\label{eq_G_sharp_decay_1}
|\partial^{\alpha}_{x}G_{\sharp}(x,x^{\prime};z)|+|\partial^{\beta}_{x^{\prime}}G_{\sharp}(x,x^{\prime};z)|\leq \frac{C}{|\text{Im }z|^{q_1}}e^{-D|\text{Im }z|^{p_1}|x-x^{\prime}|},
\end{equation}
where $C,D>0$ depends only on the coefficient matrix $A$ and the domain $K$.
\end{proposition}

\begin{proposition}
\label{prop_G_L_decay}
For $z\in K$, where $K\subset\mathbf{C}$ is compact, there exists $q_1,p_1>0$ such that for $|x-x^{\prime}|\geq \frac{1}{2}$ and $|\alpha|,|\beta|\leq 1$, the following estimates hold: 
\begin{equation}
\label{eq_G_L_decay_1}
|\partial^{\alpha}_{x}G_{L}(x,x^{\prime};z)|+|\partial^{\beta}_{x^{\prime}}G_{L}(x,x^{\prime};z)|\leq \frac{C}{|\text{Im }z|^{q_1}}e^{-D|\text{Im }z|^{p_1}|x-x^{\prime}|},
\end{equation}
where $C$ depends only on the coefficient matrix $A$, the domain $K$ and shape of $\partial \Omega$.
\end{proposition}

The proof of Proposition \ref{prop_G_sharp_decay} is provided at the end of this section. All the arguments in that proof can be adapted to establish Proposition \ref{prop_G_L_decay}, with the only difference being the use of Proposition \ref{prop_de_giorgi_local_boundary} in place of Proposition \ref{prop_de_giorgi_local} to obtain the pointwise estimate.

\begin{remark}
The exponential decay demonstrated here is not new in principle. However, we provide a detailed proof because it constitutes a crucial technical result leading to our main theorem, and the existing estimates in the literature are not directly applicable to our specific framework. The exponential decay of Green functions has been extensively studied in various contexts, including $L^2$-estimates for $2D$ classical wave systems \cite{figotin1996localization} and electronic system \cite{agmon82decay,lu11green,fefferman2020continuum}, pointwise estimates for divergence-form operators via various methods such as the heat kernel method \cite{auscher1998heat,auscher2019gaussian,auscher2000equivalence}, iteration method \cite{blanc2013asymptotic_green,miyazaki2004lp_resolvent,kim2013green} and Floquet transform method \cite{kuchment2012green,elst2001periodic_heat_kernel}. See also \cite{gruter1982green,hofmann2007green,alfonseca2011analyticity,kang2010green_global} for results in higher dimensions.   
\end{remark}

We also require estimates on the singularities of Green functions, as proposed below. Similar estimates are provided in Theorems 4.8 and 4.9 of \cite{wang2023uniform} for the case where the coefficient matrix $A(\bm{x})$ is real-valued. While it is reasonable to expect that the arguments in that paper could be extended to our setting ( 2D elliptic operators with $C^3$ complex-valued coefficients), at least when $\text{Im }A(x)$ is sufficiently small, we choose to retain these as hypotheses for the sake of caution. Proving them in full detail is beyond the scope of this paper, as it is not the primary focus of our work.

\begin{hypothesis}
\label{prop_GL_singularities}
For $z\in K$, where $K\subset\mathbf{C}$ is compact, and any $L>1$, $s\in (0,1)$, there exists $q_2>0$ such that
\begin{equation}
\label{eq_GL_singularities_1}
|G_{L}(x,x^{\prime};z)|\leq \frac{C}{|\text{Im }z|^{q_2}}\Big\{
L^s+\log|x-x^{\prime}|\Big\},
\end{equation}
\begin{equation}
\label{eq_GL_singularities_2}
|\partial^{\alpha}_{x} G_{L}(x,x^{\prime};z)|+|\partial^{\beta}_{x^{\prime}} G_{L}(x,x^{\prime};z)|\leq \frac{C}{|\text{Im }z|^{q_2}}\frac{1}{|x-x^{\prime}|},
\end{equation}
for $|\alpha|,|\beta|=1$, 
where $C$ depends on $s$, the coefficient matrix $A$, the domain $K$ and the shape of $\partial \Omega$.
\end{hypothesis}
\begin{remark}
Note that the estimate \eqref{eq_GL_singularities_1} depends on the domain size $L$ as a consequence of the boundary condition. This dependence can be understood by considering the Green function $\log\frac{L}{|x|}$ for the Laplacian on the open ball $B_L$. On the other hand, the polynomial dependence of the estimate on $|\text{Im }z|$ is a consequence of the resolvent estimate $\|(\mathcal{L}_{\Omega_L}-z)^{-1}\|_{\mathcal{B}(L^2(\Omega_L))}=\mathcal{O}(|\text{Im }z|^{-1})$, which is implicitly used in the proof in \cite{wang2023uniform}. Roughly speaking, 
the pointwise estimates stated in Hypothesis \ref{prop_GL_singularities} are obtained by starting from the \(L^2\)-bound of the resolvent. An iterative argument, employing De Giorgi-type estimates, is then applied to transition from the \(L^2\)-bound to a pointwise estimate. Importantly, this second step does not introduce any additional singularity with respect to \(|\text{Im } z|\) (see \cite{auscher96complex_de_giorgi,dindos19regularity}). Consequently, the resulting estimate exhibits, at worst, polynomial growth as \(|\text{Im } z| \to 0\).
\end{remark}

We also need the estimation of the logarithmic behavior of the whole-space Green function $G_{\sharp}(x,x^{\prime};z)$. Note that we do not have the $L-$factor since we work in the whole space.
\begin{hypothesis}
\label{prop_G_sharp_singularities}
For $z\in K$, where $K\subset\mathbf{C}$ is compact, there exists $q_3>0$ such that
\begin{equation}
\label{eq_G_sharp_singularities_1}
|G_{\sharp}(x,x^{\prime};z)|\leq \frac{C}{|\text{Im }z|^{q_3}}\big(1+\log|x-x^{\prime}|\big),
\end{equation}
\begin{equation}
\label{eq_G_sharp_singularities_2}
|\partial^{\alpha}_{x} G_{\sharp}(x,x^{\prime};z)|+|\partial^{\beta}_{x^{\prime}} G_{\sharp}(x,x^{\prime};z)|\leq \frac{C}{|\text{Im }z|^{q_3}}\frac{1}{|x-x^{\prime}|} \quad \mbox{for}\,\,|\alpha|,|\beta|=1, 
\end{equation}
where $C$ depends only on the coefficient matrix $A$ and the domain $K$.
\end{hypothesis}

\begin{proof}[Proof of Proposition \ref{prop_G_sharp_decay}]
The proof is structured as follows: In Step 1, we use a standard Combes-Thomas argument to obtain exponential decay in the $L^2$-norm. In Step 2, we apply a De Giorgi-type estimate to extend this result to the pointwise estimate \eqref{eq_G_sharp_decay_1} for $\alpha = \beta = 0$. In Step 3, we outline the proof of \eqref{eq_G_sharp_decay_1} for $|\alpha| = |\beta|=1$.

{\color{blue}Step 1.} Fix $x_0, x_0^{\prime}$ such that $|x_0 - x_0^{\prime}| \geq 1/2$. Using the Combes-Thomas argument (as in Lemma 12 of \cite{figotin1996localization}), we show:
\begin{equation}
\label{eq_G_sharp_decay_proof_0}
\|f\cdot \chi_{x_0,1/6}(\mathcal{L}-z)^{-1}(g\cdot \chi_{x^{\prime}_0,1/6})\|_{L^2(\mathbf{R}^2)}\leq \frac{C}{|\text{Im }z|^{q}}e^{-D|\text{Im }z|^{p}|x_0-x^{\prime}_0|}\|f\|_{L^{\infty}(\mathbf{R}^2)}\|g\cdot \chi_{x^{\prime}_0,1/6}\|_{L^2(\mathbf{R}^2)}
\end{equation}
for any $f,g\in L^{\infty}(\mathbf{R}^2)$,  where $\chi_{x, 1/6}$ is the indicator function of the ball $B_{x, 1/6}$, and $C, D$ depend only on the coefficient matrix $A$ and the range of $z$.

To prove this, we define $\mathcal{L}_a := e^{a \cdot x} \mathcal{L} e^{-a \cdot x}$, where $|a| < 1$, as the self-adjoint realization of the quadratic form:
\begin{equation*}
L_a[\psi]:=\big(
A\nabla e^{-a\cdot x}\psi,\nabla e^{a\cdot x}\psi \big)_{(L^2(\mathbf{R}^2))^3}
=\big(
A(\nabla -a) \psi,(\nabla +a)\psi \big)_{(L^2(\mathbf{R}^2))^3}
,\quad \psi\in C_0^1(\mathbf{R}^2).
\end{equation*}
Note that $L_0$ is the quadratic form associated with $\mathcal{L}$. We now estimate
\begin{equation} \label{eq_G_sharp_decay_proof_1}
\begin{aligned}
|L_a[\psi]-L_0[\psi]|
\leq \big|\big(
A\cdot a \psi,\nabla\psi \big)_{(L^2(\mathbf{R}^2))^3} \big|
+\big|\big(
A\nabla \psi,a\psi \big)_{(L^2(\mathbf{R}^2))^3} \big|
+\big|\big(
A\cdot a \psi,a\psi \big)_{(L^2(\mathbf{R}^2))^3} \big|. 
\end{aligned}
\end{equation}
Using the Cauchy-Schwarz inequality, we can derive that
\begin{equation}
\label{eq_G_sharp_decay_proof_2}
\begin{aligned}
|L_a[\psi]-L_0[\psi]|
&\leq C\Big(|a|(1+|a|)\|\psi\|_{L^2(\mathbf{R}^2)}^2+|a|\|\nabla\psi\|_{(L^2(\mathbf{R}^2))^3}^2\Big) \\
&\leq C\Big(|a|\|\psi\|_{L^2(\mathbf{R}^2)}^2+|a|L_0[\psi]\Big),
\end{aligned}
\end{equation}
where $C>0$ depends only on $\|A\|_{\infty}$. Here the ellipticity of $L_0$ and the fact that $|a|<1$ is used in the second inequality. Hence, for $|a|<1$ being sufficiently small, \eqref{eq_G_sharp_decay_proof_2} implies that the quadratic form $L_a$ is relatively bounded by $L_0$, ensuring the invertibility of $\mathcal{L}_a - z$. More precisely, by Theorem 3.9, Chapter VI of \cite{kato2013perturbation}, when the following holds for $z\notin \sigma(\mathcal{L})$: 
\begin{equation*}
\label{eq_G_sharp_decay_proof_3}
t_{a,z}:=C|a|\|(1+\mathcal{L})(\mathcal{L}-z)^{-1}\|<1, 
\end{equation*}
the operator $\mathcal{L}_a-z$ is invertible with the estimate
\begin{equation*}
\|(\mathcal{L}_a-z)^{-1}\|\leq \big(
1+\frac{4t_{a,z}}{(1-t_{a,z})^2}
\big)\|(\mathcal{L}-z)^{-1}\|.
\end{equation*}
Using the standard bound $|(\mathcal{L} - z)^{-1}| \leq 1 / |\text{Im } z|$, we can set
\begin{equation*}
|a| \leq m_z := \frac{|\text{Im } z|}{2C(2 + |z|)},
\end{equation*}
to ensure 
\begin{equation*}
t_{a,z}=C|a|\|(\mathcal{L}-z)^{-1}+1+z(\mathcal{L}-z)^{-1} \|
\leq C|a|\big(1+\frac{1+|z|}{|\text{Im } z|} \big) \leq \frac{1}{2}.
\end{equation*}
This yields the estimate
\begin{equation}
\label{eq_G_sharp_decay_proof_4}
\|(\mathcal{L}_a-z)^{-1}\|\leq \frac{9}{|\text{Im }z|}.
\end{equation}
Now, we take $a=m_z \frac{x_0-x_0^{\prime}}{|x_0-x_0^{\prime}|}$. By the identity $(\mathcal{L}-z)^{-1}=e^{-a\cdot x}(\mathcal{L}_a-z)^{-1}e^{a\cdot x}$, we calculate
\begin{equation*}
\begin{aligned}
f\cdot \chi_{x_0,1/6}(\mathcal{L}-z)^{-1}(g\cdot \chi_{x^{\prime}_0,1/6})
&=f\cdot \chi_{x_0,1/6}e^{-a\cdot x}(\mathcal{L}_a-z)^{-1}e^{a\cdot x}(g\cdot \chi_{x^{\prime}_0,1/6}) \\
&=e^{-m_z|x_0-x^{\prime}_0|}f\cdot \chi_{x_0,1/6}e^{-a\cdot (x-x_0)}(\mathcal{L}_a-z)^{-1}e^{a\cdot (x-x^{\prime}_0)}(g\cdot \chi_{x^{\prime}_0,1/6}).
\end{aligned}
\end{equation*}
Hence
\begin{equation*}
\begin{aligned}
&\|f\cdot \chi_{x_0,1/6}(\mathcal{L}-z)^{-1}(g\cdot \chi_{x^{\prime}_0,1/6})\|_{L^2(\mathbf{R}^2)} \\
&\leq e^{-m_z|x_0-x^{\prime}_0|}
\|f\cdot \chi_{x_0,1/6}e^{-a\cdot (x-x_0)}\|_{L^\infty}\cdot \|(\mathcal{L}_a-z)^{-1}\|\cdot \|e^{a\cdot (x-x^{\prime}_0)}(g\cdot \chi_{x^{\prime}_0,1/6})\|_{L^2} \\
&\leq \frac{C}{|\text{Im }z|}\cdot e^{-m_z|x_0-x^{\prime}_0|} \|f\|_{L^{\infty}(\mathbf{R}^2)}\|g\cdot \chi_{x^{\prime}_0,1/6}\|_{L^2(\mathbf{R}^2)}.
\end{aligned}   
\end{equation*}
This proves \eqref{eq_G_sharp_decay_proof_0} with $p=q=1$ and $D= \frac{1}{2C(2+ \min_{z\in K}|z|)}$. 

{\color{blue}Step 2.} We prove the pointwise estimate \eqref{eq_G_sharp_decay_1} with $\alpha=\beta=0$ using a De Giorgi-type estimate. Define $u=(\mathcal{L}-z)^{-1}(g\cdot \chi_{x^{\prime}_0,1/6})$. By taking $f\equiv 1$ in \eqref{eq_G_sharp_decay_proof_0}, we have
\begin{equation} \label{eq_G_sharp_decay_proof_6}
\|u\|_{L^2(B_{x_0,\frac{1}{6}})}\leq \frac{C}{|\text{Im }z|^{q}}e^{-D|\text{Im }z|^{p}|x_0-x^{\prime}_0|}\|g\cdot \chi_{x^{\prime}_0,1/6}\|_{L^2(\mathbf{R}^2)}.
\end{equation}
Since $B_{x_0,1/6}\cap B_{x_0^{\prime},1/6}=\emptyset$, the function $u$ satisfies the following equation in the weak sense: 
\begin{equation} \label{eq_G_sharp_decay_proof_solu}
    \mathcal{L}u=zu+ g\cdot \chi_{x^{\prime}_0,1/6} =zu \quad \text{in $B_{x_0,1/6}$}.
\end{equation}
By Proposition \ref{prop_de_giorgi_local}, we can deduce that
\begin{equation*}
\begin{aligned}
|u(x_0)|\leq \|u\|_{L^{\infty}(B(x_0,1/12))}\leq \frac{C|z|}{|\text{Im }z|^{q}}e^{-D|\text{Im }z|^{p}|x_0-x^{\prime}_0|}\|g\cdot \chi_{x^{\prime}_0,1/6}\|_{L^2(\mathbf{R}^2)},
\end{aligned}
\end{equation*}
or equivalently
\begin{equation*}
\begin{aligned}
\Big|\int_{B_{x^{\prime}_0,1/6}}dx^{\prime}G_{\sharp}(x_0,x^{\prime};z)g(x^{\prime})\Big|
\leq \frac{C|z|}{|\text{Im }z|^{q}}e^{-D|\text{Im }z|^{p}|x_0-x^{\prime}_0|}\|g\|_{L^2(B_{x^{\prime}_0,1/6})}.
\end{aligned}
\end{equation*}
Now, take $g(x^{\prime}):=\overline{G_{\sharp}(x_0,x^{\prime};z)}$, which lies in $L^{\infty}$ since $B_{x_0,1/6}\cap B_{x_0^{\prime},1/6}=\emptyset$. Substituting this into the inequality above, we find
\begin{equation} \label{eq_G_sharp_decay_proof_5}
\begin{aligned}
\|G_{\sharp}(x_0,\cdot;z)\|_{L^2(B_{x^{\prime}_0,1/6})}
\leq \frac{C|z|}{|\text{Im }z|^{q}}e^{-D|\text{Im }z|^{p}|x_0-x^{\prime}_0|}.
\end{aligned}
\end{equation}
Then the pointwise estimate \eqref{eq_G_sharp_decay_1} with $\alpha=\beta=0$ follows from 
Proposition \ref{prop_de_giorgi_local}.


{\color{blue}Step 3.} The proof of \eqref{eq_G_sharp_decay_1} for $|\alpha| = 1$ and $\beta = 0$ follows a similar argument as in the case of $\alpha = \beta = 0$.
First, using \eqref{eq_G_sharp_decay_proof_solu} and a standard cut-off argument for estimating the interior $H^2$-norm of solutions to elliptic equations, we obtain
\begin{equation*}
    \|u\|_{H^2(B_{x_0,1/12})}
    \leq C|z| \|u\|_{L^2(B_{x_0,1/6})}.
\end{equation*}
Substituting the decay estimate from \eqref{eq_G_sharp_decay_proof_6}, we derive
\begin{equation} \label{eq_G_sharp_decay_proof_7}
\|u\|_{H^2(B_{x_0,1/12})}
\leq \frac{C|z|}{|\text{Im }z|^{q}}e^{-D|\text{Im }z|^{p}|x_0-x^{\prime}_0|}\|g\cdot \chi_{x^{\prime}_0,1/6}\|_{L^2(\mathbf{R}^2)}.
\end{equation}
Next, note that $\partial_{x_i} u$ (for $i = 1, 2$) satisfies the following equation in $B_{x_0, 1/6}$:
\begin{equation*}
\mathcal{L}_x(\partial_{x_i}u)=z\partial_{x_i}u-\nabla\cdot (\partial_{x_i}A)\nabla u.
\end{equation*}
Since $\partial_{x_i} A \in C^1$ by assumption, we can apply Proposition \ref{prop_de_giorgi_local} to obtain
\begin{equation*}
\begin{aligned}
|\partial_{x_i}u(x_0)|
\leq \|\partial_{x_i}u(x_0)\|_{L^{\infty}(B_{x_0,1/24})}
&\leq C\big(|z|\|\partial_{x_i}u\|_{L^2(B_{x_0,1/12})}+\|\nabla\cdot (\partial_{x_i}A)\nabla u\|_{L^2(B_{x_0,1/12})} \big) \\
&\leq C(|z|+1)\|u\|_{H^2(B_{x_0,1/12})}.
\end{aligned}
\end{equation*}
Substituting the exponential decay of $\|u\|_{H^2(B_{x_0, 1/12})}$ from \eqref{eq_G_sharp_decay_proof_7}, we arrive at
\begin{equation*}
|\partial_{x_i}u(x_0)|
\leq  \frac{C(|z|^2+|z|)}{|\text{Im }z|^{q}}e^{-D|\text{Im }z|^{p}|x_0-x^{\prime}_0|}\|g\cdot \chi_{x^{\prime}_0,1/6}\|_{L^2(\mathbf{R}^2)}.
\end{equation*}
Recalling that $u=(\mathcal{L}-z)^{-1}(g\cdot \chi_{x^{\prime}_0,1/6})$, the above inequality becomes
\begin{equation*}
\Big|\int_{B_{x^{\prime}_0,1/6}}dx^{\prime}\partial_{x_i}G_{\sharp}(x_0,x^{\prime};z)g(x^{\prime})\Big|
\leq  \frac{C(|z|^2+|z|)}{|\text{Im }z|^{q}}e^{-D|\text{Im }z|^{p}|x_0-x^{\prime}_0|}\|g\cdot \chi_{x^{\prime}_0,1/6}\|_{L^2(\mathbf{R}^2)}.
\end{equation*}
Finally, the pointwise estimate \eqref{eq_G_sharp_decay_1} for $|\alpha|=1$ and $\beta=0$ follows by choosing $g(x^{\prime})=\partial_{x_i}\overline{G_{\sharp}(x_0,x^{\prime};z)}$ in the above formula and applying Proposition \ref{prop_de_giorgi_local}, 
as in Step 2. The proof for the other case $\alpha=0, |\beta| = 1$, follows similarly.
\end{proof}

In the sequel of this paper, we set $q=\min\{q_1,q_2,q_3\}$ and $p=\min\{p_1,p_2,p_3\}$ such that the estimates in Proposition \ref{prop_G_sharp_decay}, \ref{prop_G_L_decay}, Hypothesis \ref{prop_GL_singularities} and \ref{prop_G_sharp_singularities} hold with the exponents $q_i$ and $p_i$ are replaced by $q$ and $p$, respectively. Also, we fix the constant $s=\frac{1}{3}$ in Hypothesis \ref{prop_GL_singularities}. Since we focus on a compact subset in the complex plane (i.e. $\text{supp }\tilde{g}$), we can set $D=1$ in Proposition \ref{prop_G_sharp_decay} and \ref{prop_G_L_decay} without loss of generality. Recalling that $\mathcal{V}_i=-e_{i}\cdot \nabla A-\nabla\cdot A e_i$ are first-order differential operators with uniformly bounded coefficients, we now summarize the results of this section into the following Corollary:

\begin{corollary} \label{corol_estimates_Vi_G}
Let $z\in\text{supp }\tilde{g}$. For any $x\neq x^{\prime}$ and $i,j\in\{1,2\}$, the following holds
\footnotesize
\begin{equation} \label{eq_estimates_Vi_G_1}
\begin{aligned}
|G_{\sharp}(x,x^{\prime};z)|+
|\mathcal{V}_{i,x}G_{\sharp}(x,x^{\prime};z)|+|\mathcal{V}_{j,x^{\prime}}G_{\sharp}(x,x^{\prime};z)|&\leq \frac{C}{|\text{Im }z|^q}\big(1+|\log|x-x^{\prime}||+\frac{1}{|x-x^{\prime}|}\big), \\
|G_{L}(x,x^{\prime};z)|+
|\mathcal{V}_{i,x}G_{L}(x,x^{\prime};z)|+|\mathcal{V}_{j,x^{\prime}}G_{L}(x,x^{\prime};z)|&\leq \frac{C}{|\text{Im }z|^q}\big(L^{\frac{1}{3}}+|\log|x-x^{\prime}||+\frac{1}{|x-x^{\prime}|}\big).
\end{aligned}
\end{equation}
\normalsize
When $|x-x^{\prime}|\geq \frac{1}{2}$, we have the following exponential decay estimates:
\footnotesize
\begin{equation} \label{eq_estimates_Vi_G_2}
\begin{aligned}
|G_{\sharp}(x,x^{\prime};z)|+
|\mathcal{V}_{i,x}G_{\sharp}(x,x^{\prime};z)|+|\mathcal{V}_{j,x^{\prime}}G_{\sharp}(x,x^{\prime};z)|&\leq \frac{C}{|\text{Im }z|^q}e^{-|\text{Im }z|^p |x-x^{\prime}|}, \\
|G_{L}(x,x^{\prime};z)|+
|\mathcal{V}_{i,x}G_{L}(x,x^{\prime};z)|+|\mathcal{V}_{j,x^{\prime}}G_{L}(x,x^{\prime};z)|&\leq \frac{C}{|\text{Im }z|^q}e^{-|\text{Im }z|^p |x-x^{\prime}|}.
\end{aligned}
\end{equation}
\normalsize
\end{corollary}

\subsection{Estimate of the quasi-periodic Green function $G_{\sharp}(x,x^{\prime};z,\kappa)$}
In this section, we estimate the quasi-periodic Green function $G_{\sharp}(x,x^{\prime};z,\kappa)$ ($\kappa\in Y^{*}$), which is defined as the integral kernel of $(\mathcal{L}(\kappa)-z)^{-1}$. Here $\mathcal{L}(\kappa)=\mathcal{L}|_{L_{\kappa}^2}$ is the Floquet transform of $\mathcal{L}$. $G_{\sharp}(x,x^{\prime};z,\kappa)$ can be written as the Floquet transform on $G_{\sharp}(x,x^{\prime};z)$
\begin{equation} \label{eq_floquet_expansion_1}
    G_{\sharp}(x,x^{\prime};z,\kappa)=\sum_{n\in\mathbf{Z}^2}e^{-i\kappa\cdot n}G_{\sharp}(x+n,x^{\prime};z).
\end{equation}
The above series converges pointwisely for $z\notin \mathbf{R}$ by the exponential decay of $G_{\sharp}(x+n,x^{\prime};z)$ in Proposition \ref{prop_G_sharp_decay}. For the same reason, when $x$ and $x^{\prime}$ are restricted in the unit cell, the summation of $n\neq 0$ terms in \eqref{eq_floquet_expansion_1} is uniformly bounded. Hence the local estimate of $G_{\sharp}(x,x^{\prime};z,\kappa)$ is determined by the $n=0$ term in \eqref{eq_floquet_expansion_1}, which is in turn given by Hypothesis \ref{prop_G_sharp_singularities}. More precisely, we have
\begin{proposition} \label{prop_G_sharp_kappa_local_estimate}
The following estimates hold for $G_{\sharp}(x,x^{\prime};z,\kappa)$ and are uniform for $x,x^{\prime}\in Y,\kappa\in Y^{*}$
\begin{equation}
\label{eq_G_sharp_kappa_singularities_1}
|G_{\sharp}(x,x^{\prime};z,\kappa)|\leq \frac{C}{|\text{Im }z|^{q}}\big(1+\log|x-x^{\prime}|\big),
\end{equation}
\begin{equation}
\label{eq_G_sharp_kappa_singularities_2}
|\partial^{\alpha}_{x} G_{\sharp}(x,x^{\prime};z,\kappa)|+|\partial^{\beta}_{x^{\prime}} G_{\sharp}(x,x^{\prime};z,\kappa)|\leq \frac{C}{|\text{Im }z|^{q}}\frac{1}{|x-x^{\prime}|}, \quad \mbox{for}\,\,|\alpha|,|\beta|=1. 
\end{equation}
\end{proposition}
Note that both $\log|x-x^{\prime}|$ and $\frac{1}{|x-x^{\prime}|}$ belongs to $L^p(Y)$ as functions of $x^{\prime}$ for $p<2$. Moreover, they belong to the following space
\begin{equation*}
L^{\infty}(Y;L^{p}(Y)):=\Big\{\text{$f:Y\to L^{p}(Y)$ is measurable with }\, \text{ess}\sup_{x\in Y}\|f(x,\cdot)\|_{L^p(Y)}<\infty   \Big\}.
\end{equation*}
Note that $L^{\infty}(Y;L^{p}(Y))\subset L^{2}(Y;L^{p}(Y))$, where
\begin{equation*}
L^{2}(Y;L^{p}(Y)):=\Big\{\text{$f:Y\to L^{p}(Y)$ is measurable with }\, \Big(\int_{Y}\|f(x,\cdot)\|_{L^p(Y)}^2 dx\Big)^{\frac{1}{2}}<\infty   \Big\}.
\end{equation*}
We adopt this notation for Banach space-valued functions from \cite{conca1995fluids}. For $p>1$, $L^{2}(Y;L^{p}(Y))$ is the dual of $L^{2}(Y;L^{q}(Y))$ where $1/q+1/p=1$ with the natural pairing
\begin{equation} \label{eq_banach_valued_dual_pair}
\langle f,g \rangle:=\int_{Y}\Big(\int_{Y}f(x,x^{\prime})\overline{g(x,x^{\prime})}dx^{\prime}\Big)dx,\quad
f\in L^{2}(Y;L^{p}(Y)),g\in L^{2}(Y;L^{q}(Y)).
\end{equation}
Here and henceforth, we adopt the convention that $x^{\prime}$ denotes the `fiber' variable, and $x$ denotes the `base' variable. Consequently, Proposition \ref{prop_G_sharp_kappa_local_estimate} indicates
\begin{equation*}
\partial^{\alpha} G_{\sharp}(x,x^{\prime};z,\kappa)\in L^{\infty}(Y;L^{p}(Y))\subset L^{2}(Y;L^{p}(Y)). 
\end{equation*}
We also introduce the Sobolev space
\begin{equation*}
H^1(Y;L^2(Y)):=\big\{f\in L^{2}(Y;L^{p}(Y)):\, \partial^{\alpha}_{x}f(x,x^{\prime})\in L^{2}(Y;L^{p}(Y))\quad \forall |\alpha|\leq 1\big\}. 
\end{equation*}
The dual space of $H^1(Y;L^2(Y))$ induced by the pairing \eqref{eq_banach_valued_dual_pair} is denoted as $(H^1)^{*}(Y;L^2(Y))$. Note that $L^{2}(Y;L^{2}(Y))\subset (H^1)^{*}(Y;L^2(Y))$. In conclusion,
\begin{corollary} \label{corol_banach_valued_Gsharp_1}
For $z\notin \mathbf{R}$, $|\alpha|\leq 1$ and $1\leq p\leq 2$,
\begin{equation} \label{eq_banach_valued_Gsharp_1}
\partial^{\alpha} G_{\sharp}(x,x^{\prime};z,\kappa)\in L^{\infty}(Y;L^{p}(Y))\subset L^{2}(Y;L^{p}(Y))\cap (H^1)^{*}(Y;L^2(Y)). 
\end{equation}
\end{corollary}
We need an analogous estimate for $\partial_{z}G_{\sharp}(x,x^{\prime};z,\kappa)$, the integral kernel of the operator $\partial_{z}(\mathcal{L}(\kappa)-z)^{-1}=(\mathcal{L}(\kappa)-z)^{-2}$. By Proposition \ref{prop_composition rule} and \ref{prop_G_sharp_kappa_local_estimate}, for $|\alpha|\leq 1$, $\partial^{\alpha}\partial_{z}G_{\sharp}(x,x^{\prime};z,\kappa)$ can be expanded as 
\begin{equation*}
\int_{Y}\partial^{\alpha}G_{\sharp}(x,x^{\prime\prime};z,\kappa)G_{\sharp}(x^{\prime\prime},x^{\prime};z,\kappa)dx^{\prime\prime}. 
\end{equation*}
Using Proposition \ref{prop_G_sharp_kappa_local_estimate}, the above integral is bounded by
\begin{equation*}
\int_{Y}\frac{1+|\log|x^{\prime\prime}-x^{\prime}||}{|x^{\prime\prime}-x|}dx^{\prime\prime},
\end{equation*}
which is uniformly bounded for $x,x^{\prime}\in Y$. This implies:
\begin{corollary} \label{corol_banach_valued_2}
For $z\notin \mathbf{R}$, $|\alpha|\leq 1$ and $q>2$,
\begin{equation} \label{eq_banach_valued_2}
\partial^{\alpha}\partial_{z} G_{\sharp}(x,x^{\prime};z,\kappa),\partial^{\alpha}\partial_{z} G_{\sharp}(x^{\prime},x;z,\kappa)\in L^{\infty}(Y\times Y)\subset L^{\infty}(Y;L^{p}(Y)).
\end{equation}
\end{corollary}

On the other hand, $\partial^{\alpha}G_{\sharp}(x,x^{\prime};z,\kappa)$ can be expanded as 
\begin{equation} \label{eq_Gsharp_floquet_series}
\partial^{\alpha}G_{\sharp}(x,x^{\prime};z,\kappa)
=\sum_{n\geq 1}\frac{\partial^{\alpha}v_{n}(x;\kappa)\overline{v_{n}(x^{\prime};\kappa)}}{z-\lambda_n(\kappa)},
\end{equation}
where $\{\lambda_n(\kappa),v_{n}(x;\kappa)\}$ are the Bloch eiganpairs of $\mathcal{L}(\kappa)$.
It is straightforward to verify that the above Floquet series \eqref{eq_Gsharp_floquet_series} converges in the space of distributions $\mathcal{D}^{\prime}(Y\times Y)$. We further show that the convergence also holds in  $(H^1)^*(Y,L^2(Y))$, in accordance to Corollary \ref{corol_banach_valued_Gsharp_1}.
To see this, we first recall the following result. 
\begin{proposition}[Corollary 3.2 of \cite{serov2010green}] \label{prop_lambda_n_kappa_estimate}
For any $z\notin\{\lambda_n(\kappa),\, n\geq 1\}$ and $s>1$,
\begin{equation} \label{eq_lambda_n_kappa_estimate}
    \sum_{n\geq 1}\frac{1}{|z-\lambda_{n}(\kappa)|^s}<\infty.
\end{equation}
\end{proposition}
Heuristically, Proposition \ref{prop_lambda_n_kappa_estimate} indicates asymptotics $\lambda_{n}(\kappa)\sim n$, similar to Weyl law. 
This further leads to 

\begin{proposition}\label{prop_Gsharp_floquet_series_convergence}
For $z\notin \mathbf{R}$, $|\alpha|\leq 1$, the Floquet series \eqref{eq_Gsharp_floquet_series} of $\partial^{\alpha}G_{\sharp}(x,x^{\prime};z,\kappa)$ converges in $(H^1)^*(Y,L^2(Y))$.
\end{proposition}

\begin{proof}
We prove Proposition \ref{prop_Gsharp_floquet_series_convergence} by the Cauchy criterion. With the duality $(H^1)^*(Y,L^2(Y))=(H^1(Y,L^2(Y)))^{\prime}$, it's sufficient to prove: for any $\epsilon>0$ there exists $N>0$ such that
\begin{equation} \label{eq_floquet_series_convergence_proof_1}
\sup_{\substack{g\in C_c^{\infty}(Y\times Y) \\ \|g\|_{H^1(Y;L^{2}(Y))}=1}}\Big|\Big\langle \sum_{n=N_1}^{N_2}\frac{\partial^{\alpha}v_{n}(x;\kappa)\overline{v_{n}(x^{\prime};\kappa)}}{z-\lambda_n(\kappa)},g(x,x^{\prime})  \Big\rangle \Big|\leq \epsilon,
\end{equation}
for any $N_2>N_1>N$. Integrating by parts and utilizing the Cauchy-Schwarz inequality yield
\footnotesize
\begin{equation*}
\begin{aligned}
\Big|\Big\langle\sum_{n=N_1}^{N_2}\frac{\partial^{\alpha}v_{n}(x;\kappa)\overline{v_{n}(x^{\prime};\kappa)}}{z-\lambda_n(\kappa)},g(x,x^{\prime})  \Big\rangle \Big|
&=\Big|\Big\langle\sum_{n=N_1}^{N_2}\frac{v_{n}(x;\kappa)\overline{v_{n}(x^{\prime};\kappa)}}{z-\lambda_n(\kappa)},\partial^{\alpha}g(x,x^{\prime})  \Big\rangle \Big| \\
&= \Big|\sum_{n=N_1}^{N_2}\frac{1}{z-\lambda_n(\kappa)}\int_{Y}v_{n}(x;\kappa)
\Big(\int_{Y}\overline{\partial^{\alpha}g(x,x^{\prime})}\cdot \overline{v_{n}(x^{\prime};\kappa)}dx^{\prime} \Big)dx \Big|\\
&\leq \sum_{n=N_1}^{N_2}\frac{1}{|z-\lambda_n(\kappa)|} \Big( \int_{Y}
\Big|\int_{Y}\overline{\partial^{\alpha}g(x,x^{\prime})}\cdot \overline{v_{n}(x^{\prime};\kappa)}dx^{\prime} \Big|^2 dx\Big)^{1/2} \\
&\leq \Big(\sum_{n=N_1}^{N_2}\frac{1}{|z-\lambda_n(\kappa)|^2} \Big)^{1/2}
\Big(\sum_{n=N_1}^{N_2}\int_{Y}
\Big|\int_{Y}\overline{\partial^{\alpha}g(x,x^{\prime})}\cdot \overline{v_{n}(x^{\prime};\kappa)}dx^{\prime} \Big|^2 dx \Big)^{1/2}. 
\end{aligned}
\end{equation*}
\normalsize
By the Parseval identity, 
\footnotesize
\begin{equation*}
\begin{aligned}
\Big(\sum_{n=N_1}^{N_2}\int_{Y}
\Big|\int_{Y}\overline{\partial^{\alpha}g(x,x^{\prime})}\cdot \overline{v_{n}(x^{\prime};\kappa)}dx^{\prime} \Big|^2 dx \Big)^{1/2}
&\leq \Big(\int_{Y} \|\partial^{\alpha}g(x,\cdot)\|_{L^2(Y)}^2 dx \Big)^{1/2}
\leq \|g\|_{H^1(Y;L^{2}(Y))}.
\end{aligned}
\end{equation*}
\normalsize
Hence we conclude that
\begin{equation*}
\begin{aligned}
\Big|\Big\langle\sum_{n=N_1}^{N_2}\frac{\partial^{\alpha}v_{n}(x;\kappa)\overline{v_{n}(x^{\prime};\kappa)}}{z-\lambda_n(\kappa)},g(x,x^{\prime})  \Big\rangle \Big|
&\leq \Big(\sum_{n=N_1}^{N_2}\frac{1}{|z-\lambda_n(\kappa)|^2} \Big)^{1/2}\|g\|_{H^1(Y;L^{2}(Y))}.
\end{aligned}
\end{equation*}
Then \eqref{eq_floquet_series_convergence_proof_1} follows by Proposition \ref{prop_lambda_n_kappa_estimate}. 
\end{proof}

We finally consider the Floquet expansion of $\partial^{\alpha} \partial_z G_{\sharp} (x^{\prime},x;z,\kappa)$:
\begin{equation} \label{eq_Gsharp_z_floquet_series}
\partial^{\alpha} \partial_z G_{\sharp}(x,x^{\prime};z,\kappa)
=-\sum_{n\geq 1}\frac{\partial^{\alpha}v_{n}(x^{\prime};\kappa)\overline{v_{n}(x;\kappa)}}{(z-\lambda_n(\kappa))^2}. 
\end{equation}
We show that the series above converges in the space $L^{2}(Y;L^{p}(Y))$, i.e.,

\begin{proposition}\label{prop_Gsharp_z_floquet_series_convergence}
For $z\notin \mathbf{R}$, $|\alpha|\leq 1$ and $p\in [2,3)$, the Floquet series \eqref{eq_Gsharp_z_floquet_series} of $\partial^{\alpha} \partial_z G_{\sharp}(x^{\prime},x;z,\kappa)$ converges in $L^{2}(Y;L^{p}(Y))$.
\end{proposition}

\begin{proof}
By the duality $L^{2}(Y;L^{p}(Y))=\big(L^{2}(Y;L^{q}(Y)) \big)^{\prime}$ ($1/p+1/q=1$), it is sufficient to show 
\begin{equation} \label{eq_floquet_series_convergence_proof_2}
\sup_{\substack{g\in C_c^{\infty}(Y\times Y) \\ \|g\|_{L^{2}(Y;L^{q}(Y))}=1}}\Big|\Big\langle \sum_{n\geq N}\frac{\partial^{\alpha}v_{n}(x^{\prime};\kappa)\overline{v_{n}(x;\kappa)}}{(z-\lambda_n(\kappa))^2},g(x,x^{\prime})  \Big\rangle \Big|\to 0,
\end{equation}
as $N\to \infty$. Using the Hölder inequality,
\footnotesize
\begin{equation*}
\begin{aligned}
\Big|\Big\langle \sum_{n\geq N}\frac{\partial^{\alpha}v_{n}(x^{\prime};\kappa)\overline{v_{n}(x;\kappa)}}{(z-\lambda_n(\kappa))^2},g(x,x^{\prime})  \Big\rangle \Big|
&\leq \sum_{n\geq N}\frac{1}{|z-\lambda_n(\kappa)|^2}\int_{Y}|v_{n}(x;\kappa)|
\Big(\int_{Y}|g(x,x^{\prime})|\cdot |\partial^{\alpha}v_{n}(x^{\prime};\kappa)|dx^{\prime} \Big)dx \\
&\leq \sum_{n\geq N}\frac{1}{|z-\lambda_n(\kappa)|^2}\int_{Y}|v_{n}(x;\kappa)| \|g(x,\cdot)\|_{L^q(Y)} \|v_{n}(\cdot;\kappa)\|_{W^{1,p}(Y)}dx \\
&\leq \sum_{n\geq N}\frac{1}{|z-\lambda_n(\kappa)|^2}\|g\|_{L^{2}(Y;L^{q}(Y))} \|v_{n}(\cdot;\kappa)\|_{W^{1,p}(Y)}.
\end{aligned}
\end{equation*}
\normalsize
Since $p<3<2+\frac{2}{1-\nu}$ for any $\nu\in (0,1)$, using the estimate in Corollary \ref{corol_eigenfunction_w1p}, we obtain
\begin{equation*}
\begin{aligned}
\Big|\Big\langle \sum_{n\geq N}\frac{\partial^{\alpha}v_{n}(x^{\prime};\kappa)\overline{v_{n}(x;\kappa)}}{(z-\lambda_n(\kappa))^2},g(x,x^{\prime})  \Big\rangle \Big|
&\leq C\|g\|_{L^{2}(Y;L^{q}(Y))}\sum_{n\geq N}\frac{1}{|z-\lambda_n(\kappa)|^2}\lambda_{n}^{2-3/p}(\kappa) \\
&\leq C\|g\|_{L^{2}(Y;L^{q}(Y))}\sum_{n\geq N} \frac{1}{|z-\lambda_n(\kappa)|^{3/p}},
\end{aligned}
\end{equation*}
for any fixed $z\notin\mathbf{R}$. Hence \eqref{eq_floquet_series_convergence_proof_2} follows by Proposition \ref{prop_lambda_n_kappa_estimate}. 
\end{proof}

\section{Representation of bulk index using Green function}
In this section, we prove Theorem \ref{prop_Chern_number_expression} by deriving the representation formula of the gap Chern number using the Green function $G_{\sharp}$. For this purpose, we first present several useful identities. Recall that $\tilde{\mathcal{L}}(\kappa)=-(\nabla+i\kappa)\cdot A(\nabla+i\kappa)$ and $\mathcal{L}(\kappa)=e^{-i\kappa\cdot x}\tilde{\mathcal{L}}(\kappa)e^{i\kappa\cdot x}$,
\footnotesize
\begin{equation} \label{eq_partial_k_L}
\begin{aligned}
\big(\partial_{\kappa_n}\tilde{\mathcal{L}}(\kappa)u,v\big)
&=\Big(\big(-i e_n\cdot A(\nabla+i\kappa)-(\nabla+i\kappa)\cdot A(i e_n) \big)u,v\Big) \\
&=\Big(\big(-i e_n\cdot A\nabla-\nabla\cdot A(i e_n) \big)(e^{i\kappa\cdot x}u),e^{i\kappa\cdot x}v\Big)
=i\big(\mathcal{V}_n(e^{i\kappa\cdot x}u),e^{i\kappa\cdot x}v\big)
\end{aligned}
\end{equation}
\normalsize
for any $u,v\in H^1_{\kappa=0}(Y)$ and $n=1,2$. Since $\tilde{\mathcal{L}}(k)$ is self-adjoint on $L_{\kappa=0}^2(Y)$, we know $\partial_{\kappa_n}\tilde{\mathcal{L}}(k)$ is also self-adjoint. From \eqref{eq_partial_k_L}, we conclude that $\mathcal{V}_n$ is anti-self-adjoint, meaning:
\begin{equation} \label{eq_V_i_anti_selfadj}
\begin{aligned}
\big(\mathcal{V}_n u,v\big)
=-\big(u,\mathcal{V}_n v\big),
\end{aligned}
\end{equation}
for any $u,v\in H^1_{\kappa}(Y)$.
This observation will play an important role in the derivation of the representation formula.

Now we prove Theorem \ref{prop_Chern_number_expression} in four steps.

{\color{blue}Step 1.} We first prove the following expansion
\begin{equation} \label{eq_proof_chern_expression_2}
\begin{aligned}
&\frac{1}{2i}\sum_{n\in F}
\text{Im}\big(\partial_{\kappa_1}u_{n}(\cdot;\kappa),\partial_{\kappa_2}u_{n}(\cdot;\kappa)\big) \\
&=\sum_{j,k\in\{1,2\}}\tilde{\varepsilon}_{jk}\sum_{n\in F}\sum_{m\in E}
\frac{1}{(\lambda_{n}(\kappa)-\lambda_{m}(\kappa))^2}\big(\mathcal{V}_j v_{n}(\cdot;\kappa),v_{m}(\cdot;\kappa)\big)
\big(v_{m}(\cdot;\kappa),\mathcal{V}_k v_{n}(\cdot;\kappa)\big).
\end{aligned}
\end{equation}
Recall that $F=\{n:\,\lambda_n(\kappa)<\lambda_{\text{low}}\}$ denotes the index of filled bands and $E=\{n:\,\lambda_n(\kappa)>\lambda_{\text{upp}}\}=\mathbf{N}\backslash F$ denotes the unfilled bands. To begin, we note by the completeness of $\{u_n(x;\kappa)\}$ in $L^2(Y)$
\begin{equation*}
\begin{aligned}
\frac{1}{2i}\text{Im }\big(\partial_{\kappa_1}u_{n}(\cdot;\kappa),\partial_{\kappa_2}u_{n}(\cdot;\kappa)\big)
&=\sum_{j,k\in\{1,2\}}\tilde{\varepsilon}_{jk}\big(\partial_{\kappa_j}u_{n}(\cdot;\kappa),\partial_{\kappa_k}u_{n}(\cdot;\kappa)\big) \\
&=\sum_{j,k\in\{1,2\}}\tilde{\varepsilon}_{jk}\sum_{m\neq n}\big(\partial_{\kappa_j}u_{n}(\cdot;\kappa),u_{m}(\cdot;\kappa)\big)
\big(u_{m}(\cdot;\kappa),\partial_{\kappa_k}u_{n}(\cdot;\kappa)\big).
\end{aligned}
\end{equation*}
The $m=n$ term above vanishes because 
\footnotesize
\begin{equation*}
\begin{aligned}
\sum_{j,k\in\{1,2\}}\tilde{\varepsilon}_{jk}\big(\partial_{\kappa_j}u_{n}(\cdot;\kappa),u_{n}(\cdot;\kappa)\big)
\big(u_{n}(\cdot;\kappa),\partial_{\kappa_k}u_{n}(\cdot;\kappa)\big)
&=2i\text{Im }\big(\partial_{\kappa_1}u_{n}(\cdot;\kappa),u_{n}(\cdot;\kappa)\big)\big(u_{n}(\cdot;\kappa),\partial_{\kappa_2}u_{n}(\cdot;\kappa)\big)=0,
\end{aligned}
\end{equation*}
\normalsize
where we used the fact that $\text{Re }\big(\partial_{\kappa_j}u_{n}(\cdot;\kappa),u_{n}(\cdot;\kappa)\big)=0$, a consequence of the normalization condition $\|u_{n}(\cdot;\kappa)\|\equiv 1$. By differentiating the identity $\tilde{\mathcal{L}}(\kappa)u_{n}(x;\kappa)=\lambda_{n}(\kappa)u_{n}(x;\kappa)$ and taking inner product with $u_{m}(x;\kappa)$, we find
\begin{equation*}
\begin{aligned}
\big(\partial_{\kappa_j}u_{n}(\cdot;\kappa),u_{m}(\cdot;\kappa)\big)
&=\frac{1}{\lambda_{n}(\kappa)-\lambda_{m}(\kappa)}\big((\partial_{\kappa_j}\tilde{\mathcal{L}}(\kappa))u_{n}(\cdot;\kappa),u_{m}(\cdot;\kappa)\big) \\
&=\frac{i}{\lambda_{n}(\kappa)-\lambda_{m}(\kappa)}\big(\mathcal{V}_j v_{n}(\cdot;\kappa),v_{m}(\cdot;\kappa)\big),
\end{aligned}
\end{equation*}
where $v_{n}(x;\kappa)=e^{i\kappa\cdot x}u_{n}(x;\kappa)$, and \eqref{eq_partial_k_L} is applied. Hence
\begin{equation} \label{eq_proof_chern_expression_1}
\begin{aligned}
&\frac{1}{2i}\sum_{n\in F}
\text{Im}\big(\partial_{\kappa_1}u_{n}(\cdot;\kappa),\partial_{\kappa_2}u_{n}(\cdot;\kappa)\big) \\
&=\sum_{j,k\in\{1,2\}}\tilde{\varepsilon}_{jk}\sum_{n\in F}\sum_{m\neq n}
\frac{1}{(\lambda_{n}(\kappa)-\lambda_{m}(\kappa))^2}\big(\mathcal{V}_j v_{n}(\cdot;\kappa),v_{m}(\cdot;\kappa)\big)
\big(v_{m}(\cdot;\kappa),\mathcal{V}_k v_{n}(\cdot;\kappa)\big).
\end{aligned}
\end{equation}
Note that \eqref{eq_V_i_anti_selfadj} gives
\begin{equation*}
\begin{aligned}
&\sum_{j,k\in\{1,2\}}\tilde{\varepsilon}_{jk}\big(\mathcal{V}_j v_{n}(\cdot;\kappa),v_{m}(\cdot;\kappa)\big)
\big(v_{m}(\cdot;\kappa),\mathcal{V}_k v_{n}(\cdot;\kappa)\big) \\
&=\sum_{j,k\in\{1,2\}}\tilde{\varepsilon}_{jk}\big( v_{n}(\cdot;\kappa),\mathcal{V}_j v_{m}(\cdot;\kappa)\big)
\big(\mathcal{V}_k v_{m}(\cdot;\kappa),v_{n}(\cdot;\kappa)\big) \\
&=-\sum_{j,k\in\{1,2\}}\tilde{\varepsilon}_{jk}\big(\mathcal{V}_j v_{m}(\cdot;\kappa),v_{n}(\cdot;\kappa)\big)
\big(v_{n}(\cdot;\kappa),\mathcal{V}_kv_{m}(\cdot;\kappa)\big).
\end{aligned}
\end{equation*}
This implies the double sum within the filled bands equals zero, i.e.
\begin{equation*}
\sum_{j,k\in\{1,2\}}\tilde{\varepsilon}_{jk}\sum_{n\in F}\sum_{m\in F,\, m\neq n}
\frac{1}{(\lambda_{n}(\kappa)-\lambda_{m}(\kappa))^2}\big(\mathcal{V}_j v_{n}(\cdot;\kappa),v_{m}(\cdot;\kappa)\big)
\big(v_{m}(\cdot;\kappa),\mathcal{V}_kv_{n}(\cdot;\kappa)\big)=0.
\end{equation*}
This identity together with \eqref{eq_proof_chern_expression_1} leads to \eqref{eq_proof_chern_expression_2}.

{\color{blue}Step 2.}
In this step, we prove the following expansion based on \eqref{eq_proof_chern_expression_2}:
\begin{equation} \label{eq_proof_chern_expression_3}
\begin{aligned}
&\frac{1}{2i}\sum_{n\in F}
\text{Im}\big(\partial_{\kappa_1}u_{n}(\cdot;\kappa),\partial_{\kappa_2}u_{n}(\cdot;\kappa)\big) \\
&=\frac{1}{2\pi i}\sum_{j,k\in\{1,2\}}\tilde{\varepsilon}_{jk}\int_{\mathbf{C}}
dm\partial_{\overline{z}}\tilde{g}(z)
\sum_{n,m\geq 1}
\int_{Y}dx\frac{v_n(x;\kappa)\overline{(\mathcal{V}_j v_m)(x;\kappa)}}{z-\lambda_{n}(\kappa)}
\int_{Y}dx^{\prime}\frac{\overline{v_n(x^{\prime};\kappa)}(\mathcal{V}_k v_m)(x^{\prime};\kappa)}{(z-\lambda_{m}(\kappa))^2}.
\end{aligned}
\end{equation}
To do this, we introduce the following interband and intraband transition functions
\footnotesize
\begin{equation*}
\begin{aligned}
&M_{FE}(z;\kappa):=
\sum_{j,k\in\{1,2\}}\tilde{\varepsilon}_{jk}\sum_{n\in F,m\in E}\frac{1}{z-\lambda_{n}(\kappa)}\frac{1}{(z-\lambda_{m}(\kappa))^2}\big(\mathcal{V}_j v_{n}(\cdot;\kappa),v_{m}(\cdot;\kappa)\big)
\big(v_{m}(\cdot;\kappa),\mathcal{V}_k v_{n}(\cdot;\kappa)\big) ,\\
&M_{EF}(z;\kappa):=\sum_{j,k\in\{1,2\}}\tilde{\varepsilon}_{jk}
\sum_{n\in E,m\in F}\frac{1}{z-\lambda_{n}(\kappa)}\frac{1}{(z-\lambda_{m}(\kappa))^2}\big(\mathcal{V}_j v_{n}(\cdot;\kappa),v_{m}(\cdot;\kappa)\big)
\big(v_{m}(\cdot;\kappa),\mathcal{V}_k v_{n}(\cdot;\kappa)\big) ,\\
&M_{EE}(z;\kappa):=\sum_{j,k\in\{1,2\}}\tilde{\varepsilon}_{jk}
\sum_{n\in E,m\in E}\frac{1}{z-\lambda_{n}(\kappa)}\frac{1}{(z-\lambda_{m}(\kappa))^2}\big(\mathcal{V}_j v_{n}(\cdot;\kappa),v_{m}(\cdot;\kappa)\big)
\big(v_{m}(\cdot;\kappa),\mathcal{V}_k v_{n}(\cdot;\kappa)\big) ,\\
&M_{FF}(z;\kappa):=\sum_{j,k\in\{1,2\}}\tilde{\varepsilon}_{jk}
\sum_{n\in F,m\in F}\frac{1}{z-\lambda_{n}(\kappa)}\frac{1}{(z-\lambda_{m}(\kappa))^2}\big(\mathcal{V}_j v_{n}(\cdot;\kappa),v_{m}(\cdot;\kappa)\big)
\big(v_{m}(\cdot;\kappa),\mathcal{V}_k v_{n}(\cdot;\kappa)\big).
\end{aligned}
\end{equation*}
\normalsize
Physically, the in-gap energy excitation in photonic systems described by the gap Chern number in our case results from the interband transition of modes \cite{winn1999interband}. This is indeed the case as is shown in Lemma \ref{lem_four_mero_functions} at the end of this section. By Lemma \ref{lem_four_mero_functions}, 
\footnotesize
\begin{equation*}
\begin{aligned}
&\frac{1}{2i}\sum_{n\in F}
\text{Im}\big(\partial_{\kappa_1}u_{n}(\cdot;\kappa),\partial_{\kappa_2}u_{n}(\cdot;\kappa)\big) \\
&=\frac{1}{2\pi i}\int_{\mathbf{C}}
dm\partial_{\overline{z}}\tilde{g}(z)\big(M_{FE}(z;\kappa)+M_{EF}(z;\kappa)+M_{FF}(z;\kappa)+M_{EE}(z;\kappa) \big) \\
&=\frac{1}{2\pi i}\sum_{j,k\in\{1,2\}}\tilde{\varepsilon}_{jk}\int_{\mathbf{C}}
dm\partial_{\overline{z}}\tilde{g}(z)
\sum_{n,m\geq 1}\frac{1}{z-\lambda_{n}(\kappa)}\frac{1}{(z-\lambda_{m}(\kappa))^2}\big(\mathcal{V}_j v_{n}(\cdot;\kappa),v_{m}(\cdot;\kappa)\big)
\big(v_{m}(\cdot;\kappa),\mathcal{V}_k v_{n}(\cdot;\kappa)\big)\\
&=\frac{1}{2\pi i}\sum_{j,k\in\{1,2\}}\tilde{\varepsilon}_{jk}\int_{\mathbf{C}}
dm\partial_{\overline{z}}\tilde{g}(z)
\sum_{n,m\geq 1}\frac{1}{z-\lambda_{n}(\kappa)}\frac{1}{(z-\lambda_{m}(\kappa))^2}\big( v_{n}(\cdot;\kappa),\mathcal{V}_j v_{m}(\cdot;\kappa)\big)
\big(\mathcal{V}_k v_{m}(\cdot;\kappa),v_{n}(\cdot;\kappa)\big) \\
&=\frac{1}{2\pi i}\sum_{j,k\in\{1,2\}}\tilde{\varepsilon}_{jk}\int_{\mathbf{C}}
dm\partial_{\overline{z}}\tilde{g}(z)
\sum_{n,m\geq 1}
\int_{Y}dx\frac{v_n(x;\kappa)\overline{(\mathcal{V}_j v_m)(x;\kappa)}}{z-\lambda_{n}(\kappa)}
\int_{Y}dx^{\prime}\frac{\overline{v_n(x^{\prime};\kappa)}(\mathcal{V}_k v_m)(x^{\prime};\kappa)}{(z-\lambda_{m}(\kappa))^2}, 
\end{aligned}
\end{equation*}
\normalsize
where \eqref{eq_V_i_anti_selfadj} is applied to derive the third equality. This gives \eqref{eq_proof_chern_expression_3}.

{\color{blue}Step 3.} We interchange the integrals and summations in \eqref{eq_proof_chern_expression_3} to show that
\begin{equation} \label{eq_proof_chern_expression_3_7}
\begin{aligned}
&\frac{1}{2i}\sum_{n\in F}
\text{Im}\big(\partial_{\kappa_1}u_{n}(\cdot;\kappa),\partial_{\kappa_2}u_{n}(\cdot;\kappa)\big) \\
&=\frac{1}{2\pi i}\sum_{j,k\in\{1,2\}}\tilde{\varepsilon}_{jk}\int_{\mathbf{C}}
dm\partial_{\overline{z}}\tilde{g}(z)
\int_{Y}dx\int_{Y}dx^{\prime}(\mathcal{V}_j G_{\sharp})(x,x^\prime;z,\kappa)(\mathcal{V}_k \partial_z G_{\sharp})(x^\prime,x;z,\kappa).
\end{aligned}
\end{equation}
We proceed in two steps.

{\color{blue}Step 3.1} The first step is to interchange the spatial integral and $n$-summation of \eqref{eq_proof_chern_expression_3}. To begin, by a integration by part in $x$,  \eqref{eq_proof_chern_expression_3} leads to
\begin{equation} \label{eq_proof_chern_expression_7}
\begin{aligned}
&\frac{1}{2i}\sum_{n\in F}
\text{Im}\big(\partial_{\kappa_1}u_{n}(\cdot;\kappa),\partial_{\kappa_2}u_{n}(\cdot;\kappa)\big) \\
&=-\frac{1}{2\pi i}\sum_{j,k\in\{1,2\}}\tilde{\varepsilon}_{jk}\int_{\mathbf{C}}
dm\partial_{\overline{z}}\tilde{g}(z)
\sum_{n,m\geq 1}
\int_{Y\times Y}dxdx^{\prime}\frac{\mathcal{V}_j v_n(x;\kappa)\overline{v_n(x^{\prime};\kappa)}}{z-\lambda_{n}(\kappa)}\frac{(\mathcal{V}_k v_m)(x^{\prime};\kappa)\overline{v_m(x;\kappa)}}{(z-\lambda_{m}(\kappa))^2}.
\end{aligned}
\end{equation}
Note that for any fixed $m\geq 1$ and $z\notin\mathbf{R}$, $$\frac{(\mathcal{V}_k v_m)(x^{\prime};\kappa)\overline{v_m(x;\kappa)}}{(z-\lambda_{m}(\kappa))^2}\in H^1(Y,L^2(Y))
$$ by the fact that $v_m\in H^1(Y)$. On the other hand, $\sum_{n\geq 1}\frac{\mathcal{V}_j v_n(x;\kappa)\overline{v_n(x^{\prime};\kappa)}}{z-\lambda_{n}(\kappa)}$ converges in the dual space $(H^1)^{*}(Y,L^2(Y))$ by Proposition \ref{prop_Gsharp_floquet_series_convergence}. These imply
\footnotesize
\begin{equation*}
\begin{aligned}
&\Big|\int_{Y\times Y}dxdx^{\prime}\sum_{n\geq N}\frac{\mathcal{V}_j v_n(x;\kappa)\overline{v_n(x^{\prime};\kappa)}}{z-\lambda_{n}(\kappa)}\frac{(\mathcal{V}_k v_m)(x^{\prime};\kappa)\overline{v_m(x;\kappa)}}{(z-\lambda_{m}(\kappa))^2}\Big| \\
&\leq \big\|\sum_{n\geq N}\frac{\mathcal{V}_j v_n(x;\kappa)\overline{v_n(x^{\prime};\kappa)}}{z-\lambda_{n}(\kappa)}\big\|_{(H^1)^{*}(Y,L^2(Y))}
\big\|\frac{(\mathcal{V}_k v_m)(x^{\prime};\kappa)\overline{v_m(x;\kappa)}}{(z-\lambda_{m}(\kappa))^2}\big\|_{H^1(Y,L^2(Y))} \to 0
\end{aligned}
\end{equation*}
\normalsize
as $N\to\infty$. Hence we can interchange the spatial integral and $n$-summation of \eqref{eq_proof_chern_expression_7} to obtain
\begin{equation} \label{eq_proof_chern_expression_8}
\begin{aligned}
&\frac{1}{2i}\sum_{n\in F}
\text{Im}\big(\partial_{\kappa_1}u_{n}(\cdot;\kappa),\partial_{\kappa_2}u_{n}(\cdot;\kappa)\big) \\
&=-\frac{1}{2\pi i}\sum_{j,k\in\{1,2\}}\tilde{\varepsilon}_{jk}\int_{\mathbf{C}}
dm\partial_{\overline{z}}\tilde{g}(z)
\sum_{m\geq 1}
\int_{Y\times Y}dxdx^{\prime}\sum_{n\geq 1}\frac{\mathcal{V}_j v_n(x;\kappa)\overline{v_n(x^{\prime};\kappa)}}{z-\lambda_{n}(\kappa)}\frac{(\mathcal{V}_k v_m)(x^{\prime};\kappa)\overline{v_m(x;\kappa)}}{(z-\lambda_{m}(\kappa))^2} \\
&=-\frac{1}{2\pi i}\sum_{j,k\in\{1,2\}}\tilde{\varepsilon}_{jk}\int_{\mathbf{C}}
dm\partial_{\overline{z}}\tilde{g}(z)
\sum_{m\geq 1}
\int_{Y\times Y}dxdx^{\prime}\mathcal{V}_j G_{\sharp}(x,x^{\prime};z,\kappa)\frac{(\mathcal{V}_k v_m)(x^{\prime};\kappa)\overline{v_m(x;\kappa)}}{(z-\lambda_{m}(\kappa))^2}.
\end{aligned}
\end{equation}

{\color{blue}Step 3.2} We next interchange the spatial integral and $m$-summation of \eqref{eq_proof_chern_expression_8} to conclude the proof of \eqref{eq_proof_chern_expression_3_7}. Indeed, by Proposition \ref{prop_Gsharp_z_floquet_series_convergence}, there exists $p>2$ such that 
$$-\sum_{m\geq 1}\frac{(\mathcal{V}_k v_m)(x^{\prime};\kappa)\overline{v_m(x;\kappa)}}{(z-\lambda_{m}(\kappa))^2}\longrightarrow (\mathcal{V}_k \partial_z G_{\sharp})(x^\prime,x;z,\kappa) \quad\text{in $L^2(Y,L^p(Y))$}.
$$
In addition, $(\mathcal{V}_j G_{\sharp})(x,x^\prime;z,\kappa)\in L^2(Y,L^q(Y))$ for $1/q+1/p=1$ by Corollary \ref{corol_banach_valued_Gsharp_1}. By a similar argument as in Step 3.1, we have
\begin{equation} \label{eq_proof_chern_expression_9}
\begin{aligned}
&\frac{1}{2i}\sum_{n\in F}
\text{Im}\big(\partial_{\kappa_1}u_{n}(\cdot;\kappa),\partial_{\kappa_2}u_{n}(\cdot;\kappa)\big) \\
&=-\frac{1}{2\pi i}\sum_{j,k\in\{1,2\}}\tilde{\varepsilon}_{jk}\int_{\mathbf{C}}
dm\partial_{\overline{z}}\tilde{g}(z)
\int_{Y\times Y}dxdx^{\prime}\mathcal{V}_j G_{\sharp}(x,x^{\prime};z,\kappa)\sum_{m\geq 1}\frac{(\mathcal{V}_k v_m)(x^{\prime};\kappa)\overline{v_m(x;\kappa)}}{(z-\lambda_{m}(\kappa))^2} \\
&=\frac{1}{2\pi i}\sum_{j,k\in\{1,2\}}\tilde{\varepsilon}_{jk}\int_{\mathbf{C}}
dm\partial_{\overline{z}}\tilde{g}(z)
\int_{Y}dx\int_{Y}dx^{\prime}(\mathcal{V}_j G_{\sharp})(x,x^\prime;z,\kappa)(\mathcal{V}_k \partial_z G_{\sharp})(x^\prime,x;z,\kappa). 
\end{aligned}
\end{equation}

{\color{blue}Step 4.} Finally, we use \eqref{eq_proof_chern_expression_3_7} to calculate the gap Chern number to conclude the proof of the theorem. By applying the Floquet transform in $x^{\prime}$, we obtain 
{\footnotesize
\begin{equation*} 
\begin{aligned}
\mathcal{C}_{\Delta}
&=\frac{-1}{4\pi}\int_{\mathbf{T}^2}d\kappa\sum_{n\in F}
\text{Im}\big(\partial_{\kappa_1}u_{n}(\cdot;\kappa),\partial_{\kappa_2}u_{n}(\cdot;\kappa)\big) \\
&=-\frac{1}{4\pi^2}\sum_{i,j\in\{1,2\}}\tilde{\varepsilon}_{ij}\int_{\mathbf{T}^2}d\kappa\int_{\mathbf{C}}
dm\partial_{\overline{z}}\tilde{g}(z)
\int_{Y}dx\int_{Y}dx^{\prime}(\mathcal{V}_i G_{\sharp})(x,x^\prime;z,\kappa)(\mathcal{V}_j \partial_z G_{\sharp})(x^\prime,x;z,\kappa).
\end{aligned}
\end{equation*}
}
By interchanging the integral in $dz$  and $d\kappa$, which is justified by the Fubini theorem and the uniform boundedness of the integral $$\partial_{\overline{z}}\tilde{g}(z)
\int_{Y\times Y}dxdx^{\prime}\big|(\mathcal{V}_i G_{\sharp})(x,x^\prime;z,\kappa)(\mathcal{V}_j \partial_z G_{\sharp})(x^\prime,x;z,\kappa)\big|$$ as shown in 
Proposition \ref{prop_interchange_dkappa_dz} in Subsection \ref{sec-42}, we further get 
\begin{equation*}
\begin{aligned}
\mathcal{C}_{\Delta}=-\frac{1}{4\pi^2}\sum_{i,j\in\{1,2\}}\tilde{\varepsilon}_{ij}\int_{\mathbf{C}}
dm\partial_{\overline{z}}\tilde{g}(z)\int_{\mathbf{T}^2}d\kappa
\int_{Y}dx\int_{Y}dx^{\prime}(\mathcal{V}_i G_{\sharp})(x,x^\prime;z,\kappa)(\mathcal{V}_j \partial_z G_{\sharp})(x^\prime,x;z,\kappa) .
\end{aligned}
\end{equation*}
The Parseval identity for Floquet transform in $x^{\prime}$ variable ($\int_{\mathbf{R}^2}u(x^{\prime})\overline{v(x^{\prime})}dx^{\prime}=\int_{\mathbf{T}^2}d\kappa\int_{Y}u(x^{\prime};\kappa)\overline{v(x^{\prime};\kappa)}dx^{\prime}$) then yields
\begin{equation} \label{eq_proof_chern_expression_5}
\begin{aligned}
\mathcal{C}_{\Delta}=-\sum_{i,j\in\{1,2\}}\tilde{\varepsilon}_{ij}\int_{\mathbf{C}}
dm\partial_{\overline{z}}\tilde{g}(z)
\int_{Y\times \mathbf{R}^2}dxdx^{\prime}
(\mathcal{V}_i G_{\sharp})(x,x^\prime;z)(\mathcal{V}_j \partial_{z}G_{\sharp})(x^\prime,x;z).
\end{aligned}
\end{equation}
Since the structure is invariant under $\mathbf{Z}^2-$translation, 
\begin{equation*}
G_{\sharp}(x+e,x^\prime+e;z)=G_{\sharp}(x,x^\prime;z),\quad \forall e\in\mathbf{Z}^2.
\end{equation*}
This implies that the $dx-$integral in \eqref{eq_proof_chern_expression_5} can be performed over any cell $Y_{e}:=Y+e$ without changing the result. Hence we can replace the integral in $Y$ by the average over $\Omega_L$:
\begin{equation*} \label{eq_proof_chern_expression_6}
\begin{aligned}
\mathcal{C}_{\Delta}
&=-\frac{1}{\pi}\sum_{i,j\in\{1,2\}}\tilde{\varepsilon}_{ij}\lim_{L\to \infty}\frac{1}{|\Omega_L|}\int_{\mathbf{C}}
dm\partial_{\overline{z}}\tilde{g}(z)
\int_{\Omega_L\times \mathbf{R}^2}dxdx^{\prime}
(\mathcal{V}_i G_{\sharp})(x,x^\prime;z)(\mathcal{V}_j\partial_{z}G_{\sharp})(x^\prime,x;z).
\end{aligned}
\end{equation*}
Using Proposition \ref{prop_interchange_dkappa_dz} again, the dominated convergence theorem yields
\begin{equation*}
\begin{aligned}
\mathcal{C}_{\Delta}
&=-\frac{1}{\pi}\sum_{i,j\in\{1,2\}}\tilde{\varepsilon}_{ij}\lim_{L\to \infty}\frac{1}{|\Omega_L|}\int_{\mathbf{C}}
dm\partial_{\overline{z}}\tilde{g}(z)
\int_{\Omega_L\times \Omega_L}dxdx^{\prime}
(\mathcal{V}_i G_{\sharp})(x,x^\prime;z)(\mathcal{V}_j \partial_{z}G_{\sharp})(x^\prime,x;z).
\end{aligned}
\end{equation*}
This establishes \eqref{eq_Chern_number_expression}.

\subsection{Band transition functions}
In this section, we present the basic properties of the interband and intraband transition functions introduced in Step 2 in the proof of Theorem \ref{prop_Chern_number_expression}.

\begin{lemma} \label{lem_four_mero_functions}
$M_{FE}(z;\kappa),M_{EF}(z;\kappa),M_{EE}(z;\kappa),M_{FF}(z;\kappa)$ are well-defined and meromorphic. Moreover,
\begin{equation} \label{eq_four_mero_functions_1}
\begin{aligned}
&\frac{1}{\pi i}\int_{\mathbf{C}}dm\partial_{\overline{z}}\tilde{g}(z)M_{FE}(z;\kappa)
=\frac{1}{\pi i}\int_{\mathbf{C}}dm\partial_{\overline{z}}\tilde{g}(z)M_{EF}(z;\kappa) \\
&=\sum_{j,k\in\{1,2\}}\tilde{\varepsilon}_{jk}\sum_{n\in F}\sum_{m\in E}
\frac{1}{(\lambda_{n}(\kappa)-\lambda_{m}(\kappa))^2}\big(\mathcal{V}_j v_{n}(\cdot;\kappa),v_{m}(\cdot;\kappa)\big)
\big(v_{m}(\cdot;\kappa),\mathcal{V}_k v_{n}(\cdot;\kappa)\big),
\end{aligned}
\end{equation}
and
\begin{equation} \label{eq_four_mero_functions_2}
\frac{1}{\pi i}\int_{\mathbf{C}}dm\partial_{\overline{z}}\tilde{g}(z)M_{FF}(z;\kappa)
=\frac{1}{\pi i}\int_{\mathbf{C}}dm\partial_{\overline{z}}\tilde{g}(z)M_{EE}(z;\kappa)
=0.
\end{equation}
\end{lemma}

\begin{proof}
{\color{blue}Step 1.} We show the unfilled-unfilled transition function $M_{EE}(z;\kappa)$ is a well-defined meromorphic function. This function involves two infinite sums, requiring careful handling. The proofs for the other transition functions follow a similar approach.

To establish this, it is sufficient to show that the double sum defining $M_{EE}(z; \kappa)$ converges absolutely and uniformly on any compact region $K$ in the complex plane that excludes the eigenvalues $\lambda_n(\kappa)$. Using the Cauchy-Schwarz inequality, the analysis reduces to proving the uniform convergence of
\begin{equation} \label{eq_mero_1}
    \sum_{n,m\geq 1}\frac{1}{|z-\lambda_{n}(\kappa)|}\frac{1}{|z-\lambda_{m}(\kappa)|^2}|\big(\mathcal{V}_i v_{n}(\cdot;\kappa),v_{m}(\cdot;\kappa)\big)|^2
    \quad i=1,2.
\end{equation}
Since $z$ is fixed within the compact region $K$ and $\lambda_n(\kappa) \to \infty$ as $n \to \infty$, we can establish the following bound for any $\mu > 0$:
\begin{equation} \label{eq_mero_2}
\begin{aligned}
&\sum_{n,m\geq 1}\frac{1}{|z-\lambda_{n}(\kappa)|}\frac{1}{|z-\lambda_{m}(\kappa)|^2}|\big(\mathcal{V}_i v_{n}(\cdot;\kappa),v_{m}(\cdot;\kappa)\big)|^2 \\
&\leq C^2_{\mu}\sum_{n,m\geq 1}\frac{1}{\mu+\lambda_{n}(\kappa)}\frac{1}{(\mu+\lambda_{m}(\kappa))^2}|\big(\mathcal{V}_i v_{n}(\cdot;\kappa),v_{m}(\cdot;\kappa)\big)|^2 ,
\end{aligned}
\end{equation}
where $C_{\mu}=\max\{\sup_{z\in K,\, n\geq 1}\frac{|\lambda_n(\kappa)-z|}{\lambda_n(\kappa)+\mu},1 \}<\infty$ depends only on $\mu$ and the range of $z$. By repeatedly applying the Cauchy-Schwarz inequality, we find that 
\footnotesize
\begin{equation*}
\begin{aligned}
&\sum_{n,m\geq 1}\frac{1}{|z-\lambda_{n}(\kappa)|}\frac{1}{|z-\lambda_{m}(\kappa)|^2}|\big(\mathcal{V}_i v_{n}(\cdot;\kappa),v_{m}(\cdot;\kappa)\big)|^2 \\
&\leq C_{\mu}^2\sum_{m\geq 1}\frac{1}{(\mu+\lambda_{m}(\kappa))^2}
\Big(\sum_{n\geq 1}\big|( v_{n}(\cdot;\kappa),\mathcal{V}_i v_{m}(\cdot;\kappa)\big)|^2\Big)^{\frac{1}{2}}
\Big(\sum_{n\geq 1}\frac{\big|(v_{n}(\cdot;\kappa),\mathcal{V}_i v_{m}(\cdot;\kappa)\big)|^2}{(\mu+\lambda_{n}(\kappa))^2}\Big)^{\frac{1}{2}} \\
&\leq C_{\mu}^2\sum_{m\geq 1}\frac{1}{(\mu+\lambda_{m}(\kappa))^2}
\Big(\sum_{n\geq 1}\big|( v_{n}(\cdot;\kappa),\mathcal{V}_i v_{m}(\cdot;\kappa)\big)|^2\Big)^{\frac{3}{4}}
\Big(\sum_{n\geq 1}\frac{\big|(\mathcal{V}_i v_{n}(\cdot;\kappa),v_{m}(\cdot;\kappa)\big)|^2}{(\mu+\lambda_{n}(\kappa))^4}\Big)^{\frac{1}{4}} \\
&\leq C_{\mu}^2\sum_{m\geq 1}\frac{\|\mathcal{V}_i v_{m}(\cdot;\kappa)\|^{\frac{3}{2}}_{L^2(Y)}}{(\mu+\lambda_{m}(\kappa))^2}
\Big(\sum_{n\geq 1}\frac{\|\mathcal{V}_i v_{n}(\cdot;\kappa)\|^2_{L^2(Y)}}{(\mu+\lambda_{n}(\kappa))^4}\Big)^{\frac{1}{4}},
\end{aligned}
\end{equation*}
\normalsize
where the Parseval identity $\sum_{n\geq 1}\big|( v_{n}(\cdot;\kappa),\mathcal{V}_i v_{m}(\cdot;\kappa)\big)|^2=\|\mathcal{V}_i v_{m}(\cdot;\kappa)\|^2$ is used in the last inequality above. Note that $\|\mathcal{V}_i v_m(x;\kappa)\|_{L^2(Y)}\lesssim \| v_m(x;\kappa)\|_{H^1(Y)}=\mathcal{O}(\lambda_{m}^{\frac{1}{2}}(\kappa))$ by Corollary \ref{corol_eigenfunction_w1p}. Hence
\footnotesize
\begin{equation*}
\begin{aligned}
\sum_{n,m\geq 1}\frac{1}{|z-\lambda_{n}(\kappa)|}\frac{1}{|z-\lambda_{m}(\kappa)|^2}|\big(\mathcal{V}_i v_{n}(\cdot;\kappa),v_{m}(\cdot;\kappa)\big)|^2 
\lesssim C_{\mu}^2\sum_{m\geq 1}\frac{1}{(\mu+\lambda_{m}(\kappa))^{\frac{5}{4}}}
\Big(\sum_{n\geq 1}\frac{1}{(\mu+\lambda_{n}(\kappa))^3}\Big)^{\frac{1}{4}}. 
\end{aligned}
\end{equation*}
\normalsize
The right side above is finite by Proposition \ref{prop_lambda_n_kappa_estimate}. This concludes the uniform convergence of \eqref{eq_mero_1}.

{\color{blue}Step 2.} \eqref{eq_four_mero_functions_1} and \eqref{eq_four_mero_functions_2} are derived by identifying the poles of the meromorphic transition functions and applying Corollary \ref{corol_residue_formula_for_g}. Consider, for example, the filled-unfilled transition function $M_{FE}(z;\kappa)$. Within the support of $\tilde{g}$, $M_{FE}(z;\kappa)$ has simple poles at $\lambda_{n}(\kappa)$ ($n\in F$) with coefficient
\begin{equation*}
\sum_{i,j\in\{1,2\}}\tilde{\varepsilon}_{ij}\sum_{m\in E}
\frac{1}{(\lambda_{n}(\kappa)-\lambda_{m}(\kappa))^2}\big(\mathcal{V}_i v_{n}(\cdot;\kappa),v_{m}(\cdot;\kappa)\big)
\big(v_{m}(\cdot;\kappa),\mathcal{V}_j v_{n}(\cdot;\kappa)\big)
\end{equation*}
and no other poles. Combining this with Corollary \ref{corol_residue_formula_for_g} establishes \eqref{eq_four_mero_functions_1} for $M_{FE}$. The remaining three identities are derived in a similar manner by identifying their poles and using the following identity:
\begin{equation*}
\frac{1}{z-\lambda_n}\frac{1}{(z-\lambda_m)^2}=\frac{1}{(z-\lambda_m)^2}\frac{1}{\lambda_m-\lambda_n}+\frac{1}{(\lambda_n-\lambda_m)^2}\Big[\frac{1}{z-\lambda_n}-\frac{1}{z-\lambda_m}\Big]\quad \text{for $\lambda_n\neq\lambda_m$}. 
\end{equation*}
\end{proof}

\subsection{Uniform boundedness of integrals involving Green functions} \label{sec-42}
In this section, we show the uniform boundedness of the integrals involving Green functions in Step 4 in the proof of 
Theorem \ref{prop_Chern_number_expression}.

\begin{proposition} \label{prop_interchange_dkappa_dz}
For $i,j\in\{1,2\}$, the following integral is uniformly bounded for $z\in\text{supp }\partial_{\overline{z}}\tilde{g}$ and $\kappa\in\mathbf{T}^2$
\begin{equation*}
\partial_{\overline{z}}\tilde{g}(z)
\int_{Y\times Y}dxdx^{\prime}\big|(\mathcal{V}_i G_{\sharp})(x,x^\prime;z,\kappa)(\mathcal{V}_j \partial_z G_{\sharp})(x^\prime,x;z,\kappa)\big|. 
\end{equation*}
On the other hand, the following integral is uniformly bounded and converges to zero as $L\to\infty$ for $z\in\text{supp }\partial_{\overline{z}}\tilde{g}\,\backslash \mathbf{R}$
\begin{equation*}
\frac{1}{|\Omega_L|}\cdot \partial_{\overline{z}}\tilde{g}(z)\int_{\Omega_L\times \mathbf{R}^2\backslash \Omega_L}dxdx^{\prime}\big|
(\mathcal{V}_iG_{\sharp})(x,x^\prime;z)(\mathcal{V}_j\partial_{z}G_{\sharp})(x^\prime,x;z)\big|.
\end{equation*}
\end{proposition}
We omit the proof of Proposition \ref{prop_interchange_dkappa_dz}. These estimates are justified by a similar calculation as in the proof of Proposition \ref{prop_trace_approximation} using Proposition \ref{prop_G_sharp_kappa_local_estimate}.

\section{Representation of edge index using Green function}
In this section, we prove Theorem \ref{prop_edge_index_expression} by deriving the representation of edge index using the Green function $G_L(x, x';z)$. 

\subsection{Preliminaries}
We start by preparing two preliminary results. The first is on the following trace-class estimate.
\begin{proposition}\label{prop_uniform_trace_interchange_z_trace}
For $z\in\text{supp }\tilde{g}\backslash \mathbf{R}$, $(\mathcal{V}_2x_1-\mathcal{V}_1x_2)\partial_{z}R(z,\mathcal{L}_{\Omega_L})\in S_1(L^2(\Omega_L))$. Moreover,
\begin{equation} \label{eq_uniform_trace_interchange_z_trace}
\int_{\mathbf{C}}dm(z)|\partial_{\overline{z}}\tilde{g}(z)|\int_{\Omega_L}\Big|\big((\mathcal{V}_2x_1-\mathcal{V}_1x_2)\partial_{z}R(z,\mathcal{L}_{\Omega_L})\big)(x,x)\Big|dx <\infty .
\end{equation}
\end{proposition}
\begin{proof}
Fix $L>1$ in the following paragraphs. We prove that
\begin{equation} \label{eq_interchange_z_trace_proof_1}
\Big|\big((\mathcal{V}_2x_1-\mathcal{V}_1x_2)\partial_{z}R(z,\mathcal{L}_{\Omega_L})\big)(x,x^{\prime})\Big|\lesssim \frac{1}{|\text{Im }z|^{2q}}
\end{equation}
uniformly for $x,x^{\prime}\in \Omega_L$. Then $(\mathcal{V}_2x_1-\mathcal{V}_1x_2)\partial_{z}R(z,\mathcal{L}_{\Omega_L})\in S_1(L^2(\Omega_L))$ follows by Proposition \ref{prop_trace_class_kernel}. \eqref{eq_interchange_z_trace_proof_1} also leads to \eqref{eq_uniform_trace_interchange_z_trace} as a consequence of $\partial_{\overline{z}}\tilde{g}=\mathcal{O}(|\text{Im }z|^{\infty})$.

Note that $\mathcal{V}_i x_j$ is first-order differential operators with uniformly bounded coefficients when the domain size is fixed. Hence
\footnotesize
\begin{equation*}
\begin{aligned}
\Big|\big((\mathcal{V}_2x_1-\mathcal{V}_1x_2)\partial_{z}R(z,\mathcal{L}_{\Omega_L})\big)(x,x^{\prime})\Big|
&\lesssim \sum_{i=1,2}\Big|\big(\mathcal{V}_i \partial_{z}R(z,\mathcal{L}_{\Omega_L})\big)(x,x^{\prime})\Big|
+\Big|\big(\partial_{z}R(z,\mathcal{L}_{\Omega_L})\big)(x,x^{\prime})\Big| \\
&\leq \sum_{i=1,2}\int_{\Omega_L}dx^{\prime\prime}|\mathcal{V}_i G_{L}(x,x^{\prime\prime};z)||G_{L}(x^{\prime\prime},x^{\prime};z)|  \\
&\quad +\int_{\Omega_L}dx^{\prime\prime}|G_{L}(x,x^{\prime\prime};z)||G_{L}(x^{\prime\prime},x^{\prime};z)|,
\end{aligned}
\end{equation*}
\normalsize
where in the second inequality above we used the fact that the integral kernel of $\partial_{z}R(z,\mathcal{L}_{\Omega_L})=-R(z,\mathcal{L}_{\Omega_L})R(z,\mathcal{L}_{\Omega_L})$ can be expressed as the composition of the Green function $G_L$ with itself by Proposition \ref{prop_composition rule}. By Corollary \ref{corol_estimates_Vi_G}, the first integral at the right side is estimated as follows
\footnotesize
\begin{equation*}
\begin{aligned}
\sum_{i=1,2}\int_{\Omega_L}dx^{\prime\prime}|\mathcal{V}_{i} G_{L}(x,x^{\prime\prime};z)||G_{L}(x^{\prime\prime},x^{\prime};z)|  
&\lesssim \frac{1}{|\text{Im }z|^{2q}}\int_{\Omega_L}dx^{\prime\prime}\frac{1+|\log|x^{\prime\prime}-x^{\prime}||}{|x^{\prime\prime}-x|} 
\lesssim \frac{1}{|\text{Im }z|^{2q}}.
\end{aligned}
\end{equation*}
\normalsize
The second integral is estimated in the same way. This gives \eqref{eq_interchange_z_trace_proof_1}.
\end{proof}

The second is on the following kernel estimation. 
\begin{lemma} \label{lem_composition_integral_operators}
For $P\in\{R(z,\mathcal{L}),\mathcal{V}_iR(z,\mathcal{L})\}$ and $Q\in\{ R(z,\mathcal{L}_{\Omega_L}),\mathcal{V}_i R(z,\mathcal{L}_{\Omega_L})\}$ as integral operators on $\mathbf{R}^2$ and $\Omega_L$, respectively, their associated integral kernels satisfy \eqref{eq_composition_rule_cond_1} when $z\in \text{supp }\tilde{g}\backslash \mathbf{R}$.
\end{lemma}
\begin{proof}
We prove \eqref{eq_composition_rule_cond_1} for $R(z,\mathcal{L})$ and $\mathcal{V}_i R(z,\mathcal{L})$. The proof of $R(z,\mathcal{L}_{\Omega_L})$ and $\mathcal{V}_i R(z,\mathcal{L}_{\Omega_L})$ are similar.

Recall that the integral kernel of $R(z,\mathcal{L})$ is the Green function $G_{\sharp}(x,x^{\prime};z)$. Proposition \ref{prop_G_sharp_singularities} and \ref{prop_G_sharp_decay} yield that
\footnotesize
\begin{equation*}
\begin{aligned}
\int_{\mathbf{R}^2}\big|\big(R(z,\mathcal{L})\big)(x,x^{\prime})\big|dx^{\prime}
&=\int_{\mathbf{R}^2\cap B_{x,1/2}}|G_{\sharp}(x,x^{\prime};z)|dx^{\prime}
+\int_{\mathbf{R}^2\backslash B_{x,1/2}}|G_{\sharp}(x,x^{\prime};z)|dx^{\prime} \\
&\lesssim
\frac{1}{|\text{Im }z|^q}
\int_{B_{x,1/2}}dx^{\prime}\big(1+|\log|x-x^{\prime}||\big)
+\frac{1}{|\text{Im }z|^q}
\int_{\mathbf{R^2}}dx^{\prime}e^{-|\text{Im }z|^p|x-x^{\prime}|} \\
&<\infty.
\end{aligned}
\end{equation*}
\normalsize
Similarly, Corollary \ref{corol_estimates_Vi_G} implies
\footnotesize
\begin{equation*}
\begin{aligned}
\int_{\mathbf{R}^2}\big|\big(\mathcal{V}_i R(z,\mathcal{L})\big)(x,x^{\prime})\big|dx^{\prime}
&=\int_{\mathbf{R}^2\cap B_{x,1/2}}|\mathcal{V}_i G_{\sharp}(x,x^{\prime};z)|dx^{\prime}
+\int_{\mathbf{R}^2\backslash B_{x,1/2}}|\mathcal{V}_i G_{\sharp}(x,x^{\prime};z)|dx^{\prime} \\
&\lesssim
\frac{1}{|\text{Im }z|^q}
\int_{B_{x,1/2}}dx^{\prime}\big(1+|\log|x-x^{\prime}||+\frac{1}{|x-x^{\prime}|}\big)
+\frac{1}{|\text{Im }z|^q}
\int_{\mathbf{R^2}}dx^{\prime}e^{-|\text{Im }z|^p|x-x^{\prime}|} \\
&<\infty. 
\end{aligned}
\end{equation*}
\normalsize
Observe that the above estimates are uniformly for $x\in\mathbf{R}^2$. Hence $$\sup_{x\in\mathbf{R}^2}\int_{\mathbf{R}^2}\big|\big(R(z,\mathcal{L})\big)(x,x^{\prime})\big|dx^{\prime}<\infty, \,\, \sup_{x\in\mathbf{R}^2}\int_{\mathbf{R}^2}\big|\big(\mathcal{V}_i R(z,\mathcal{L})\big)(x,x^{\prime})\big|dx^{\prime}<\infty.$$
\end{proof}

\subsection{Proof of Theorem \ref{prop_edge_index_expression}}

{\color{blue}Step 1.} Note that $x_1\mathcal{V}_2-x_2\mathcal{V}_1=\mathcal{V}_2x_1-\mathcal{V}_1x_2$, which follows from the Jacobi identity for commutators. Expressing $EI_{L}(\Delta)$ using the Hellfer-Sjöstrand formula \eqref{eq_hellfer_sjostrand} gives
\begin{equation} \label{eq_proof_edgeindex_expression_0}
\begin{aligned}
EI_{L}(\Delta)
&=\frac{1}{\pi}\text{Tr}_{\Omega_L}\Big\{
\int_{\mathbf{C}}dm\frac{\partial ^2 \tilde{g}(z)}{\partial \overline{z}\partial z}(\mathcal{V}_2x_1-\mathcal{V}_1x_2)R(z,\mathcal{L}_{\Omega_L})
\Big\} \\
&=-\frac{1}{\pi}\text{Tr}_{\Omega_L}\Big\{
\int_{\mathbf{C}}dm\partial_{\overline{z}}\tilde{g}(z)(\mathcal{V}_2x_1-\mathcal{V}_1x_2)\partial_{z}R(z,\mathcal{L}_{\Omega_L})
\Big\},
\end{aligned}
\end{equation}
where the second equality follows from an integration by parts in $z$. 

By Proposition \ref{prop_uniform_trace_interchange_z_trace}, and applying the Fubini theorem, we can interchange the integral and the trace in \eqref{eq_proof_edgeindex_expression_0}, obtaining:
\begin{equation}
\label{eq_proof_edgeindex_expression_1}
EI_{L}(\Delta)=-\frac{1}{\pi}
\int_{\mathbf{C}}dm\partial_{\overline{z}}\tilde{g}(z)\text{Tr}_{\Omega_L}\Big[(\mathcal{V}_1x_2-\mathcal{V}_2x_1)\partial_{z}R(z,\mathcal{L}_{\Omega_L}) \Big].
\end{equation}

{\color{blue}Step 2.} We calculate the integral kernel in \eqref{eq_proof_edgeindex_expression_1}. Note that $x_i\partial_z G_{L}(x,x^{\prime};z)$ solves the following equation for $x\in \Omega_L$:
\footnotesize
\begin{equation*}
\begin{aligned}
(-\nabla\cdot A\nabla-z)\big(x_i\partial_{z}G_{L}(x,x^{\prime};z) \big)
&=x_i(-\nabla\cdot A\nabla-z)\partial_{z}G_{L}(x,x^{\prime};z)
+[-\nabla\cdot A\nabla,x_i]\partial_{z}G_{L}(x,x^{\prime};z) \\
&=x_i\big[\partial_z (-\nabla\cdot A\nabla-z)G_{L}(x,x^{\prime};z) +G_{L}(x,x^{\prime};z)\big]
+\mathcal{V}_i \partial_z G_{L}(x,x^{\prime};z) \\
&=x_i G_{L}(x,x^{\prime};z)+\mathcal{V}_i \partial_z G_{L}(x,x^{\prime};z),
\end{aligned}
\end{equation*}
\normalsize
where $\partial_z \delta(x-x^{\prime})=0$ is used in the third equality. Moreover, this function satisfies the Dirichlet boundary conditions on $\partial \Omega_L$. Hence, we can invoke Green's formula to express it as
\begin{equation*}
\begin{aligned}
x_i\partial_{z}G_{L}(x,x^{\prime};z)
&=\int_{\Omega_L}dx^{\prime\prime}
G_{L}(x,x^{\prime\prime};z)x_i^{\prime\prime} G_{L}(x^{\prime\prime},x^{\prime};z) \\
&\quad +
\int_{\Omega_L}dx^{\prime\prime}
G_{L}(x,x^{\prime\prime};z)\mathcal{V}_i \partial_z G_{L}(x^{\prime\prime},x^{\prime};z).
\end{aligned}
\end{equation*}
Consequently, we obtain
\begin{equation} \label{eq_proof_edgeindex_expression_8}
\begin{aligned}
\mathcal{V}_jx_i\partial_{z}G_{L}(x,x^{\prime};z)
&=\int_{\Omega_L}dx^{\prime\prime}
\mathcal{V}_jG_{L}(x,x^{\prime\prime};z)x_i^{\prime\prime} G_{L}(x^{\prime\prime},x^{\prime};z) \\
&\quad +
\int_{\Omega_L}dx^{\prime\prime}
\mathcal{V}_jG_{L}(x,x^{\prime\prime};z)\mathcal{V}_i \partial_z G_{L}(x^{\prime\prime},x^{\prime};z).
\end{aligned}
\end{equation}
Here, the differential operator $\mathcal{V}_j$ can be interchanged with the integral because the integrands are absolutely integrable. For instance, since $\mathcal{V}_j G_{L} \in L^{\infty}(\Omega_L, L^{q}(\Omega_L))$ and $\mathcal{V}_j \partial_z G_{L} \in L^{\infty}(\Omega_L, L^{p}(\Omega_L))$ for some $p > 2$ and $1/q + 1/p = 1$ (similar to Corollaries \ref{corol_banach_valued_Gsharp_1} and \ref{corol_banach_valued_2}), we have:
\begin{equation} \label{eq_proof_edgeindex_expression_9}
\begin{aligned}
&\int_{\Omega_L}dx^{\prime\prime}
|\mathcal{V}_jG_{L}(x,x^{\prime\prime};z)||\mathcal{V}_i \partial_z G_{L}(x^{\prime\prime},x^{\prime};z)| \\
&\leq \|\mathcal{V}_jG_{L}(x,x^{\prime\prime};z)\|_{L^{\infty}((\Omega_L),L^{q}(\Omega_L))}\|\mathcal{V}_i \partial_z G_{L}(x^{\prime},x;z)\|_{L^{\infty}((\Omega_L),L^{p}(\Omega_L))}<\infty .
\end{aligned}
\end{equation}
The same argument applies to the other integral in \eqref{eq_proof_edgeindex_expression_8}.

{\color{blue}Step 3.} By Proposition \ref{prop_uniform_trace_interchange_z_trace} and using \eqref{eq_proof_edgeindex_expression_8}, we have:
\begin{equation*}
\begin{aligned}
\text{Tr}_{\Omega_L}\Big[(\mathcal{V}_1x_2-\mathcal{V}_2x_1)\partial_{z}R(z,\mathcal{L}_{\Omega_L}) \Big]
&=\int_{\Omega_L}dx\sum_{i,j\in\{1,2\}}\tilde{\varepsilon}_{ij}(\mathcal{V}_i x_j\partial_{z}G_{\Omega_L})(x,x;z) \\
&=\int_{\Omega_L}dx\int_{\Omega_L}dx^{\prime}\sum_{i,j\in\{1,2\}}\tilde{\varepsilon}_{ij} \mathcal{V}_i G_{L}(x,x^{\prime};z)\mathcal{V}_j \partial_z G_{L}(x^{\prime},x;z) \\
&\quad +
\int_{\Omega_L}dx\int_{\Omega_L}dx^{\prime}\sum_{i,j\in\{1,2\}}\tilde{\varepsilon}_{ij} \mathcal{V}_i G_{L}(x,x^{\prime};z)x^{\prime}_j  G_{L}(x^{\prime},x;z).
\end{aligned}
\end{equation*}
The first integral in the above identity corresponds exactly to the expression of the edge index in Theorem \ref{prop_edge_index_expression}. To complete the proof of Theorem \ref{prop_edge_index_expression}, it remains to show that:
\begin{equation} \label{eq_proof_edgeindex_expression_2}
\int_{\Omega_L}dx\int_{\Omega_L}dx^{\prime}\sum_{i,j\in\{1,2\}}\tilde{\varepsilon}_{ij} \mathcal{V}_i G_{L}(x,x^{\prime};z)x^{\prime}_j  G_{L}(x^{\prime},x;z)=0.
\end{equation}
{\color{blue}Step 4.} We prove \eqref{eq_proof_edgeindex_expression_2}. As demonstrated in the proof, the Dirichlet boundary conditions play a crucial role, particularly in ensuring the vanishing of boundary terms (see \eqref{eq_proof_edgeindex_expression_4}). 

Before presenting the rigorous proof, we provide an informal calculation that offers a quick path to the result.
With $\mathcal{V}_i=[\mathcal{L},x_i]$, the left side of \eqref{eq_proof_edgeindex_expression_2} equals to
\footnotesize
\begin{equation} \label{eq_proof_edgeindex_expression_3_informal_1}
\begin{aligned}
&\int_{\Omega_L}dx\int_{\Omega_L}dx^{\prime}\sum_{i,j\in\{1,2\}}\tilde{\varepsilon}_{ij} \mathcal{V}_i G_{L}(x,x^{\prime};z)x^{\prime}_j  G_{L}(x^{\prime},x;z)\\
&=\int_{\Omega_L}dx^{\prime}\int_{\Omega_L}dx \sum_{i,j\in\{1,2\}}\tilde{\varepsilon}_{ij}\Big[ (\mathcal{L}_x x_i G_{L})(x,x^{\prime};z)x^{\prime}_j  G_{L}(x^{\prime},x;z) \\
&\quad\quad\quad\quad\quad\quad\quad\quad\quad\quad\quad\quad\quad\quad\quad\quad\quad -(x_i\mathcal{L}_x G_{L})(x,x^{\prime};z)x^{\prime}_j  G_{L}(x^{\prime},x;z)\Big] \\
&=\int_{\Omega_L}dx^{\prime}\int_{\Omega_L}dx\sum_{i,j\in\{1,2\}}\tilde{\varepsilon}_{ij} \Big[ ((\mathcal{L}_x-z)x_i G_{L})(x,x^{\prime};z)x^{\prime}_j  G_{L}(x^{\prime},x;z) \\
&\quad\quad\quad\quad\quad\quad\quad\quad\quad\quad\quad\quad\quad\quad\quad\quad\quad -x_i(\mathcal{L}_x-z) G_{L}(x,x^{\prime};z)x^{\prime}_j  G_{L}(x^{\prime},x;z)\Big] \\
&=\int_{\Omega_L}dx^{\prime}\int_{\Omega_L}dx\sum_{i,j\in\{1,2\}}\tilde{\varepsilon}_{ij} \Big[ ((\mathcal{L}_x-z)x_i G_{L})(x,x^{\prime};z)x^{\prime}_j  G_{L}(x^{\prime},x;z) -x_i\delta(x-x^{\prime})x^{\prime}_j  G_{L}(x^{\prime},x;z)\Big].
\end{aligned}
\end{equation}
\normalsize
Note both $x_iG(x,x^{\prime};z)$ and $G(x^{\prime},x;z)$ satisfy the Dirichlet boundary conditions when $x\in\partial \Omega_L$. An integration by part in $x$ further yields
\footnotesize
\begin{equation} \label{eq_proof_edgeindex_expression_3_informal_2}
\begin{aligned}
&\int_{\Omega_L}dx\int_{\Omega_L}dx^{\prime}\sum_{i,j\in\{1,2\}}\tilde{\varepsilon}_{ij} \mathcal{V}_i G_{L}(x,x^{\prime};z)x^{\prime}_j  G_{L}(x^{\prime},x;z) \\
&=\int_{\Omega_L}dx^{\prime}\int_{\Omega_L}dx\sum_{i,j\in\{1,2\}}\tilde{\varepsilon}_{ij} \Big[ (x_i G_{L})(x,x^{\prime};z)x^{\prime}_j  (\overline{\mathcal{L}}_x-z)G_{L}(x^{\prime},x;z) -x_i\delta(x-x^{\prime})x^{\prime}_j  G_{L}(x^{\prime},x;z)\Big] \\
&=\int_{\Omega_L}dx^{\prime}\int_{\Omega_L}dx\sum_{i,j\in\{1,2\}}\tilde{\varepsilon}_{ij} \Big[ (x_i G_{L})(x,x^{\prime};z)x^{\prime}_j  \delta(x-x^{\prime}) -x_i\delta(x-x^{\prime})x^{\prime}_j  G_{L}(x^{\prime},x;z)\Big] \\
&=\int_{\Omega_L}dx^{\prime}\sum_{i,j\in\{1,2\}}\tilde{\varepsilon}_{ij} \Big[ x_i^{\prime}x_{j}^{\prime} G_{L}(x^{\prime},x^{\prime};z) -x_i^{\prime}x^{\prime}_j  G_{L}(x^{\prime},x^{\prime};z)\Big]
=0,
\end{aligned}
\end{equation}
\normalsize
where $\overline{\mathcal{L}}=-\nabla\cdot \overline{A}\nabla$.  The above informal calculation ignores the singularity of the Green function. This issue is addressed through the following rigorous argument.

We observe that the left side of \eqref{eq_proof_edgeindex_expression_2} is absolutely integrable by Hypothesis \ref{prop_GL_singularities}. Hence 
\footnotesize
\begin{equation} \label{eq_proof_edgeindex_expression_3}
\begin{aligned}
&\sum_{i,j\in\{1,2\}}\tilde{\varepsilon}_{ij}\int_{\Omega_L}dx\int_{\Omega_L}dx^{\prime} \mathcal{V}_i G_{L}(x,x^{\prime};z)x^{\prime}_j  G_{L}(x^{\prime},x;z)\\
&=\lim_{\epsilon\to 0}\sum_{i,j\in\{1,2\}}\tilde{\varepsilon}_{ij}\int_{\Omega_L}dx^{\prime}\int_{\Omega_L\backslash B(x^{\prime},\epsilon)}dx \tilde{\varepsilon}_{ij}\mathcal{V}_i G_{L}(x,x^{\prime};z)x^{\prime}_j  G_{L}(x^{\prime},x;z) \\
&=\lim_{\epsilon\to 0}\sum_{i,j\in\{1,2\}}\tilde{\varepsilon}_{ij}\int_{\Omega_L}dx^{\prime}\int_{\Omega_L\backslash B(x^{\prime},\epsilon)}dx \tilde{\varepsilon}_{ij}\Big[ (\mathcal{L}_x x_i G_{L})(x,x^{\prime};z)x^{\prime}_j  G_{L}(x^{\prime},x;z) \\
&\quad\quad\quad\quad\quad\quad\quad\quad\quad\quad\quad\quad\quad\quad\quad\quad\quad -(x_i\mathcal{L}_x G_{L})(x,x^{\prime};z)x^{\prime}_j  G_{L}(x^{\prime},x;z)\Big] \\
&=\lim_{\epsilon\to 0}\sum_{i,j\in\{1,2\}}\tilde{\varepsilon}_{ij}\int_{\Omega_L}dx^{\prime}\int_{\Omega_L\backslash B(x^{\prime},\epsilon)}dx \Big[ ((\mathcal{L}_x-z)x_i G_{L})(x,x^{\prime};z)x^{\prime}_j  G_{L}(x^{\prime},x;z) \\
&\quad\quad\quad\quad\quad\quad\quad\quad\quad\quad\quad\quad\quad\quad\quad\quad\quad -x_i(\mathcal{L}_x-z) G_{L}(x,x^{\prime};z)x^{\prime}_j  G_{L}(x^{\prime},x;z)\Big] \\
&=:\lim_{\epsilon\to 0}\sum_{i,j\in\{1,2\}}\tilde{\varepsilon}_{ij}
\big(I_{ij,1}(z,\epsilon)-I_{ij,2}(z,\epsilon)\big),
\end{aligned}
\end{equation}
\normalsize
where we applied $\mathcal{V}_i=[\mathcal{L},x_i]$ in the second equality above. This limiting process can circumvent the technical issues involving the singularity of the Green function, as previously discussed. Now, consider the integral $I_{12,1}(z,\epsilon)$. Performing an integration by parts, we find that:
\begin{equation}
\label{eq_proof_edgeindex_expression_4}
\begin{aligned}
I_{12,1}(z,\epsilon)&=\int_{\Omega_L}dx^{\prime}\int_{\Omega_L\backslash B(x^{\prime},\epsilon)}dx ((\mathcal{L}_x-z)x_1 G_{L})(x,x^{\prime};z)x^{\prime}_2  G_{L}(x^{\prime},x;z) \\
&=\int_{\Omega_L}dx^{\prime}\int_{\Omega_L\backslash B(x^{\prime},\epsilon)}dx \big(x_1 G_{L}(x,x^{\prime};z)x^{\prime}_2  (\overline{\mathcal{L}}_x-z)G_{L}(x^{\prime},x;z)\big) \\
&\quad +\int_{\Omega_L}dx^{\prime}\int_{\partial B(x^{\prime},\epsilon)}ds(x)n(x)\cdot\Big[
A(x)\nabla_{x}\big(x_1G_{L}(x,x^{\prime};z)\big)x_2^{\prime}G_{L}(x^{\prime},x;z) \\
&\quad\quad\quad\quad\quad\quad\quad\quad\quad\quad\quad\quad\quad\quad\quad
-x_1x_2^{\prime}G_{L}(x,x^{\prime};z)\overline{A}(x)\nabla_{x}G_{L}(x^{\prime},x;z) \Big].
\end{aligned}
\end{equation}
Note that the boundary term at $\partial \Omega_L$ vanishes because $x_iG_L$ satisfies the Dirichlet boundary conditions. Moreover, the spatial integral in \eqref{eq_proof_edgeindex_expression_4} vanishes because $\Omega_L\backslash B(x^{\prime},\epsilon)$ is disjoint from the support of the Dirac mass $\delta(x-x^{\prime})$. Explicitly, we have:
\footnotesize
\begin{equation} \label{eq_proof_edgeindex_expression_5}
\begin{aligned}
&\int_{\Omega_L\backslash B(x^{\prime},\epsilon)}dx \big(x_1 G_{L}(x,x^{\prime};z)x^{\prime}_2  (\overline{\mathcal{L}}_x-z)G_{L}(x^{\prime},x;z)\big) =\int_{\Omega_L\backslash B(x^{\prime},\epsilon)}dx \big(x_1 G_{L}(x,x^{\prime};z)x^{\prime}_2  \delta(x-x^{\prime})\big) =0.
\end{aligned}
\end{equation}
\normalsize
On the other hand, the integral over $\partial B(x^{\prime}, \epsilon)$ becomes:
\footnotesize
\begin{equation}
\label{eq_proof_edgeindex_expression_6}
\begin{aligned}
&\int_{\Omega_L}dx^{\prime}\int_{\partial B(x^{\prime},\epsilon)}ds(x)n(x)\cdot\Big[
A(x)\nabla_{x}\big(x_1G_{L}(x,x^{\prime};z)\big)x_2^{\prime}G_{L}(x^{\prime},x;z) \\
&\quad\quad\quad\quad\quad\quad\quad\quad\quad\quad\quad\quad\quad\quad\quad
-x_1x_2^{\prime}G_{L}(x,x^{\prime};z)\overline{A}(x)\nabla_{x}G_{L}(x^{\prime},x;z) \Big] \\
&=\int_{\Omega_L}x_2^{\prime}dx^{\prime}\int_{\partial B(x^{\prime},\epsilon)}ds(x)x_1 n(x)\cdot\Big[
\big(A(x)\nabla_{x}G_{L}(x,x^{\prime};z)\big)G_{L}(x^{\prime},x;z) \\
&\quad\quad\quad\quad\quad\quad\quad\quad\quad\quad\quad\quad\quad\quad\quad
-G_{L}(x,x^{\prime};z)\overline{A}(x)\nabla_{x}G_{L}(x^{\prime},x;z) \Big] \\
&\quad+
\int_{\Omega_L}x_2^{\prime}dx^{\prime}\int_{\partial B(x^{\prime},\epsilon)}ds(x)(n(x)\cdot A(x)e_1) G_{L}(x,x^{\prime};z)G_{L}(x^{\prime},x;z) \\
&=\int_{\Omega_L}x_1^{\prime}x_2^{\prime}dx^{\prime}\int_{\partial B(x^{\prime},\epsilon)}ds(x)n(x)\cdot\Big[
\big(A(x)\nabla_{x}G_{L}(x,x^{\prime};z)\big)G_{L}(x^{\prime},x;z) -G_{L}(x,x^{\prime};z)\overline{A}(x)\nabla_{x}G_{L}(x^{\prime},x;z) \Big] \\
&\quad+ \Bigg\{
\int_{\Omega_L}x_2^{\prime}dx^{\prime}\int_{\partial B(x^{\prime},\epsilon)}ds(x)(n(x)\cdot A(x)e_1) G_{L}(x,x^{\prime};z)G_{L}(x^{\prime},x;z) \\
&\quad\quad\quad+
\int_{\Omega_L}x_2^{\prime}dx^{\prime}\int_{\partial B(x^{\prime},\epsilon)}ds(x)\big(x_1-x_1^{\prime}\big) n(x)\cdot\Big[
\big(A(x)\nabla_{x}G_{L}(x,x^{\prime};z)\big)G_{L}(x^{\prime},x;z) \\
&\quad\quad\quad\quad\quad\quad\quad\quad\quad\quad\quad\quad\quad\quad\quad\quad\quad\quad\quad\quad\quad
-G_{L}(x,x^{\prime};z)\overline{A}(x)\nabla_{x}G_{L}(x^{\prime},x;z) \Big] \Bigg\}  \\
&=:\int_{\Omega_L}x_1^{\prime}x_2^{\prime}dx^{\prime}\int_{\partial B(x^{\prime},\epsilon)}ds(x)n(x)\cdot\Big[
\big(A(x)\nabla_{x}G_{L}(x,x^{\prime};z)\big)G_{L}(x^{\prime},x;z) 
-G_{L}(x,x^{\prime};z)\overline{A}(x)\nabla_{x}G_{L}(x^{\prime},x;z) \Big] \\
&\quad +\text{(remainder terms)}.
\end{aligned}
\end{equation}
\normalsize
The integrands in the remainder terms admit the estimate $\mathcal{O}(\log^2|x-x^{\prime}|)$ for $\text{Im }z\neq 0$ by Hypothesis \ref{prop_GL_singularities}. Hence their integrals over $\partial B(x^{\prime},\epsilon)$ vanish as $\epsilon\to 0$. Meanwhile, the leading term in \eqref{eq_proof_edgeindex_expression_6} cancels with the corresponding term in $I_{21,1}(z,\epsilon)$. Thus, we conclude:
\begin{equation}
\label{eq_proof_edgeindex_expression_7}
\begin{aligned}
\lim_{\epsilon\to 0} \sum_{i,j\in\{1,2\}}I_{ij,1}(z,\epsilon)=0.
\end{aligned}
\end{equation}
Similarly, it can be shown that:
\begin{equation}
\label{eq_proof_edgeindex_expression_10}
\begin{aligned}
\lim_{\epsilon\to 0} \sum_{i,j\in\{1,2\}}I_{ij,2}(z,\epsilon)=0.
\end{aligned}
\end{equation}
Finally, combining \eqref{eq_proof_edgeindex_expression_3}, \eqref{eq_proof_edgeindex_expression_7}, and \eqref{eq_proof_edgeindex_expression_10}, we conclude the proof of \eqref{eq_proof_edgeindex_expression_2}.

\section{Proof of Theorem \ref{thm_bec}}
In this section, we prove Theorem \ref{thm_bec}. As outlined in Section 1.3, it remains to estimate the difference between the two expressions \eqref{eq_Chern_number_expression} and \eqref{eq_edge_index_expression} as $L\to\infty$, namely, 
\footnotesize
\begin{equation} \label{eq_main_proof_1}
\begin{aligned}
&\lim_{L\to\infty}\frac{1}{|\Omega_L|}EI_{L}(\Delta)
-\mathcal{C}_{\Delta}  \\
&=-\frac{1}{\pi}\lim_{L\to\infty}\frac{1}{|\Omega_L|}\int_{\mathbf{C}}dm\frac{\partial \tilde{g}(z)}{\partial \overline{z}}\int_{\Omega_L}dx\int_{\Omega_L}dx^{\prime}\sum_{i,j\in\{1,2\}}\tilde{\varepsilon}_{ij}\Big[
\big(\mathcal{V}_i G_{L}-\mathcal{V}_i G_{\sharp}\big)(x,x^{\prime};z)
\big(\mathcal{V}_j \partial_z G_{L}\big)(x^{\prime},x;z)  \\
&\quad\quad\quad\quad\quad\quad\quad\quad\quad\quad\quad\quad\quad\quad\quad\quad\quad\quad\quad\quad\quad\quad\quad\quad +
\big(\mathcal{V}_i G_{\sharp}\big)(x,x^{\prime};z)
\big(\mathcal{V}_j \partial_z G_{L}-\mathcal{V}_j \partial_z G_{\sharp}\big)(x^{\prime},x;z)
\Big].
\end{aligned}
\end{equation}
\normalsize
By Proposition \ref{prop_composition rule} and Lemma \ref{lem_composition_integral_operators},  $\mathcal{V}_j \partial_z G_{L}(x^{\prime},x;z)$ and $\mathcal{V}_j \partial_z G_{\sharp}(x^{\prime},x;z)$ can be written in the following composition form
\begin{equation*} \label{eq_main_proof_2}
\begin{aligned}
\mathcal{V}_i\partial_{z}G_{L}(x^{\prime},x;z)
&=\int_{\Omega_L}dx^{\prime\prime}\mathcal{V}_iG_{L}(x^{\prime},x^{\prime\prime};z)G_{L}(x^{\prime\prime},x;z), \\
\mathcal{V}_i\partial_{z}G_{\sharp}(x^{\prime},x;z)
&=\int_{\mathbf{R}^2}dx^{\prime\prime}\mathcal{V}_iG_{\sharp}(x^{\prime},x^{\prime\prime};z)G_{\sharp}(x^{\prime\prime},x;z)\\
&=\int_{\Omega_L}dx^{\prime\prime}\mathcal{V}_iG_{\sharp}(x^{\prime},x^{\prime\prime};z)G_{\sharp}(x^{\prime\prime},x;z)+ \int_{\mathbf{R}^2\backslash \Omega_L}dx^{\prime\prime}\mathcal{V}_iG_{\sharp}(x^{\prime},x^{\prime\prime};z)G_{\sharp}(x^{\prime\prime},x;z). 
\end{aligned}
\end{equation*}
Then the spatial integral in \eqref{eq_main_proof_1} can be estimated by
\footnotesize
\begin{equation} \label{eq_main_proof_3}
\begin{aligned}
&\Big|\int_{\Omega_L}dx\int_{\Omega_L}dx^{\prime}\sum_{i,j\in\{1,2\}}\tilde{\varepsilon}_{ij}\Big[
\big(\mathcal{V}_i G_{L}-\mathcal{V}_i G_{\sharp}\big)(x,x^{\prime};z)
\big(\mathcal{V}_j \partial_z G_{L}\big)(x^{\prime},x;z) 
+
\mathcal{V}_i G_{\sharp}(x,x^{\prime};z)
\big(\mathcal{V}_j \partial_z G_{L}-\mathcal{V}_j \partial_z G_{\sharp}\big)(x^{\prime},x;z)
\Big]\Big| \\
&\leq \int_{\Omega_L\times\Omega_L\times\Omega_L}
dxdx^{\prime}dx^{\prime\prime}\sum_{i,j\in\{1,2\}}
\big|\big(\mathcal{V}_i G_{L}-\mathcal{V}_i G_{\sharp}\big)(x,x^{\prime};z) \big|
\big|\mathcal{V}_j G_{L}(x^{\prime},x^{\prime\prime};z) \big|
\big|G_{L}(x^{\prime\prime},x;z) \big| \\
&\quad +
\int_{\Omega_L\times\Omega_L\times\Omega_L}
dxdx^{\prime}dx^{\prime\prime}\sum_{i,j\in\{1,2\}}
\big|\mathcal{V}_i G_{\sharp}(x,x^{\prime};z) \big|
\big|\big(\mathcal{V}_j G_{L}-\mathcal{V}_j G_{\sharp}\big)(x^{\prime},x^{\prime\prime};z) \big|
\big|G_{L}(x^{\prime\prime},x;z) \big| \\
&\quad +
\int_{\Omega_L\times\Omega_L\times\Omega_L}
dxdx^{\prime}dx^{\prime\prime}\sum_{i,j\in\{1,2\}}
\big|\mathcal{V}_i G_{\sharp}(x,x^{\prime};z) \big|
\big|\mathcal{V}_j G_{\sharp}(x^{\prime},x^{\prime\prime};z) \big|
\big|\big( G_{L}-G_{\sharp}\big)(x^{\prime\prime},x;z) \big| \\
&\quad +
\int_{\Omega_L\times\Omega_L\times(\mathbf{R}^2\backslash\Omega_L)}
dxdx^{\prime}dx^{\prime\prime}\sum_{i,j\in\{1,2\}}
\big|\mathcal{V}_i G_{\sharp}(x,x^{\prime};z) \big|
\big|\mathcal{V}_j G_{\sharp}(x^{\prime},x^{\prime\prime};z) \big|
\big|G_{\sharp}(x^{\prime\prime},x;z) \big| \\
&=:\sum_{k=1}^{4}\sum_{i,j\in\{1,2\}}I^{(k)}_{ij}(z). 
\end{aligned}
\end{equation}
\normalsize
The estimation of the above integrals is a key component in the proof of Theorem \ref{thm_bec}, which is deferred to Proposition \ref{prop_trace_approximation} and Proposition \ref{prop_trace_approximation4} in Subsection \ref{sec-61} and Subsection \ref{sec-62} respectively. Using the results therein, along with the dominated convergence theorem, we can conclude that
\begin{equation*}
\begin{aligned}
\lim_{L\to\infty}\Big|\frac{1}{|\Omega_L|}EI_{L}(\Delta)
-\mathcal{C}_{\Delta} \Big|
\leq \frac{1}{\pi}\int_{\mathbf{C}}dm(z)\lim_{L\to\infty}\frac{1}{|\Omega_L|}\frac{\partial \tilde{g}(z)}{\partial \overline{z}}\sum_{k=1}^{4}\sum_{i,j\in\{1,2\}}I^{(k)}_{ij}(z)
=0.
\end{aligned}
\end{equation*}
This concludes the proof of Theorem \ref{thm_bec}.

\subsection{Estimation of the integrals $I_{ij}^{(k)}(z)$ for $k=1, 2, 3$} \label{sec-61}

\begin{proposition}
\label{prop_trace_approximation}
The following hold for $I^{(k)}_{ij}(z)$ for $k=1,2,3, i,j=1,2$:
\begin{equation} \label{eq_trace_approximation_1}
\frac{1}{|\Omega_L|}\cdot \partial_{\overline{z}}\tilde{g}(z)I^{(k)}_{ij}
\quad \text{is uniformly bounded for }z\in\text{supp }\partial_{\overline{z}}\tilde{g}.
\end{equation}
Moreover, as $L \to \infty$, 
\begin{equation} \label{eq_trace_approximation_2}
\frac{1}{|\Omega_L|}\cdot \partial_{\overline{z}}\tilde{g}(z)I^{(k)}_{ij}\to 0
\quad \text{for any }z\in\text{supp }\partial_{\overline{z}}\tilde{g}\,\backslash \mathbf{R}.
\end{equation}
\end{proposition}

\begin{proof}
We prove the estimate for $I_{ij}^{(1)}$ only, as the proofs for $I_{ij}^{(k)}$, where $k = 2, 3$, are similar.

{\color{blue}Step 1.} We begin by expanding the integral kernel of $\mathcal{V}_i R(z, \mathcal{L}_{\Omega_L}) - \mathcal{V}_i R(z, \mathcal{L})$ using Green's formula. Let the integral kernel of $R(z, \mathcal{L}_{\Omega_L}) - R(z, \mathcal{L})$ be denoted by
$$G_{L,\sharp}(x, x'; z) := G_{L}(x, x'; z) - G_{\sharp}(x, x'; z).$$
Note that $G_{L,\sharp}$ solves the following problem:
\begin{equation*}
\begin{aligned}
(\mathcal{L}_{\Omega_L}-z)G_{L,\sharp}(x,x^{\prime};z)&=0 \quad \text{in }\Omega_L, \\
G_{L,\sharp}(x,x^{\prime};z)&=-G_{\sharp}(x,x^{\prime};z) \quad \text{on }\partial \Omega_L.  
\end{aligned}
\end{equation*}
Thus we can express $G_{L,\sharp}$ by the Green's formula
\begin{equation} \label{eq_trace_proof_Green}
G_{L,\sharp}(x,x^{\prime};z)
=\int_{\partial \Omega_{L}}ds(x^{\prime\prime\prime})n_{x^{\prime\prime\prime}}\cdot A(x^{\prime\prime\prime})\nabla_{x^{\prime\prime\prime}}\overline{G_{L}(x^{\prime},x^{\prime\prime\prime};\overline{z})}G_{\sharp}(x,x^{\prime\prime\prime};z),
\end{equation}
where $n_{x^{\prime\prime\prime}}$ is the outward unit normal at $x^{\prime\prime\prime}$. Hence
\begin{equation} \label{eq_trace_approx_proof_1}
\begin{aligned}
\big(\mathcal{V}_i G_{L}-\mathcal{V}_i G_{\sharp}\big)(x,x^{\prime})
&=\mathcal{V}_iG_{L,\sharp}(x,x^{\prime};z) \\
&=\int_{\partial \Omega_L}ds(x^{\prime\prime\prime})n_{x^{\prime\prime\prime}}\cdot A(x^{\prime\prime\prime})\nabla_{x^{\prime\prime\prime}}\overline{G_{L}(x^{\prime},x^{\prime\prime\prime};\overline{z})}\mathcal{V}_i G_{\sharp}(x,x^{\prime\prime\prime};z).
\end{aligned}
\end{equation}
Note that the interchange of the differential operator $\mathcal{V}_i$ and the integral is justified by observing \eqref{eq_trace_approx_proof_1} is absolutely integrable for $x,x^{\prime}\in \Omega_L$ and $x\neq x^{\prime}$, as seen from the estimate \eqref{eq_estimates_Vi_G_1}. Since $A(x)$ is uniformly bounded, we conclude
\begin{equation} \label{eq_trace_approx_proof_2}
\big|\big(\mathcal{V}_iG_{L}-\mathcal{V}_i G_{\sharp}\big)(x,x^{\prime})\big|
\leq C\int_{\partial \Omega_L}ds(x^{\prime\prime\prime})|\nabla_{x^{\prime\prime\prime}}G_{L}(x^{\prime},x^{\prime\prime\prime};\overline{z})||\mathcal{V}_i G_{\sharp}(x,x^{\prime\prime\prime};z)|. 
\end{equation}

{\color{blue}Step 2.} Substituting \eqref{eq_trace_approx_proof_2} into \eqref{eq_main_proof_3}, we get
\footnotesize
\begin{equation*}
\begin{aligned}
I_{ij}^{(1)} 
&\lesssim \int_{\partial \Omega_L}ds(x^{\prime\prime\prime})\int_{\Omega_L\times\Omega_L\times\Omega_L}
dxdx^{\prime}dx^{\prime\prime}|\nabla_{x^{\prime\prime\prime}}G_{L}(x^{\prime},x^{\prime\prime\prime};\overline{z})||\mathcal{V}_i G_{\sharp}(x,x^{\prime\prime\prime};z)|\big|\mathcal{V}_j G_{L}(x^{\prime},x^{\prime\prime}) \big|
\big|G_{L}(x^{\prime\prime},x)\big| \\
&=\int_{\partial \Omega_L}ds(x^{\prime\prime\prime})
\int_{\Omega_L}dx |\mathcal{V}_i G_{\sharp}(x,x^{\prime\prime\prime};z)|
\int_{\Omega_L}dx^{\prime}|\nabla_{x^{\prime\prime\prime}}G_{L}(x^{\prime},x^{\prime\prime\prime};\overline{z})|
\int_{\Omega_L}dx^{\prime\prime}|\mathcal{V}_j G_L(x^{\prime},x^{\prime\prime};z) |
| G_L(x^{\prime\prime},x;z)|\\
&=\int_{\partial \Omega_L}ds(x^{\prime\prime\prime})
\sum_{m=1}^{4}J_m(x^{\prime\prime\prime};z),
\end{aligned}
\end{equation*}
\normalsize
where
\footnotesize
\begin{equation} \label{eq_trace_approx_proof_3}
\begin{aligned}
J_1(x^{\prime\prime\prime};z)&=: 
\int_{\Omega_L\cap B_{x^{\prime\prime\prime},1/2}}dx |\mathcal{V}_i G_{\sharp}(x,x^{\prime\prime\prime};z)|
\int_{\Omega_L\cap B_{x^{\prime\prime\prime},1/2}}dx^{\prime}|\nabla_{x^{\prime\prime\prime}}G_{L}(x^{\prime},x^{\prime\prime\prime};\overline{z})|
\int_{\Omega_L}dx^{\prime\prime}|\mathcal{V}_j G_L(x^{\prime},x^{\prime\prime};z) |
| G_L(x^{\prime\prime},x;z)|, \\
J_2(x^{\prime\prime\prime};z)&=:
\int_{\Omega_L\cap B_{x^{\prime\prime\prime},1/2}}dx |\mathcal{V}_i G_{\sharp}(x,x^{\prime\prime\prime};z)|
\int_{\Omega_L\backslash B_{x^{\prime\prime\prime},1/2}}dx^{\prime}|\nabla_{x^{\prime\prime\prime}}G_{L}(x^{\prime},x^{\prime\prime\prime};\overline{z})|
\int_{\Omega_L}dx^{\prime\prime}|\mathcal{V}_j G_L(x^{\prime},x^{\prime\prime};z) |
| G_L(x^{\prime\prime},x;z)|, \\
J_3(x^{\prime\prime\prime};z)&=:
\int_{\Omega_L\backslash B_{x^{\prime\prime\prime},1/2}}dx |\mathcal{V}_i G_{\sharp}(x,x^{\prime\prime\prime};z)|
\int_{\Omega_L\cap B_{x^{\prime\prime\prime},1/2}}dx^{\prime}|\nabla_{x^{\prime\prime\prime}}G_{L}(x^{\prime},x^{\prime\prime\prime};\overline{z})|
\int_{\Omega_L}dx^{\prime\prime}|\mathcal{V}_j G_L(x^{\prime},x^{\prime\prime};z) |
| G_L(x^{\prime\prime},x;z)|, \\
J_4(x^{\prime\prime\prime};z)&=:
\int_{\Omega_L\backslash B_{x^{\prime\prime\prime},1/2}}dx |\mathcal{V}_i G_{\sharp}(x,x^{\prime\prime\prime};z)|
\int_{\Omega_L\backslash B_{x^{\prime\prime\prime},1/2}}dx^{\prime}|\nabla_{x^{\prime\prime\prime}}G_{L}(x^{\prime},x^{\prime\prime\prime};\overline{z})|
\int_{\Omega_L}dx^{\prime\prime}|\mathcal{V}_j G_L(x^{\prime},x^{\prime\prime};z) |
| G_L(x^{\prime\prime},x;z)|. 
\end{aligned}
\end{equation}
\normalsize
We estimate $J_m(x^{\prime\prime\prime};z)$ for $m=1,2,3,4$ in the following four steps.

{\color{blue}Step 3.1.} In this step we prove
\begin{equation} \label{eq_trace_approx_proof_9}
\begin{aligned}
J_1(x^{\prime\prime\prime};z)
\lesssim
\frac{1}{|\text{Im }z|^{4q}}\Big(\frac{1}{|\text{Im }z|^{2p}}+L^{\frac{2}{3}} \Big).
\end{aligned}
\end{equation}
First, by the estimate \eqref{eq_estimates_Vi_G_1}
\footnotesize
\begin{equation} \label{eq_trace_approx_proof_4}
\begin{aligned}
J_1(x^{\prime\prime\prime};z)
&\lesssim
\frac{1}{|\text{Im }z|^{2q}}
\int_{\Omega_L\cap B_{x^{\prime\prime\prime},1/2}}\frac{dx}{|x-x^{\prime\prime\prime}|}
\int_{\Omega_L\cap B_{x^{\prime\prime\prime},1/2}}\frac{dx^{\prime}}{|x^{\prime}-x^{\prime\prime\prime}|}\int_{\Omega_L}dx^{\prime\prime}|\mathcal{V}_j G_L(x^{\prime},x^{\prime\prime};z) |
| G_L(x^{\prime\prime},x;z)|.
\end{aligned}
\end{equation}
\normalsize
Next, \eqref{eq_estimates_Vi_G_1} and \eqref{eq_estimates_Vi_G_2} imply that
\footnotesize
\begin{equation} \label{eq_trace_approx_proof_5}
\begin{aligned}
\int_{\Omega_L}dx^{\prime\prime}|\mathcal{V}_j G_L(x^{\prime},x^{\prime\prime};z) |
| G_L(x^{\prime\prime},x;z)|
&=\int_{\Omega_L\cap B_{x,1/2}}dx^{\prime\prime}|\mathcal{V}_j G_L(x^{\prime},x^{\prime\prime};z) |
| G_L(x^{\prime\prime},x;z)| \\
&\quad+ \int_{(\Omega_L\backslash B_{x,1/2})\cap B_{x^{\prime},1/2}}dx^{\prime\prime}|\mathcal{V}_j G_L(x^{\prime},x^{\prime\prime};z) |
| G_L(x^{\prime\prime},x;z)| \\
&\quad+ \int_{(\Omega_L\backslash B_{x,1/2})\backslash B_{x^{\prime},1/2}}dx^{\prime\prime}|\mathcal{V}_j G_L(x^{\prime},x^{\prime\prime};z) |
| G_L(x^{\prime\prime},x;z)| \\
&\lesssim \frac{1}{|\text{Im }z|^{2q}}
\int_{\Omega_L\cap B_{x,1/2}}dx^{\prime\prime} \big(L^{\frac{1}{3}}+\frac{1}{|x^{\prime\prime}-x^{\prime}|}\big)\big(L^{\frac{1}{3}}+|\log|x^{\prime\prime}-x||\big) \\
&\quad + \frac{1}{|\text{Im }z|^{2q}}
\int_{(\Omega_L\backslash B_{x,1/2})\cap B_{x^{\prime},1/2}}dx^{\prime\prime} 
\big(L^{\frac{1}{3}}+\frac{1}{|x^{\prime\prime}-x^{\prime}|}\big)
e^{-|\text{Im }z|^p |x^{\prime\prime}-x|} \\
&\quad + \frac{1}{|\text{Im }z|^{2q}}
\int_{(\Omega_L\backslash B_{x,1/2})\backslash B_{x^{\prime},1/2}}dx^{\prime\prime} e^{-|\text{Im }z|^p |x^{\prime\prime}-x^{\prime}|}e^{-|\text{Im }z|^p |x^{\prime\prime}-x|}. 
\end{aligned}
\end{equation}
\normalsize
Since $\int_{B_{x,1/2}}dx^{\prime\prime}\frac{\log|x^{\prime\prime}-x|}{|x^{\prime\prime}-x^{\prime}|}$ is uniformly bounded in $L$ for $x,x^{\prime}\in \Omega_L$, it holds that
\footnotesize
\begin{equation} \label{eq_trace_approx_proof_6}
\begin{aligned}
\int_{\Omega_L\cap B_{x,1/2}}dx^{\prime\prime} \big(L^{\frac{1}{3}}+\frac{1}{|x^{\prime\prime}-x^{\prime}|}\big)\big(L^{\frac{1}{3}}+|\log|x^{\prime\prime}-x||\big) 
&\leq \int_{B_{x,1/2}}dx^{\prime\prime} \big(L^{\frac{1}{3}}+\frac{1}{|x^{\prime\prime}-x^{\prime}|}\big)\big(L^{\frac{1}{3}}+|\log|x^{\prime\prime}-x||\big) \lesssim L^{\frac{2}{3}}.
\end{aligned}
\end{equation}
\normalsize
In addition, 
\footnotesize
\begin{equation} \label{eq_trace_approx_proof_7}
\begin{aligned}
\int_{(\Omega_L\backslash B_{x,1/2})\cap B_{x^{\prime},1/2}}dx^{\prime\prime} 
\big(L^{\frac{1}{3}}+\frac{1}{|x^{\prime\prime}-x^{\prime}|}\big)
e^{-|\text{Im }z|^p |x^{\prime\prime}-x|}
\leq \int_{B_{x^{\prime},1/2}}dx^{\prime\prime} 
\big(L^{\frac{1}{3}}+\frac{1}{|x^{\prime\prime}-x^{\prime}|}\big)
\lesssim L^{\frac{1}{3}}.
\end{aligned}
\end{equation}
\begin{equation} \label{eq_trace_approx_proof_8}
\begin{aligned}
\int_{(\Omega_L\backslash B_{x,1/2})\backslash B_{x^{\prime},1/2}}dx^{\prime\prime} e^{-|\text{Im }z|^p |x^{\prime\prime}-x^{\prime}|}e^{-|\text{Im }z|^p |x^{\prime\prime}-x|}
\leq \int_{\mathbf{R}^2}dx^{\prime\prime} e^{-|\text{Im }z|^p |x^{\prime\prime}-x|}
\lesssim \frac{1}{|\text{Im }z|^{2p}}.
\end{aligned}
\end{equation}
\normalsize
Combining \eqref{eq_trace_approx_proof_5}-\eqref{eq_trace_approx_proof_8}, we obtain
\begin{equation}  \label{eq_trace_approx_proof_17}
\int_{\Omega_L}dx^{\prime\prime}|\mathcal{V}_j G_L(x^{\prime},x^{\prime\prime};z) |
| G_L(x^{\prime\prime},x;z)|
\leq \frac{1}{|\text{Im }z|^{2q}}\Big(\frac{1}{|\text{Im }z|^{2p}}+L^{\frac{2}{3}} \Big).
\end{equation}
Then \eqref{eq_trace_approx_proof_4} gives
\begin{equation*} 
\begin{aligned}
J_1(x^{\prime\prime\prime};z)
&\lesssim
\frac{1}{|\text{Im }z|^{4q}}\Big(\frac{1}{|\text{Im }z|^{2p}}+L^{\frac{2}{3}} \Big)
\int_{B_{x^{\prime\prime\prime},1/2}}\frac{dx}{|x-x^{\prime\prime\prime}|} \int_{B_{x^{\prime\prime\prime},1/2}}\frac{dx^{\prime}}{|x^{\prime}-x^{\prime\prime\prime}|} \\
&\lesssim
\frac{1}{|\text{Im }z|^{4q}}\Big(\frac{1}{|\text{Im }z|^{2p}}+L^{\frac{2}{3}} \Big).
\end{aligned}
\end{equation*}
This proves \eqref{eq_trace_approx_proof_9}.

{\color{blue}Step 3.2.} The estimate of $J_2$ is obtained similarly:
\footnotesize
\begin{equation} \label{eq_trace_approx_proof_10}
\begin{aligned}
J_2(x^{\prime\prime\prime};z)
&\lesssim
\frac{1}{|\text{Im }z|^{2q}}
\int_{\Omega_L\cap B_{x^{\prime\prime\prime},1/2}}\frac{dx}{|x-x^{\prime\prime\prime}|}
\int_{\Omega_L\backslash B_{x^{\prime\prime\prime},1/2}}dx^{\prime}e^{-|\text{Im }z|^{p}|x^{\prime}-x^{\prime\prime\prime}|}\int_{\Omega_L}dx^{\prime\prime}|\mathcal{V}_j G_L(x^{\prime},x^{\prime\prime};z) |
| G_L(x^{\prime\prime},x;z)| \\
&\lesssim
\frac{1}{|\text{Im }z|^{4q}}\Big(\frac{1}{|\text{Im }z|^{2p}}+L^{\frac{2}{3}} \Big)
\int_{B_{x^{\prime\prime\prime},1/2}}\frac{dx}{|x-x^{\prime\prime\prime}|}
\int_{\mathbf{R}^2}dx^{\prime}e^{-|\text{Im }z|^{p}|x^{\prime}-x^{\prime\prime\prime}|} \\
&\lesssim \frac{1}{|\text{Im }z|^{4q+2p}}\Big(\frac{1}{|\text{Im }z|^{2p}}+L^{\frac{2}{3}} \Big),
\end{aligned}
\end{equation}
\normalsize
where \eqref{eq_trace_approx_proof_17} is applied to obtain the second inequality.

{\color{blue}Step 3.3.} Similarly, the estimate for $J_3$ is given by:
\begin{equation} \label{eq_trace_approx_proof_11}
\begin{aligned}
J_3(x^{\prime\prime\prime};z)
\lesssim \frac{1}{|\text{Im }z|^{4q+2p}}\Big(\frac{1}{|\text{Im }z|^{2p}}+L^{\frac{2}{3}} \Big).
\end{aligned}
\end{equation}

{\color{blue}Step 3.4.} Finally, the estimate for $J_4$ is obtained as
\footnotesize
\begin{equation} \label{eq_trace_approx_proof_12}
\begin{aligned}
J_4(x^{\prime\prime\prime};z)
&\lesssim
\frac{1}{|\text{Im }z|^{2q}}
\int_{\Omega_L\backslash B_{x^{\prime\prime\prime},1/2}}dx e^{-|\text{Im }z|^p|x-x^{\prime\prime\prime}|}
\int_{\Omega_L\backslash B_{x^{\prime\prime\prime},1/2}}dx^{\prime}e^{-|\text{Im }z|^{p}|x^{\prime}-x^{\prime\prime\prime}|}\int_{\Omega_L}dx^{\prime\prime}|\mathcal{V}_j G_L(x^{\prime},x^{\prime\prime};z) |
| G_L(x^{\prime\prime},x;z)| \\
&\lesssim
\frac{1}{|\text{Im }z|^{4q}}\Big(\frac{1}{|\text{Im }z|^{2p}}+L^{\frac{2}{3}} \Big)
\Big(\int_{\mathbf{R}^2}dx^{\prime}e^{-|\text{Im }z|^{p}|x^{\prime}-x^{\prime\prime\prime}|}\Big)^2 \\
&\lesssim \frac{1}{|\text{Im }z|^{4q+4p}}\Big(\frac{1}{|\text{Im }z|^{2p}}+L^{\frac{2}{3}} \Big).
\end{aligned}
\end{equation}
\normalsize

{\color{blue}Step 4.} Using \eqref{eq_trace_approx_proof_3}-\eqref{eq_trace_approx_proof_9} and \eqref{eq_trace_approx_proof_10}-\eqref{eq_trace_approx_proof_12}, and noting that $|\Omega_L|=\mathcal{O}(L^2)$, we conclude
\begin{equation*}
\frac{1}{|\Omega_L|}I_{ij}^{(1)}(z)\lesssim \frac{1}{|\text{Im }z|^{4q+4p}}\Big(\frac{1}{|\text{Im }z|^{2p}}+L^{\frac{2}{3}} \Big)\cdot \frac{s(\partial \Omega_L)}{L^2},
\end{equation*}
where $s(\partial \Omega_L)$ denotes the length of the boundary. Since $s(\partial \Omega_L)=\mathcal{O}(L)$, we have
\begin{equation} \label{eq_trace_approx_final_1}
\frac{1}{|\Omega_L|}I_{ij}^{(1)}(z)\lesssim \frac{1}{|\text{Im }z|^{4q+4p}}\Big(\frac{1}{L|\text{Im }z|^{2p}}+\frac{1}{L^{\frac{1}{3}}} \Big).
\end{equation}
Using the definition of almost analytic extension in Section 3.1 and the fact that $t^{n}e^{-t}$ is bounded on $t\in [0,\infty)$ for any $n>0$, \eqref{eq_trace_approx_final_1} implies the uniform boundedness \eqref{eq_trace_approximation_1}. Finally, the convergence result \eqref{eq_trace_approximation_2} also follows directly from \eqref{eq_trace_approx_final_1}. The proof is complete.
\end{proof}

\subsection{Estimate of $I_{ij}^{(4)}$} \label{sec-62}
\begin{proposition}
\label{prop_trace_approximation4}
The following hold for $I^{(4)}_{ij}$ for $i,j=1,2$:
\begin{equation}
\frac{1}{|\Omega_L|}\cdot \partial_{\overline{z}}\tilde{g}(z)I^{(k)}_{ij}
\quad \text{is uniformly bounded for }z\in\text{supp }\partial_{\overline{z}}\tilde{g}.
\end{equation}
Moreover, as $L \to \infty$, 
\begin{equation}
\frac{1}{|\Omega_L|}\cdot \partial_{\overline{z}}\tilde{g}(z)I^{(k)}_{ij}\to 0
\quad \text{for any }z\in\text{supp }\partial_{\overline{z}}\tilde{g}\,\backslash \mathbf{R}.
\end{equation}
\end{proposition}

\begin{proof}
{\color{blue}Step 1.} First we claim that
\begin{equation} \label{eq_trace_approx_proof_13}
I_{ij}^{(4)}\lesssim
\frac{1}{|\text{Im }z|^{2p+2q}} \int_{\Omega_L}dx^{\prime}
\int_{\mathbf{R}^2\backslash\Omega_{L}}dx^{\prime\prime} |\mathcal{V}_j G_{\sharp}(x^{\prime},x^{\prime\prime};z)|.
\end{equation}
Indeed, using Fubini's theorem, 
\begin{equation*}
\begin{aligned}
I_{ij}^{(4)}
&=\int_{\Omega_L}dx^{\prime}
\int_{\mathbf{R}^2\backslash\Omega_{L}}dx^{\prime\prime} |\mathcal{V}_j G_{\sharp}(x^{\prime},x^{\prime\prime};z) |
\int_{\Omega_L}dx |\mathcal{V}_i G_{\sharp}(x,x^{\prime};z) |
|G_{\sharp}(x^{\prime\prime},x;z)|. \end{aligned}
\end{equation*}
Then \eqref{eq_trace_approx_proof_13} follows by the following estimate
\begin{equation*}
\int_{\Omega_L}dx  |\mathcal{V}_i G_{\sharp}(x,x^{\prime};z) |
|G_{\sharp}(x^{\prime\prime},x;z)|
\lesssim  \frac{1}{|\text{Im }z|^{2p+2q}},
\end{equation*}
which is derived similarly as \eqref{eq_trace_approx_proof_17}, except that it uses the estimates of $G_\sharp$ in Corollary \ref{corol_estimates_Vi_G} instead of $G_L$ in the process.

{\color{blue}Step 2.} We decompose the integral in \eqref{eq_trace_approx_proof_13} as follows: 
\begin{equation} \label{eq_trace_approx_proof_14}
\begin{aligned}
\int_{\Omega_L}dx^{\prime}
\int_{\mathbf{R}^2\backslash\Omega_{L}}dx^{\prime\prime} |\mathcal{V}_j G_{\sharp}(x^{\prime},x^{\prime\prime};z) |
&= \int_{\Omega_L}dx^{\prime}
\int_{\mathbf{R}^2\backslash N(\Omega_L,\sqrt{L})}dx^{\prime\prime} |\mathcal{V}_j G_{\sharp}(x^{\prime},x^{\prime\prime};z) | \\
&\quad +\int_{\Omega_L}dx^{\prime}
\int_{(\mathbf{R}^2\backslash\Omega_{L})\cap N(\Omega_L,\sqrt{L})}dx^{\prime\prime} |\mathcal{V}_j G_{\sharp}(x^{\prime},x^{\prime\prime};z)|\\
&:=J_1+J_2,
\end{aligned}
\end{equation}
where $N(\Omega_L,\sqrt{L}):=\{x:\, \text{dist}(x,\overline{\Omega_L})\leq \sqrt{L}\}$. $J_1$ and $J_2$ are estimated separately as follows.

Estimate of $J_1$: Since $|x^{\prime\prime}-x^{\prime}|>\sqrt{L}>\frac{1}{2}$ when $x^{\prime\prime}\in \mathbf{R}^2\backslash N(\Omega_L,\sqrt{L})$ and $x^{\prime}\in \Omega_L$, we can apply the exponential decay of Green function from Corollary \ref{corol_estimates_Vi_G} to obtain
\footnotesize
\begin{equation} \label{eq_trace_approx_proof_15}
\begin{aligned}
J_1
&\leq \int_{\Omega_L}dx^{\prime}
\int_{\mathbf{R}^2\backslash B_{x^{\prime},\sqrt{L}}}dx^{\prime\prime} |\mathcal{V}_j G_{\sharp}(x^{\prime},x^{\prime\prime};z) |  \frac{1}{|\text{Im }z|^q}\int_{\Omega_L}dx^{\prime}
\int_{\mathbf{R}^2\backslash B_{x^{\prime},\sqrt{L}}}dx^{\prime\prime} e^{-|\text{Im }z|^p|x^{\prime}-x^{\prime\prime}|} \\
&\lesssim \frac{1}{|\text{Im }z|^{q+2p}}\big(|\text{Im }z|^{p}\sqrt{L}+1\big)e^{-|\text{Im }z|^{p}\sqrt{L}}\int_{\Omega_L}dx^{\prime} \\
&\lesssim \frac{1}{|\text{Im }z|^{q+2p}}\big(|\text{Im }z|^{p}\sqrt{L}+1\big)e^{-|\text{Im }z|^{p}\sqrt{L}}\cdot L^2.
\end{aligned}
\end{equation}
\normalsize
Estimate of $J_2$: By Fubini's theorem
\begin{align*}
J_2&= \int_{(\mathbf{R}^2\backslash\Omega_{L})\cap N(\Omega_L,\sqrt{L})}dx^{\prime \prime}
\int_{\Omega_L}dx^{\prime} |\mathcal{V}_j G_{\sharp}(x^{\prime},x^{\prime\prime};z)| \\
&=\int_{(\mathbf{R}^2\backslash\Omega_{L})\cap N(\Omega_L,\sqrt{L})}dx^{\prime\prime}
\Big[
\int_{\Omega_L\cap B_{x^{\prime\prime},1/2}}dx^{\prime} |\mathcal{V}_j G_{\sharp}(x^{\prime},x^{\prime\prime};z)|
+\int_{\Omega_L\backslash B_{x^{\prime\prime},1/2}}dx^{\prime} |\mathcal{V}_j G_{\sharp}(x^{\prime},x^{\prime\prime};z)| \Big].
\end{align*}
Applying the estimate of the Green function in Corollary \ref{corol_estimates_Vi_G}, we can derive that
\begin{equation} \label{eq_trace_approx_proof_16}
\begin{aligned}
J_2
&\lesssim \frac{1}{|\text{Im }z|^{q}}
\int_{N(\Omega_L,\sqrt{L})\backslash \Omega_L}dx^{\prime\prime}
\Big[
\int_{ B_{x^{\prime\prime},1/2}} \frac{dx^{\prime}}{|x^{\prime}-x^{\prime\prime}|}
+\int_{\mathbf{R}^2}dx^{\prime}e^{-|\text{Im }z|^p|x^{\prime}-x^{\prime\prime}|} \Big] \\
&\lesssim \frac{1}{|\text{Im }z|^{q+2p}}
\int_{N(\Omega_L,\sqrt{L})\backslash \Omega_L}dx^{\prime\prime} 
\lesssim \frac{1}{|\text{Im }z|^{q+2p}}L^{\frac{3}{2}},
\end{aligned}
\end{equation}
where $\text{area}(N(\Omega_L,\sqrt{L})\backslash \Omega_L)=\mathcal{O}(\sqrt{L}\cdot s(\partial \Omega_L))=\mathcal{O}(L^{\frac{3}{2}})$ is used in the last inequality above. 

{\color{blue}Step 3.} We conclude from \eqref{eq_trace_approx_proof_13}-\eqref{eq_trace_approx_proof_16} that
\begin{equation*}
\frac{1}{|\Omega_L|}I_{ij}^{(4)}\lesssim
\frac{1}{|\text{Im }z|^{4p+3q}}\Big[\big(|\text{Im }z|^{p}\sqrt{L}+1\big)e^{-|\text{Im }z|^{p}\sqrt{L}}
+\frac{1}{\sqrt{L}}\Big].
\end{equation*}
This, together with the argument after \eqref{eq_trace_approx_final_1}, gives the desired estimates of $I_{ij}^{(4)}$ as claimed in Proposition \ref{prop_trace_approximation}.
\end{proof}

\bigskip
\textbf{Data Availability Statement}: The authors confirm that the current paper does not
have associated data.

\appendix

\section{Appendix: Physical interpretation of Photonic bulk-edge correspondence}
\setcounter{equation}{0}
\setcounter{subsection}{0}
\setcounter{theorem}{0}
\renewcommand{\theequation}{A.\arabic{equation}}
\renewcommand{\thesubsection}{A.\arabic{subsection}}
\renewcommand{\thetheorem}{A.\arabic{theorem}}

This Appendix provides a physical interpretation of the bulk-edge correspondence in Theorem \ref{thm_bec}, drawing on the principles of linear response theory and energy conservation. We begin by introducing the Hamiltonian formulation of electromagnetism to establish the framework for our discussion. Next, we apply linear response theory to demonstrate that the gap Chern number represents the bulk conductivity of photonic energy in response to external excitation, analogous to the Hall conductivity in the quantum Hall effect (QHE). Finally, we show that the bulk-edge correspondence naturally follows as a consequence of energy conservation. 

The presentation in this section is based on the works in \cite{mario23shaking,bhat2006hamiltonian,tong2016lecturesquantumhalleffect,landau2013statistical}. We also note the reference \cite{de2017symmetry_classification} for a systematic study of (bulk) topological photonic crystals based on the Hamiltonian formulation similar to that of this Appendix, with an emphasis on the symmetry classification of (bulk) mediums.

\subsection{Quantum field formulation of electromagnetism}

To establish the framework for the discussion of bulk-edge correspondence in photonic structures, we consider the unpolarized electromagnetic field $\bm{F}=(\bm{E},\bm{H})^{T}$,  rather than focusing solely on 
the $TE$ polarized wave as in the main text. The field $\bm{F}$ is governed by the Maxwell equation, which can be recast into a Schrödinger-like form: 
\begin{equation} \label{eq_phys_1}
i\partial_{t}\bm{F}=\mathfrak{M}\bm{F}
\end{equation}
where the operator $\mathfrak{M}$ is defined as
\begin{equation} \label{eq_phys_10}
\mathfrak{M}=M^{-\frac{1}{2}}
\begin{pmatrix}
0 & i\nabla\times \bm{1}_{3\times 3} \\ -i\nabla\times \bm{1}_{3\times 3} & 0
\end{pmatrix}M^{-\frac{1}{2}},\quad 
M=M(\bm{r})=
\begin{pmatrix}
\epsilon(\bm{r}) & 0 \\ 0 & \mu(\bm{r})
\end{pmatrix}.
\end{equation}
We assume the medium is non-dispersive and lossless, meaning the material response matrix $M$ is time-independent and Hermitian. For a more general treatment of dispersive or lossy material, see \cite{mario23shaking}. We also neglect the bi-anisotropy of the medium by assuming $M$ is diagonal in block form, as shown in \eqref{eq_phys_10}. Additionally, we assume the system is periodic in $xy$-plane satisfying $M(\bm{r}+\bm{e}_i)=M(\bm{r})$ for $i=1,2$, where $e_i$'s are lattice vectors of the periodic structure, and is closed in $z$-direction; specifically, this is achieved by bounding the periodic structure by Perfect Electric Conductor walls at $z=a$ and $z=b$. The $3D$ unit cell is denoted as $Y\times (a,b)$ (recall that $Y$ is the $2D$ cell as in the main text).

We present the discussion of linear response theory within the framework of quantum field theory. To achieve this, we quantize the electromagnetic field, $\bm{F}$, following the standard procedure: each normal mode of the classical problem is associated with a quantum harmonic oscillator, as detailed in\cite{bhat2006hamiltonian,peskin2018introduction}
\begin{equation*}
\bm{F}(\bm{r},t)=\int_{T^2}d\bm{\kappa}\sum_{w_{n,\bm{\kappa}}>0}\Big(e^{-iw_{n,\bm{\kappa}}t}a^{\dagger}_{n,\bm{\kappa}}\bm{F}_{n,\bm{\kappa}}(\bm{r})+e^{iw_{n,\bm{\kappa}}t}a_{n,\bm{\kappa}}\overline{\bm{F}_{n,\bm{\kappa}}(\bm{r})} \Big).
\end{equation*}
Here $a_{n,\bm{\kappa}},a^{\dagger}_{n,\bm{\kappa}}$ denote the standard annihilation and creation operators with the $[a_{n,\bm{\kappa}},a^{\dagger}_{n^{\prime},\bm{\kappa}^{\prime}}]=\delta_{n,n^{\prime}}\delta_{\bm{\kappa}\bm{\kappa}^{\prime}}$. $\bm{F}_{n,\bm{\kappa}}$ denotes the normal mode of Maxwell equation of frequency $w_{n,\bm{\kappa}}$ with the normalization condition $\int_{Y\times (a,b)}d\bm{r}\overline{\bm{F}^{T}_{n,\bm{\kappa}}}(\bm{r})\bm{F}_{n,\bm{\kappa}}(\bm{r})=1$. The Bloch-momentum $\bm{\kappa}$ ranges over the reciprocal cell $T^2=\mathbf{R}^2/\mathbf{Z}^2$. Acting on the vacuum state $|0\rangle$, $\bm{F}(\bm{r},t)$ creates the field at spacetime $(\bm{r},t)$, given by
\begin{equation*}
|\bm{r},t\rangle :=\bm{F}(\bm{r},t)|0\rangle=\int_{T^2}d\bm{\kappa}\sum_{w_{n,\bm{\kappa}}>0} e^{-iw_{n,\bm{\kappa}}t}\bm{F}_{n,\bm{\kappa}}(\bm{r})|n,\bm{\kappa}\rangle.
\end{equation*}

The field in the frequency domain is 
$$
\bm{F}(\bm{r},w)=2\pi\int_{T^2}d\bm{\kappa}\sum_{w_{n,\bm{\kappa}}>0}\big(\delta(w+w_{n,\bm{\kappa}})a^{\dagger}_{n,\bm{\kappa}}\bm{F}_{n,\bm{\kappa}}+\delta(w-w_{n,\bm{\kappa}})a_{n,\bm{\kappa}}\overline{\bm{F}_{n,\bm{\kappa}}} \big).
$$ 
For a given cutoff frequency or threhold energy $\omega_0$, 
integrating $e^{iwt}\bm{F}(\bm{r},w)$ over the frequency band $(0,w_0)$ yields the following partial field with cutoff frequency $\omega_0$:
\begin{equation} \label{eq_phys_partial_field}
\bm{F}(\bm{r},t;w_0)=\int_{T^2}d\bm{\kappa}\sum_{0<w_{n,\bm{\kappa}}<w_0}\Big(e^{-iw_{n,\bm{\kappa}}t}a^{\dagger}_{n,\bm{\kappa}}\bm{F}_{n,\bm{\kappa}}(\bm{r})+e^{iw_{n,\bm{\kappa}}t}a_{n,\bm{\kappa}}\overline{\bm{F}_{n,\bm{\kappa}}(\bm{r})} \Big). 
\end{equation}
The observable of interest is the Poynting vector $\bm{S}=(S_1, S_2, S_3)^{T}$, which measures the energy flux of electromagnetic fields. Classically, the Poynting vector is defined as $\bm{S}(\bm{r},t)=\text{Re}(\bm{E}(\bm{r},t)\times \overline{\bm{H}}(\bm{r},t))$. In quantum field theory, it's naturally associated with the operators $\hat{S}_j$ as follows
\begin{equation*}
S_{j}(\bm{r},t)=\frac{1}{2}\langle \bm{r},t|\hat{S}_j|\bm{r},t\rangle
=\frac{1}{2}\langle 0|\bm{F}^{\dagger}(\bm{r},t)\hat{S}_j\bm{F}(\bm{r},t)|0\rangle ,
\end{equation*}
where $\hat{S}_j$ is defined as
\begin{equation*}
\hat{S}_j:=
\begin{pmatrix}
0 & \bm{e}_{j}\times \bm{1}_{3\times 3} \\ -\bm{e}_{j}\times \bm{1}_{3\times 3} & 0     
\end{pmatrix}.
\end{equation*}
We note that the Poynting vector can be regarded as the photonic counterpart of the electric current density $\bm{j}$. This analogy is evident from the conservation law 
$$
\nabla\cdot \bm{S}+\partial_t W=0 ,
$$
where $W$ denotes the electromagnetic energy density. This conservation law parallels charge conservation in electronic systems
$$
\nabla\cdot \bm{j}+\partial_t\rho=0.
$$ 
This analogy forms the basis for interpreting the photonic Hall effect. 

Motivated by this analogy, we study the modified vector $\hat{S}_k^{m}$ defined below as the natural counterpart of the momentum operator $J_k=i[x_k,H]$
in electronic systems
\begin{equation*}
\hat{S}_k^{m}:=i[x_k,\mathfrak{M}]=M^{-\frac{1}{2}}\hat{S}_k M^{-\frac{1}{2}}. 
\end{equation*}

\subsection{Interpreting the gap Chern number using the linear response theory}
In this section, we demonstrate that the gap Chern number represents the bulk conductivity of photonic energy in response to external excitation, analogous to the quantum Hall effect. We assume that the unperturbed photonic structure exhibits a spectral band gap, with $w_0$ lying within this gap. To proceed, we consider the partial field with a cutoff frequency $w_0$, denoted as $\bm{F}(\bm{r}, t; w_0)|0\rangle$.  Specifically, we focus on the spatially averaged energy flux of this partial field
\begin{equation} \label{eq_phys_2}
\langle \hat{S}_k^{m}(t) \rangle_{w_0}
=\int_{Y\times (a,b)}d\bm{r}\langle 0|\bm{F}^{\dagger}(\bm{r},t;w_0)\hat{S}^{m}_j(t)\bm{F}(\bm{r},t;w_0)|0\rangle. 
\end{equation}
where the observable evolves as $\hat{S}^{m}_j(t)=e^{i\mathcal
L^{\times}t}\hat{S}^{m}_je^{-i\mathcal
L^{\times}t}$ in the interaction picture \cite{tong2016lecturesquantumhalleffect}. We are interested in the fluctuation of $\langle \hat{S}_k^{m} \rangle_{av}$ when the field is coupled with external perturbation. Following a general principle of interaction theory, for the observable $\hat{S}_k^{m}$, there exists an operator $\hat{A}_k$ such that $\hat{A}_k\hat{S}_k^{m}$ represents the peturbation to $\mathfrak{M}$ induced by the external interaction; see $\S 75$ of \cite{landau2013statistical}. Note that in the quantum Hall effect, $\delta H=-\bm{J}\cdot \bm{A}$ where $\bm{A}$ denotes the vector potential. In our photonic setting, the perturbed field is governed by
\begin{equation} \label{eq_phys_3}
i\partial_{t}\bm{F}_{\delta}=\big(\mathfrak{M}+\delta\mathfrak{M}(t)\big)\bm{F}_{\delta}
\quad \text{with} \quad
\delta\mathfrak{M}(t):=\sum_{k=1}^{2}\hat{A}_k(t)\hat{S}_k^{m}(t). 
\end{equation}
Here, we've assumed that the perturbed field has only $x$- and $y$- components (since our system is essentially 2D) and that the interaction operator $\hat{A}_k(t)$ is time-dependent. A concrete example of such interaction is provided in \cite{mario23shaking}, where $\hat{A}_k$ describes the mechanical driving of a photonic system by periodically and slowly shaking.  This mechanical driving can also be mimicked by purely electrical means using spacetime-modulated (Floquet) materials, as discussed in \cite{mario23shaking}.

Now we calculate the expectation \eqref{eq_phys_2} in response to perturbations. In the interaction picture, the perturbed field evolves as $\bm{F}_{\delta}(\bm{r},t;w_0)=T\exp\big(-i\int_{t_0}^{t}\delta\mathfrak{M}(t^{\prime})dt^{\prime} \big)\bm{F}(\bm{r},t;w_0)$, where $T$ denotes the time-ordering operator. Taking it inside \eqref{eq_phys_2} and extracting the linear-order approximation gives
\begin{equation*}
\begin{aligned}
\langle \hat{S}_k^{m}(t) \rangle_{w_0}
&=\int_{Y\times (a,b)}d\bm{r}\langle 0|(\bm{F}_{\delta})^{\dagger}(\bm{r},t;w_0)\hat{S}^{m}_j(t)\bm{F}_{\delta}(\bm{r},t;w_0)|0\rangle \\
&\approx \int_{Y\times (a,b)}d\bm{r}\langle 0|\bm{F}^{\dagger}(\bm{r},t;w_0)\Big(\hat{S}^{m}_j(t)+i\int_{t_0}^{t}dt^{\prime}[\delta\mathfrak{M}(t^{\prime}),\hat{S}_k^{m}(t)] \Big)\bm{F}(\bm{r},t;w_0)|0\rangle.
\end{aligned}
\end{equation*}
Assuming the perturbation is turned on at $t_0=-\infty$,
the fluctuation of energy flux due to the perturbation is
\begin{equation} \label{eq_phys_4}
\langle \Delta\hat{S}_k^{m}(t) \rangle_{w_0}
=\int_{Y\times (a,b)}d\bm{r}\langle 0|\bm{F}^{\dagger}(\bm{r},t;w_0)\Big(i\int_{-\infty}^{t}dt^{\prime}[\delta\mathfrak{M}(t^{\prime}),\hat{S}_k^{m}(t)]\Big)\bm{F}(\bm{r},t;w_0)|0\rangle.
\end{equation}
Now we assume the perturbed field is along one axis, i.e. $ \hat{A}_k(t)=\frac{\mathcal{A}_k\delta_{k\ell}}{i\Omega}e^{-i\Omega t}$ for a fixed $\ell\in\{1,2\}$. We investigate the response of \eqref{eq_phys_4} to $\mathcal{A}_\ell$ as $\Omega\to 0$ (mimicking the response of electric current to the electric field $\bm{E}=-\partial_{t}\bm{A}$ in the quantum Hall effect \cite{tong2016lecturesquantumhalleffect}). Assuming $\mathcal{A}_\ell$ is scalar-valued (hence commutes with $\hat{S}_{\ell}^{m}(t)$), \eqref{eq_phys_4} becomes
\begin{equation*}
\langle \Delta\hat{S}_k^{m}(t) \rangle_{w_0}
=\frac{\mathcal{A}_\ell}{\Omega}\int_{-\infty}^{t}dt^{\prime}e^{-i\Omega t^{\prime}}\int_{Y\times (a,b)}d\bm{r}\langle 0|\bm{F}^{\dagger}(\bm{r},t;w_0)\Big([\hat{S}_{\ell}^{m}(t^{\prime}),\hat{S}_k^{m}(t)]\Big)\bm{F}(\bm{r},t;w_0)|0\rangle.
\end{equation*}
Due to the time-translation invariance of \eqref{eq_phys_3}, this expression depends only on the time difference $t^{\prime\prime}=t-t^{\prime}$, and simplifies to
\begin{equation*}
\langle \Delta\hat{S}_k^{m}(t) \rangle_{w_0}
=\frac{\mathcal{A}_\ell e^{-i\Omega t}}{\Omega}\int_{0}^{\infty}dt^{\prime\prime}e^{-i\Omega t^{\prime\prime}}\int_{Y\times (a,b)}d\bm{r}\langle 0|\bm{F}^{\dagger}(\bm{r},t^{\prime\prime};w_0)\Big([\hat{S}_{\ell}^{m}(0),\hat{S}_k^{m}(t^{\prime\prime})]\Big)\bm{F}(\bm{r},t^{\prime\prime};w_0)|0\rangle.
\end{equation*}
Substituting \eqref{eq_phys_partial_field} into the above expression leads to
\footnotesize
\begin{equation} \label{eq_phys_5}
\langle \Delta\hat{S}_k^{m}(t) \rangle_{w_0}
=\frac{\mathcal{A}_\ell e^{-i\Omega t}}{\Omega}\int_{0}^{\infty}dt^{\prime\prime}e^{-i\Omega t^{\prime\prime}}
\int_{T^2}d\bm{\kappa}\sum_{0<w_{n,\bm{\kappa}}<w_0}\int_{Y\times (a,b)}d\bm{r}\overline{\bm{F}}^{T}_{n,\bm{\kappa}}(\bm{r}) \Big([\hat{S}_{\ell}^{m}(0),\hat{S}_k^{m}(t^{\prime\prime})]\Big)   \bm{F}_{n,\bm{\kappa}}(\bm{r}).
\end{equation}
\normalsize
\eqref{eq_phys_5} indicates that the perturbation at frequency $\Omega$ induces fluctuation with the same oscillating frequency. Moreover, only the off-diagonal terms in \eqref{eq_phys_5} survive in the commutator, which is analogous to the survival of transverse conductivity $\sigma_{xy}$ in the quantum Hall effect.

Dividing both sides of \eqref{eq_phys_5} by the driving force $\mathcal{A}_\ell e^{-i\Omega t}$, we obtain the ``photonic Hall conductance'' \cite{mario23shaking}
\begin{equation*}
\sigma_{k\ell}(\Omega)
=\int_{0}^{\infty}\frac{dt^{\prime\prime}e^{-i\Omega t^{\prime\prime}}}{\Omega}
\int_{T^2}d\bm{\kappa}\sum_{0<w_{n,\bm{\kappa}}<w_0}\int_{Y\times (a,b)}d\bm{r}\overline{\bm{F}}^{T}_{n,\bm{\kappa}}(\bm{r}) \Big([\hat{S}_{\ell}^{m}(0),\hat{S}_k^{m}(t^{\prime\prime})]\Big)   \bm{F}_{n,\bm{\kappa}}(\bm{r}).
\end{equation*}
Note that we only focus on the conductivity in $xOy$ plane, i.e. $k,\ell\in\{1,2\}$. Using the completeness of normal modes $\bm{F}_{n,\bm{\kappa}}$ and the evolution operator $\hat{S}^{m}_j(t)=e^{i\mathcal
L^{\times}t}\hat{S}^{m}_je^{-i\mathcal
L^{\times}t}$, we further get
\footnotesize
\begin{equation*}
\begin{aligned}
\sigma_{k\ell}(\Omega)
&=\int_{0}^{\infty}\frac{dte^{-i\Omega t}}{\Omega}
\int_{T^2}d\bm{\kappa}\sum_{\substack{0<w_{n,\bm{\kappa}}<w_0 \\ w_{n^{\prime},\bm{\kappa}}>0}}
\Big\{\big(\hat{S}_{\ell}^{m}\bm{F}_{n^{\prime},\bm{\kappa}},\bm{F}_{n,\bm{\kappa}} \big)_{3D }
\big(\hat{S}_k^{m}(t)\bm{F}_{n,\bm{\kappa}},\bm{F}_{n^{\prime},\bm{\kappa}} \big)_{3D } \\
&\quad\quad\quad\quad\quad\quad\quad\quad\quad\quad\quad\quad\quad\quad\quad -\big(\hat{S}_{k}^{m}(t)\bm{F}_{n^{\prime},\bm{\kappa}},\bm{F}_{n,\bm{\kappa}} \big)_{3D }
\big(\hat{S}_{\ell}^{m}\bm{F}_{n,\bm{\kappa}},\bm{F}_{n^{\prime},\bm{\kappa}} \big)_{3D }
\Big\} \\
&=\int_{0}^{\infty}\frac{dte^{-i\Omega t}}{\Omega}
\int_{T^2}d\bm{\kappa}\sum_{\substack{0<w_{n,\bm{\kappa}}<w_0 \\ w_{n^{\prime},\bm{\kappa}}>0}}
\Big\{
e^{-i(w_{n,\bm{\kappa}}-w_{n^{\prime},\bm{\kappa}})t}
\big(\hat{S}_k^{m}\bm{F}_{n,\bm{\kappa}},\bm{F}_{n^{\prime},\bm{\kappa}} \big)_{3D }
\big(\bm{F}_{n^{\prime},\bm{\kappa}},\hat{S}_{\ell}^{m}\bm{F}_{n,\bm{\kappa}} \big)_{3D } \\
&\quad\quad\quad\quad\quad\quad\quad\quad\quad\quad\quad\quad\quad\quad\quad
-e^{i(w_{n,\bm{\kappa}}-w_{n^{\prime},\bm{\kappa}})t}\big(\hat{S}_{k}^{m}\bm{F}_{n^{\prime},\bm{\kappa}},\bm{F}_{n,\bm{\kappa}} \big)_{3D }
\big(\bm{F}_{n,\bm{\kappa}},\hat{S}_{\ell}^{m}\bm{F}_{n^{\prime},\bm{\kappa}} \big)_{3D }
\Big\} ,
\end{aligned}
\end{equation*}
\normalsize
where the bracket $(\bm{F},\bm{G})_{3D }=\int_{Y\times (a,b)}d\bm{r}\overline{\bm{G}}^{T}(\bm{r})\bm{F}(\bm{r})$ denotes the $3D$ inner product between vectorial functions. Obviously the above quantity vanishes for $k=\ell$; hence the only interesting ones are $\sigma_{xy}$ and $\sigma_{yx}$ (transverse Hall conductivity). Recall that $w_0$ lies in a band gap, 
\begin{equation*}
\begin{aligned}
\sigma_{xy}(\Omega)=-\sigma_{yx}(\Omega)
=\frac{i}{\Omega}\int_{T^2}d\bm{\kappa}\sum_{\substack{0<w_{n,\bm{\kappa}}<w_0 \\ w_{n^{\prime},\bm{\kappa}}>w_0}}\Big\{&
\frac{\big(\hat{S}_x^{m}\bm{F}_{n,\bm{\kappa}},\bm{F}_{n^{\prime},\bm{\kappa}} \big)_{3D }
\big(\bm{F}_{n^{\prime},\bm{\kappa}},\hat{S}_{y}^{m}\bm{F}_{n,\bm{\kappa}} \big)_{3D }}{\Omega+w_{n,\bm{\kappa}}-w_{n^{\prime},\bm{\kappa}}} \\
&-\frac{\big(\hat{S}_{x}^{m}\bm{F}_{n^{\prime},\bm{\kappa}},\bm{F}_{n,\bm{\kappa}} \big)_{3D }
\big(\bm{F}_{n,\bm{\kappa}},\hat{S}_{y}^{m}\bm{F}_{n^{\prime},\bm{\kappa}} \big)_{3D }}{\Omega+w_{n^{\prime},\bm{\kappa}}-w_{n,\bm{\kappa}}}
\Big\}. 
\end{aligned}
\end{equation*}
Here a more rigorous treatment involves replacing $\Omega$ by $ \Omega+ i\epsilon$ with $\epsilon>0$ and taking the limit $\epsilon \to 0$. This ensures the cauality in the linear response and also justifies the convergence of the integrals. We omit this procedure and refer the reader to \cite{tong2012kinetic} for more detail. The limit of $\Omega\to 0$ is
\begin{equation} \label{eq_phys_6}
\sigma_{xy}
=i\int_{T^2}d\bm{\kappa}\sum_{\substack{0<w_{n,\bm{\kappa}}<w_0 \\ w_{n^{\prime},\bm{\kappa}}>w_0}}
\frac{\big(\hat{S}_x^{m}\bm{F}_{n,\bm{\kappa}},\bm{F}_{n^{\prime},\bm{\kappa}} \big)_{3D }
\big(\bm{F}_{n^{\prime},\bm{\kappa}},\hat{S}_{y}^{m}\bm{F}_{n,\bm{\kappa}} \big)_{3D }-\big(\hat{S}_{x}^{m}\bm{F}_{n^{\prime},\bm{\kappa}},\bm{F}_{n,\bm{\kappa}} \big)_{3D }
\big(\bm{F}_{n,\bm{\kappa}},\hat{S}_{y}^{m}\bm{F}_{n^{\prime},\bm{\kappa}} \big)_{3D }}{(w_{n,\bm{\kappa}}-w_{n^{\prime},\bm{\kappa}})^2}
\end{equation}
which resembles the Kubo formula for Hall conductance; see $\S 2.2.3$ of \cite{tong2016lecturesquantumhalleffect}. We note its obvious analogy to \eqref{eq_proof_chern_expression_2} in the main text (Step 1 in the proof of Theorem \ref{prop_Chern_number_expression}). Following the derivation of \eqref{eq_proof_chern_expression_2}, we obtain
\begin{equation} \label{eq_phys_7}
\sigma_{xy}=\frac{1}{2}\int_{T^2}d\bm{\kappa}\sum_{0<w_{n,\bm{\kappa}}<w_0}\text{Im}(\partial_{\kappa_1}\tilde{\bm{F}}_{n,\bm{\kappa}},\partial_{\kappa_2}\tilde{\bm{F}}_{n,\bm{\kappa}})_{3D }
\end{equation}
where $\tilde{\bm{F}}_{n,\bm{\kappa}}=e^{-i\bm{\kappa}\cdot \bm{x}}\bm{F}_{n,\bm{\kappa}}$
That's exactly (up to a constant) the Chern number associated with the frequency bands in the range $w\in (w,w_0)$. Notably, \eqref{eq_phys_6} indicates the photonic Hall effect results from the interband transition \cite{winn1999interband} between the low-frequency bands (with $w<w_0$) and high-frequency bands (with $w>w_0$), as being proved rigorously in Lemma \ref{lem_four_mero_functions}.

In conclusion, we derive the analogy of the quantum Hall effect in photonics based on the quantization of electromagnetic field and the linear response theory. In this setup, the Chern number of frequency bands is illustrated as the energy flux induced by external interaction in the sense of \eqref{eq_phys_3}.

\subsubsection{Reduction to 2D system}
We briefly digress to demonstrate how the Chern number expression \eqref{eq_chern_num} for TE waves can be derived from \eqref{eq_phys_7}. This reduction applies when the system is homogeneous in $z$ direction, meaning that the material response matrix $M$ in \eqref{eq_phys_10} is independent of $z$. Under this condition, each normal mode can be classified as either TE or TM polarized. Specifically, the mode takes the form $\tilde{\bm{F}}_{n,\bm{\kappa}}=(E_{1,n,\bm{\kappa}}, E_{2,n,\bm{\kappa}}, 0, 0, 0, 0, H_{3,n,\bm{\kappa}})^{T}$
for TE polarization, or
$(0, 0, E_{3,n,\bm{\kappa}}, H_{1,n,\bm{\kappa}}, H_{2,n,\bm{\kappa}}, 0)^{T}$ for TM polarization. Consider the case of TE polarization with a band gap near $\omega_0$. Following \eqref{eq_phys_1}, then the normal mode can be written as 
\begin{equation} \label{eq_phys_8}
\tilde{\bm{F}}_{\bm{\kappa}}=(\bm{E}_{\bm{\kappa}} \, ,\, \bm{H}_{\bm{\kappa}})^{T}\quad \text{with}\quad
\bm{E}_{\bm{\kappa}}=\frac{i}{w_{\bm{\kappa}}}\epsilon^{-\frac{1}{2}}\nabla\times(\mu^{-\frac{1}{2}}\bm{H}_{\bm{\kappa}}),\quad
\bm{H}_{\bm{\kappa}}=(0,0,u_{\bm{\kappa}}).
\end{equation}
Here we have
\begin{equation} \label{eq_phys_11}
\mu^{-\frac{1}{2}}\big(\nabla\times \epsilon^{-1} \nabla \times\big) \mu^{-\frac{1}{2}}\bm{H}_{\bm{\kappa}}=w_{\bm{\kappa}}^2\bm{H}_{\bm{\kappa}}.
\end{equation}
Substituting \eqref{eq_phys_8} into \eqref{eq_phys_7} yields 
\begin{equation} \label{eq_phys_9}
\sigma_{xy}=\frac{1}{2}\int_{T^2}d\bm{\kappa}
\Big[\text{Im}\Big(\partial_{\kappa_1}\big(\frac{i}{w_{\bm{\kappa}}}\epsilon^{-\frac{1}{2}}\nabla\times(\mu^{-\frac{1}{2}}\bm{H}_{\bm{\kappa}})\big),\partial_{\kappa_2}\big(\frac{i}{w_{\bm{\kappa}}}\epsilon^{-\frac{1}{2}}\nabla\times(\mu^{-\frac{1}{2}}\bm{H}_{\bm{\kappa}})\big)\Big)_{3D }  +\text{Im}(\partial_{\kappa_1}H_{\bm{\kappa}},\partial_{\kappa_2}H_{\bm{\kappa}})\Big].
\end{equation}
The second integral in \eqref{eq_phys_9} corresponds exactly to the expression of Chern number \eqref{eq_chern_num} in 2D system. Now we evaluate the first one. Direct calculation yields
\footnotesize
\begin{equation} \label{eq_phys_12}
\begin{aligned}
&\Big(\partial_{\kappa_1}\big(\frac{i}{w_{\bm{\kappa}}}\epsilon^{-\frac{1}{2}}\nabla\times(\mu^{-\frac{1}{2}}\bm{H}_{\bm{\kappa}})\big),\partial_{\kappa_2}\big(\frac{i}{w_{\bm{\kappa}}}\epsilon^{-\frac{1}{2}}\nabla\times(\mu^{-\frac{1}{2}}\bm{H}_{\bm{\kappa}})\big)\Big)_{3D } \\
&=\frac{1}{w_{\bm{\kappa}}^2}  \Big(\partial_{\kappa_1}\big(\epsilon^{-\frac{1}{2}}\nabla\times(\mu^{-\frac{1}{2}}\bm{H}_{\bm{\kappa}})\big),\partial_{\kappa_2}\big(\epsilon^{-\frac{1}{2}}\nabla\times(\mu^{-\frac{1}{2}}\bm{H}_{\bm{\kappa}})\big)\Big)_{3D }  -\frac{\partial_{\kappa_1}w_{\bm{\kappa}}}{w_{\bm{\kappa}}^3}  \Big(\epsilon^{-\frac{1}{2}}\nabla\times(\mu^{-\frac{1}{2}}\bm{H}_{\bm{\kappa}}),\partial_{\kappa_2}\big(\epsilon^{-\frac{1}{2}}\nabla\times(\mu^{-\frac{1}{2}}\bm{H}_{\bm{\kappa}})\big)\Big)_{3D } \\
&\quad -\frac{\partial_{\kappa_2}w_{\bm{\kappa}}}{w_{\bm{\kappa}}^3}  \Big(\partial_{\kappa_1}\big(\epsilon^{-\frac{1}{2}}\nabla\times(\mu^{-\frac{1}{2}}\bm{H}_{\bm{\kappa}})\big),\epsilon^{-\frac{1}{2}}\nabla\times(\mu^{-\frac{1}{2}}\bm{H}_{\bm{\kappa}})\Big)_{3D }  +\frac{\partial_{\kappa_1}w_{\bm{\kappa}}\partial_{\kappa_2}w_{\bm{\kappa}}}{w_{\bm{\kappa}}^4}
\Big(\epsilon^{-\frac{1}{2}}\nabla\times(\mu^{-\frac{1}{2}}\bm{H}_{\bm{\kappa}}),\epsilon^{-\frac{1}{2}}\nabla\times(\mu^{-\frac{1}{2}}\bm{H}_{\bm{\kappa}})\Big)_{3D } \\
&=:I_1(\bm{\kappa})+I_2(\bm{\kappa})+I_3(\bm{\kappa})+I_4(\bm{\kappa}). 
\end{aligned}
\end{equation}
\normalsize
Using \eqref{eq_phys_11} and the fact that $\omega_{\kappa}$ are real-valued, we find
\begin{equation} \label{eq_phys_13}
\text{Im}I_4(\bm{\kappa})
=\text{Im}\Big(\frac{\partial_{\kappa_1}w_{\bm{\kappa}}\partial_{\kappa_2}w_{\bm{\kappa}}}{w_{\bm{\kappa}}^2} \|\bm{E}_{\bm{\kappa}}\|^2\Big)
=0.
\end{equation}
We claim
\begin{equation} \label{eq_phys_14}
I_1(\bm{\kappa})+I_2(\bm{\kappa})+I_3(\bm{\kappa})
=\Big(\partial_{\kappa_1}\bm{H}_{\bm{\kappa}},\partial_{\kappa_2}\bm{H}_{\bm{\kappa}}\Big)_{3D }=(\partial_{\kappa_1}H_{\bm{\kappa}},\partial_{\kappa_2}H_{\bm{\kappa}}).
\end{equation}
Indeed, using \eqref{eq_phys_11}, an integration by parts yields
\footnotesize
\begin{equation} \label{eq_phys_15}
\begin{aligned}
I_2(\bm{\kappa})
&=-\frac{\partial_{\kappa_1}w_{\bm{\kappa}}}{w_{\bm{\kappa}}^3}  \Big(\epsilon^{-\frac{1}{2}}\nabla\times(\mu^{-\frac{1}{2}}\bm{H}_{\bm{\kappa}}),\partial_{\kappa_2}\big(\epsilon^{-\frac{1}{2}}\nabla\times(\mu^{-\frac{1}{2}}\bm{H}_{\bm{\kappa}})\big)\Big)_{3D }
=-\frac{\partial_{\kappa_1}w_{\bm{\kappa}}}{w_{\bm{\kappa}}^3}  \Big(\mu^{-\frac{1}{2}}\nabla\times\epsilon^{-1}\nabla\times(\mu^{-\frac{1}{2}}\bm{H}_{\bm{\kappa}}),\partial_{\kappa_2}\bm{H}_{\bm{\kappa}} \Big)_{3D } \\
&=-\frac{\partial_{\kappa_1}w_{\bm{\kappa}}}{w_{\bm{\kappa}}}  \Big(\bm{H}_{\bm{\kappa}},\partial_{\kappa_2}\bm{H}_{\bm{\kappa}}\Big)_{3D }.
\end{aligned}
\end{equation}
\normalsize
Similarly
\begin{equation} \label{eq_phys_16}
\begin{aligned}
I_3(\bm{\kappa})
=-\frac{\partial_{\kappa_2}w_{\bm{\kappa}}}{w_{\bm{\kappa}}}  \Big(\partial_{\kappa_1}\bm{H}_{\bm{\kappa}},\bm{H}_{\bm{\kappa}}\Big)_{3D }. 
\end{aligned}
\end{equation}
For $I_1$,  there are two ways to perform integration by parts. The first is as follows
\begin{equation*}
\begin{aligned}
I_1(\bm{\kappa})
&=\frac{1}{w_{\bm{\kappa}}^2}  \Big(\partial_{\kappa_1}\big(\mu^{-\frac{1}{2}}\nabla\times\epsilon^{-1}\nabla\times(\mu^{-\frac{1}{2}}\bm{H}_{\bm{\kappa}})\big),\partial_{\kappa_2}\bm{H}_{\bm{\kappa}}\Big)_{3D }
=\frac{1}{w_{\bm{\kappa}}^2}  \Big(\partial_{\kappa_1}\big(w_{\bm{\kappa}}^2\bm{H}_{\bm{\kappa}}\big),\partial_{\kappa_2}\bm{H}_{\bm{\kappa}}\Big)_{3D } \\
&=\Big(\partial_{\kappa_1}\bm{H}_{\bm{\kappa}},\partial_{\kappa_2}\bm{H}_{\bm{\kappa}}\Big)_{3D }
+\frac{2\partial_{\kappa_1}w_{\bm{\kappa}}}{w_{\bm{\kappa}}}\Big(\bm{H}_{\bm{\kappa}},\partial_{\kappa_2}\bm{H}_{\bm{\kappa}}\Big)_{3D }. 
\end{aligned}
\end{equation*}
Alternatively, we can perform the following integration by parts
\begin{equation*}
\begin{aligned}
I_1(\bm{\kappa})
&=\frac{1}{w_{\bm{\kappa}}^2}  \Big(\partial_{\kappa_1}\bm{H}_{\bm{\kappa}},\partial_{\kappa_2}\big(\mu^{-\frac{1}{2}}\nabla\times\epsilon^{-1}\nabla\times(\mu^{-\frac{1}{2}}\bm{H}_{\bm{\kappa}})\big)\Big)_{3D }
=\frac{1}{w_{\bm{\kappa}}^2}  \Big(\partial_{\kappa_1}\bm{H}_{\bm{\kappa}},\partial_{\kappa_2}\big(w_{\bm{\kappa}}^2\bm{H}_{\bm{\kappa}}\big)\Big)_{3D } \\
&=\Big(\partial_{\kappa_1}\bm{H}_{\bm{\kappa}},\partial_{\kappa_2}\bm{H}_{\bm{\kappa}}\Big)_{3D }
+\frac{2\partial_{\kappa_2}w_{\bm{\kappa}}}{w_{\bm{\kappa}}}\Big(\partial_{\kappa_1}\bm{H}_{\bm{\kappa}},\bm{H}_{\bm{\kappa}}\Big)_{3D}.
\end{aligned}
\end{equation*}
By taking the average of these two expressions and using \eqref{eq_phys_15}-\eqref{eq_phys_16}, we obtain \eqref{eq_phys_14}.

Finally, with \eqref{eq_phys_9}-\eqref{eq_phys_14}, we conclude
\begin{equation*} 
\sigma_{xy}=\frac{1}{2}\int_{T^2}
2\text{Im}(\partial_{\kappa_1}H_{\bm{\kappa}},\partial_{\kappa_2}H_{\bm{\kappa}})d\bm{\kappa}
=\int_{T^2}\text{Im}(\partial_{\kappa_1}u_{\bm{\kappa}},\partial_{\kappa_2}u_{\bm{\kappa}})d\bm{\kappa}. 
\end{equation*}
This recovers the Chern number expression \eqref{eq_chern_num} of TE waves in the main text.



\subsection{Bulk-edge correspondence in topological photonics}
In this part, we illustrate the bulk-edge correspondence as a natural consequence of energy conservation. Loosely speaking, in the case the system is closed and finite in $xOy$ plane, the calculation for photonic Hall conductance in the last section still works when we focus on the energy flux within a unit cell that lies deep in the bulk. This is because the effect of boundary is neglectable when the observer is deep in the bulk and when the observed frequency $w_0$ lies in the bulk gap (nearly zero propagation from the boundary to the observer). This leads to the following equality (up to a constant) by \eqref{eq_phys_6}
\begin{equation*}
\mathcal{C}_{w_0}=\langle S_{xy}\rangle^{deep}.
\end{equation*}
Here $\mathcal{C}_{w_0}$ denotes the gap Chern number at the frequency $w_0$. $\langle S_{xy}\rangle^{deep}$ represents the (instant) measured energy flux in the $x$- direction at some observation point deep in bulk, when we apply a unit perturbation \eqref{eq_phys_3} in the $y$- direction.

Now we consider the effect of boundary. For example, we bound the structure by placing lateral PEC walls along a closed curve $\Gamma\subset \mathbf{R}^2$. Since PEC boundaries are impenetrable for energy conduction and we assume the medium is lossless, the perturbation-induced energy flux is forced to circulate within the finite structure. In photonic systems, this circulating energy flux in $xOy$ plane is described by the $z$-component of the light angular momentum $(\bm{r}\times \bm{S})_{z}$ \cite{mario17angular}. On the other hand, the carriers of the circulating energy flux are necessarily the modes located near the boundary $\Gamma$ because the observed frequency $w_0$ lies in the bulk gap. Heuristically, these discussions lead to the equality
\begin{equation*}
\langle S_{xy}\rangle^{deep} = \langle (\bm{r}\times \bm{S})_{z}\rangle^{edge}.
\end{equation*}
In conclusion, we have linked the gap Chern number with the light angular momentum carried by edge modes
\begin{equation*}
\mathcal{C}_{w_0}=\langle (\bm{r}\times \bm{S})_{z}\rangle^{edge}.
\end{equation*}
This equality gives the intuition of Theorem \ref{thm_bec}. As seen from the above discussion, this bulk-edge correspondence is a natural consequence of energy conservation. We point out that it's a direct analog of the mechanism of quantum Hall effect, i.e. the charge conservation, in electronic systems.


\footnotesize
\bibliographystyle{plain}
\bibliography{ref}

\end{document}